%% file: main.tex
\documentclass[11pt, a4paper,reqno]{amsart}

\usepackage{scalerel,stackengine}
\usepackage[T1]{fontenc}
\usepackage{ae,aecompl}
\usepackage{mathabx}
\usepackage{amssymb,amscd}
\usepackage{pstricks}
\usepackage{tikz}

\usepackage{txfonts}

\numberwithin{equation}{section}
\numberwithin{table}{section}
\numberwithin{figure}{section}

\newenvironment{beweis}[1][\textit{\proofname.} ]{#1}

\newtheorem{theorem}{Theorem}[section]
\newtheorem{proposition}[theorem]{Proposition}
\newtheorem{lemma}[theorem]{Lemma}
\newtheorem{corollary}[theorem]{Corollary}

\newtheorem{definition}[theorem]{Definition}
\newtheorem{remark}[theorem]{Remark}
\newtheorem{example}[theorem]{Example}

\newcommand{\ts}{\hspace{0.5pt}}
\newcommand{\exend}{\hfill$\Diamond$}

\DeclareMathOperator{\card}{card}
\DeclareMathOperator{\dens}{dens}
\DeclareMathOperator{\supp}{supp}

\def\Bap{\mathcal{B}\hspace*{-1pt}{\mathsf{ap}}}
\def\cA{\mathcal{A}}
\def\CC{\mathbb{C}}
\def\Cc{\Csp_{\mathsf{c}}}
\def\cF{\mathcal{F}}
\def\cL{\mathcal{L}}
\def\Cu{\Csp_{\mathsf{u}}}

\def\cM{\mathcal{M}}
\def\dd{\,\mathrm{d}}
\def\KK{\mathbb{K}}
\def\NN{\mathbb N}
\def\nts{\hspace{-0.5pt}}
\def\RR{\mathbb R}
\def\TT{\mathbb{T}}
\def\vL{\varLambda}
\def\ZZ{\mathbb Z}
\newcommand{\cP}{\mathcal{P}}
\newcommand{\ii}{\ts\mathrm{i}\ts}

\newcommand{\abs}[1]{\absl#1\absr}   
\def\absl{\bigl\lvert}   
\def\absr{\bigr\rvert}   
\def\Asp{{\boldsymbol A}}   
\def\Bsp{{\boldsymbol B}}   
\newcommand{\cMb}{\boldsymbol{\mathcal{M}}} 
\def\COsp{{\Csp^{}_{\negthinspace 0}}}   
\def\Csp{{\boldsymbol C}}   
\def\Dcsp{{\boldsymbol{\mathcal D}}}   
\def\eps{\varepsilon}   
\newcommand\FLi{{\mathcal F}{\negthinspace \Lisp}}   
\def\KLsp{{\boldsymbol{K\!L}}}   
\def\Ksp{{\boldsymbol K}}   
\def\Lisp{{\Lsp^1}}   
\def\Lilocsp{{\Lsp^1_{loc}}}   
\def\lisp{{\lsp^1}}   
\def\LKsp{{\boldsymbol{L \mkern-1mu K}}}   
\def\Lsp{{\boldsymbol L}}   
\def\lsp{{\boldsymbol\ell}}   
\def\Ltsp{{\Lsp^2}}   
\def\cMb{\cM^{}_{\mkern-1mu \mathsf b}} 
\def\Scsp{{\boldsymbol{\mathcal S}}}   
\def\shah{\boldsymbol{\sqcup} \mkern-4mu\boldsymbol{\sqcup}}  
\def\Ssp{{\boldsymbol S}}   
\def\SOsp{{\Ssp^{}_{\! 0}}}   
\newcommand\So{\SOsp\mkern-1mu} 
\newcommand\Soa{\So^{}_{\mkern1mu a}\mkern-1mu}
\def\SOPsp{{\Ssp_{\! 0}'}\mkern-1mu} 
\newcommand\SOp{\SOPsp}   
\def\Tsp{{\boldsymbol T}}   
\newcommand{\WCOli}{ \WTsp \COsp  \lisp}
\newcommand\WFLili{{\Wsp(\FLi,\lisp)}}   
\def\Wasp{{\boldsymbol W^{}_{\mkern-8mu a}}\mkern-1mu} 
\def\WOsp{{\boldsymbol W^{}_{\mkern-8mu 0}}\mkern-1mu} 
\def\Wisp{{\boldsymbol W^{}_{\mkern-8mu 1}}\mkern-1mu} 
\def\Wsp{{\boldsymbol W}\mkern-2mu}   
\def\Wspp{{\boldsymbol W}'\mkern-2mu} 
\def\WTsp#1#2{{\Wsp(#1,#2)}}   

\def\omit#1{}

\newcommand{\op}[1]{{\left\vert\kern-0.25ex\left\vert\kern-0.25ex\left\vert #1
    \right\vert\kern-0.25ex\right\vert\kern-0.25ex\right\vert}}

\stackMath
\newcommand\reallywidehat[1]{%
\savestack{\tmpbox}{\stretchto{%
  \scaleto{%
    \scalerel*[\widthof{\ensuremath{#1}}]{\kern-.6pt\bigwedge\kern-.6pt}%
    {\rule[-\textheight/2]{1ex}{\textheight}}
  }{\textheight}%
}{0.5ex}}%
\stackon[1pt]{#1}{\tmpbox}%
}

\newcommand\reallywidecheck[1]{%
\savestack{\tmpbox}{\stretchto{%
  \scaleto{%
    \scalerel*[\widthof{\ensuremath{#1}}]{\kern-.6pt\bigwedge\kern-.6pt}%
    {\rule[-\textheight/2]{1ex}{\textheight}}
  }{\textheight}%
}{0.5ex}}%
\stackon[1pt]{#1}{\scalebox{-1}{\tmpbox}}%
}

\begin{document}

\title[Twice Fourier Transformable Measures]{Twice Fourier transformable measures and diffraction theory}

\author[H.G.~Feichtinger]{Hans G. Feichtinger}
\address{H.G.~Feichtinger, Faculty of Mathematics, University of Vienna and Acoustic Research Institute ARI Vienna (OEAW), Austria}
\email{hans.feichtinger@univie.ac.at}

\author[C.~Richard]{Christoph Richard}
\address{C.~Richard, Department f\"{u}r Mathematik, Friedrich-Alexander-Universit\"{a}t Erlangen-N\"{u}rnberg, Germany}
\email{christoph.richard@fau.de}

\author[C.~Schumacher]{Christoph Schumacher}
\address{C.~Schumacher, Weststr.~1, 04720 Döbeln, Germany}
\email{christoph.schumacher@mathematik.tu-dortmund.de}

\author[N.~Strungaru]{Nicolae Strungaru}
\address{N.~Strungaru, Department of Mathematical Sciences, MacEwan University, Canada\\
and \\
Institute of Mathematics ``Simon Stoilow''\\
Bucharest, Romania}
\email{strungarun@macewan.ca}

\begin{abstract}
Mathematical diffraction theory has been developed since about 1995. Hof's initial approach relied on tempered distributions in euclidean space. Nowadays often the Fourier theory by Argabright and Gil de Lamadrid is used, which applies to appropriate measures on locally compact abelian groups. We review diffraction theory using Wiener amalgams as test function spaces. For translation bounded measures, this unifies and simplifies the former two approaches. We treat weighted versions of Meyer's model sets as examples.
\end{abstract}


\include{FRSS3}

\end{document}

%% file: FRSS3.tex
\def\Cusp{{\Csp^{}_{\mkern-4mu \mathsf u}}} 
\def\Ccsp{{\Csp^{}_{\mkern-4mu \mathsf c}}} 
\def\Cbsp{{\Csp^{}_{\mkern-4mu \mathsf b}}} 

\maketitle

\section{Introduction}

Mathematical diffraction theory has emerged from the discovery of quasicrystals, in order to describe their unusual diffraction properties. From a mathematical point of view, it is concerned with Fourier analysis of translation-bounded measures, which serve as an idealisation of a piece of matter.
Hof's influential approach \cite{FRSS-Hof1} used tempered distributions in Euclidean space. An earlier approach by de Bruijn relied on Gelfand--Shilov test functions, a certain subclass of Schwartz functions \cite{FRSS-B86,FRSS-L00}.

Yet another framework is prominently being used, namely the Fourier theory of measures by Argabright and Gil de Lamadrid, called AG theory in the sequel. It extends the Fourier theory of positive-definite measures based on Bochner's theorem \cite{FRSS-BF75} and has been developed in the 1970ies \cite{FRSS-ARMA2, FRSS-ARMA1}. AG theory proved useful for diffraction analysis for three reasons. Firstly, any translation-bounded measure indeed admits an autocorrelation measure and a diffraction measure as its Fourier transform~\cite{FRSS-BaakeLenz2004}, while the corresponding problem for tempered distributions is open \cite{FRSS-ST16}. Secondly, AG theory had already incorporated almost periodicity \cite{FRSS-AG79, FRSS-ARMA}. As almost periodicity had been recognised to be central for diffraction analysis \cite{FRSS-L00}, AG theory could readily be used to study the nature of the diffraction based on the autocorrelation measure \cite{FRSS-BaakeMoody2004, FRSS-MoSt, FRSS-S19}. Moreover, it recently appeared that important classes of pure-point diffractive, translation-bounded measures are characterised by various types of almost periodicity \cite{FRSS-LSS2020a,FRSS-LLRSS, FRSS-LRS2024}. Thirdly, AG theory applies to measures on locally compact Abelian groups beyond Euclidean space, in a simpler way than Schwartz--Bruhat tempered distributions. Such generality was indeed needed, as models for quasicrystals had emerged that were constructed from non-Euclidean internal spaces \cite{FRSS-BaakeMoody2004}. These can be analysed using the Poisson summation formula of the underlying non-Euclidean lattice~\cite{FRSS-CRS}.

However, using this framework seemed subtle, as the Fourier transform is not a bijection on the space of AG-transformable measures. This holds because the Fourier transform of a measure is required to be a measure in AG theory, which constrains the transform to be translation bounded~\cite[Thm.~2.5]{FRSS-ARMA1}. Whereas the question of twice AG transformability was posed as an open problem \cite[p.~50]{FRSS-ARMA1}, a simple answer was given already in 1979: A measure is twice AG transformable if and only if it is AG transformable and translation bounded  \cite{FRSS-Fei79a, FRSS-Fei80}. This explains why AG-transformable, translation-bounded measures constitute a natural space for diffraction analysis.

As the above characterisation remained neglected in diffraction analysis until recently \cite{FRSS-RiSchu20}, it is the aim of this contribution to provide sufficient background. The above characterisation has been obtained within the framework of Wiener amalgam spaces \cite{FRSS-F83, FRSS-Fei81c}. Two important ingredients, which we will review below, are the identification of translation-bounded measures as the bounded linear functionals on the Wiener algebra \cite{FRSS-LRW74}, and Feichtinger's mild distributions \cite{FRSS-Fei81, FRSS-Jak, FRSS-Fei23}. Mild distributions can be defined on locally compact Abelian groups, where they constitute a substantial and well tractable  subclass of Schwartz--Bruhat tempered distributions.

As any AG-transformable measure is a mild distribution, AG theory is subsumed by mild distribution theory.
We will thus review AG theory using mild distributions, thereby providing detailed proofs of statements in \cite{FRSS-Fei79a, FRSS-Fei80}. In particular, we will give a characterisation of AG transformability that is considerably simpler than the original definition, and we will prove the above characterisation of twice AG transformability.

We will also rephrase diffraction theory of translation-bounded measures using mild distributions. While this streamlines arguments in comparison to the traditional approaches \cite{FRSS-Hof1, FRSS-BaakeLenz2004}, it also indicates extensions beyond translation-bounded measures. We will finally review weighted versions of Meyer's model sets \cite{FRSS-Mey72}, using mild distributions.  This will rely on a Poisson summation formula for test functions including Feichtinger's algebra, which extends \cite{FRSS-M19, FRSS-M22} beyond Euclidean space. We thus obtain diffraction results that improve previous work, as summarised in \cite{FRSS-CRS}.

Altogether, we demonstrate that mild distributions provide a flexible instrument for diffraction analysis, which subsumes and simplifies the traditional approaches on translation-bounded measures.

This chapter is organized as follows. We start with a detailed discussion of bounded uniform partitions of unity in Sections~\ref{FRSS-sec:bupu} and \ref{FRSS-cbupu}. We then carefully review the Wiener algebra in Section~\ref{FRSS-sec:WA} and Feichtinger's algebra in Section~\ref{FRSS-sect:appB}, using bounded uniform partitions of unity. The space dual to Feichtinger's algebra, the space of mild distributions, is discussed in Section~\ref{FRSS-sec:fmd}. Section~\ref{FRSS-sec:m} reviews translation-bounded measures and mild measures. The central Section \ref{FRSS-sec:FTm} reviews AG transformability using mild distributions.
Section~\ref{FRSS-sec:dtr} reviews diffraction theory from the mild distribution viewpoint, taking into account recent developments \cite{FRSS-LSS2020a} on pure-point diffraction. Section~\ref{FRSS-sec:ppdpsr} applies mild distribution theory to weighted model sets. The main result of that section has been announced elsewhere \cite{FRSS-R23}.

\section{Bounded uniform partitions of unity}\label{FRSS-sec:bupu}

Throughout this article, $G$ denotes a locally compact Abelian (LCA) group, and $\widehat G$ denotes its dual LCA group. We choose Haar measures $m_G$ on $G$ and $m_{\widehat G}$ on $\widehat G$ such that Plancherel's formula holds. 
The algebra of continuous functions from $G$ to $\CC$ is denoted by $\Csp(G)$, the subalgebra of bounded continuous functions is denoted by $\Cbsp(G)$. The latter space is a Banach algebra with respect to the supremum norm $\|\cdot\|_\infty$. The closed subalgebra of uniformly continuous and bounded functions is denoted by $\Cusp(G)$. Recall that a function $f: G \to \CC$ is said to be \emph{vanishing at infinity} if for every $\varepsilon>0$ there exists compact $K \subseteq G$ such that $|f(x)|\le\varepsilon$ for all $x \notin K$. The set of continuous functions vanishing at infinity is denoted by $\COsp(G)$. It is a closed subalgebra of $(\Cusp(G), \|\cdot\|_\infty)$. We denote by $\Ccsp(G)\subseteq \COsp(G)$ the $\|\cdot\|_\infty$-dense subalgebra of continuous functions from $G$ to $\CC$ having compact support.

\subsection{Motivation}\label{FRSS-bupumot}

For an LCA group $G$, its Wiener algebra $\Wsp(G)$ and its Feichtinger algebra $\So(G)$ are subalgebras of $\COsp(G)$ that consist of functions with prescribed local and global behaviour. This can conveniently be formulated using bounded uniform partitions of unity\footnote{Our discussion may be complemented by the recent survey \cite{FRSS-Fei24}.}. Recall that a partition of unity\index{partition~of~unity} $(\phi_i)_{i\in I}$ on $G$ is a collection of functions $\phi_i\in \Ccsp(G)$ such that $\sum_{i\in I} \phi_i(x)=1$, where the sum is finite at each $x \in G$.

In order to motivate the notion of bounded uniform partition of unity, consider a tiling of $G=\RR^d$ by translates of non-empty compact subsets $P_1, P_2, \ldots , P_k$ of $\RR^d$, which are called the \emph{prototiles}. Thus, for each $j$, there exists a discrete set $\vL_j\subset  \RR^d$ such that
\[
\RR^d \,=\, \bigcup_{j=1}^k (\varLambda_{j}+P_j) \,=\, \bigcup_{j=1}^k\bigcup_{t \in \varLambda_j} (t+P_j) \ts ,
\]
where all tiles $t+P_j$ in the above union have mutually disjoint interior; compare \cite[Def.~5.2]{FRSS-TAO1}.
If all prototiles have a boundary of zero Lebesgue measure, then
\[
\sum_{j=1}^k \sum_{t \in \vL_j} 1_{t+P_j}(x) \,=\, 1
\]
for Lebesgue-almost all $x\in\RR^d$, where $1_A$ is the characteristic function of the set $A$. We would like to have an equality which holds everywhere, where the characteristic functions of tiles are replaced by continuous functions. This can be achieved by a partition of unity. But exactly as in the tiling case, we would like to additionally use uniformly bounded functions having supports of uniform size.

\smallskip

The following simple example gives a bounded uniform partition of unity on $G=\RR$ based on a tiling of the real line by an interval prototile.

\begin{example}[Partition of unity on $\RR$]\label{FRSS-ex:standardBUPU}
Consider the triangular functions $\phi(x)= \max\{ 0 , 1-|x| \}$ and $\phi_n(x)=\phi(x-n)$ for $n\in\ZZ$.
Then, $(\phi_n)_{n\in \ZZ}$ is a partition of unity of $\RR$, where each function is bounded by $1$ and has support of length $2$.

In order to see that $\psi=\sum_{n\in \ZZ}\phi_n$ satisfies $\psi=1$, consider any $m \in \ZZ$. As $\phi$ is supported on $[-1,1]$, and linear on $[-1,0]$ and $[0,1]$,  the restriction of $\psi$ to the interval $[m,m+1]$ is the sum of two linear functions and hence linear. Moreover, we have $\psi(m)=\psi(m+1)=1$. We conclude that $\psi(x)=1$ for all $x$.
For a more explicit argument, consider arbitrary $x \in \RR$ and let $m_x =\lfloor x \rfloor$, the unique integer such that $m_x \leq x < m_x+1$. We then clearly have $\phi_n(x)=0$ for all $n \notin \{m_x,m_x+1\}$. Thus
\begin{displaymath}
\psi(x) \,=\, \phi_{m_x}(x)+\phi_{m_x+1}(x)\,=\,(1-(x-m_x))+(1+(x-(m_x+1)))\,=\,1 \ts ,
\end{displaymath}
see the figure below for illustration.
$$
\begin{tikzpicture}[scale=.6]
\draw[->] (0,0)--(10,0);
\draw (2,0)--(4,2)--(6,0);
\draw (4,0)--(6,2)--(8,0);
\draw (2,0.2)--(2,-.2);
\draw (4,0.2)--(4,-.2);
\draw (6,0.2)--(6,-.2);
\draw (8,0.2)--(8,-.2);
\draw (4.7,0.2)--(4.7,-.2);
\node[anchor=north] at (2,-.3){$m_x{-}1$};
\node[anchor=north] at (4,-.4){$m_x$};
\node[anchor=north] at (6,-.3){$m_x{+}1$};
\node[anchor=north] at (8,-.3){$m_x{+}2$};
\node[anchor=north] at (4.7,-.4){$x$};
\end{tikzpicture}
$$
Note that $\phi_{n}$ can be written as convolution of two rectangular functions. This implies that its Fourier transform is integrable with a bound that is uniform in $n\in \ZZ$; compare \cite[Ex.~1.1.7]{FRSS-Rei2}. The latter viewpoint admits an extension to general LCA groups \cite[Prop.~5.1.3]{FRSS-Rei2}.
\end{example}

\begin{example}[Partition of unity on $\RR$ with arbitrary support]\label{FRSS-ex:scaledBUPU}
Consider the partition of unity $\phi_n$ from Example~\ref{FRSS-ex:standardBUPU}. Fix some $\alpha >0$ and define
\[
\phi_{n,\alpha}(x)\,=\, \phi_n(x/\alpha) \ts 
\]
for each $n$.
Then, $(\phi_{n,\alpha})_{n\in \ZZ}$ is a partition of unity of $\RR$, where each function is bounded by $1$ and has support of length $2\alpha$. \exend
\end{example}

The following examples give partitions of unity on $G=\RR$ where boundedness is violated.

\begin{example}[Arbitrarily large support]
Consider the triangular function $\psi_0(x)= \max\{ 0 , 1-|x| \}$
and define for $k \in \NN$ the trapezoidal functions
$$\psi_k(x) \,=\,
\left\{
\begin{array}{cl}
  x-(2^{k-1}-1) \ts ,  & \mbox{ if }  x \in [2^{k-1}-1, 2^{k-1}]\ts ,\\
    1 \ts , & \mbox{ if }  x \in [2^{k-1}, 2^{k}-1]\ts ,\\
    2^{k} - x\ts ,  & \mbox{ if }  x \in [2^{k}-1, 2^{k}]\ts ,\\
     0 \ts ,& \mbox{ if }  x \notin [2^{k-1}-1, 2^{k}]\ts ,
\end{array}
\right. \ts ,
$$
see the figure below for illustration.
$$\begin{tikzpicture}[scale=1]
\draw[->] (2,0)--(9,0);
\draw (3,0)--(4,1)--(7,1)--(8,0);
\draw (3,0.1)--(3,-.1);
\draw (4,0.1)--(4,-.1);
\draw (7,0.1)--(7,-.1);
\draw (8,0.1)--(8,-.1);
\node[anchor=north] at (3,-.1){$2^{k-1}{-}1$};
\node[anchor=south] at (4,.1){$2^{k-1}$};
\node[anchor=south] at (7,.1){$2^{k}{-}1$};
\node[anchor=north] at (8,-.1){$2^{k}$};
\end{tikzpicture}
$$
Define $\psi_{-k}(x)=\psi_k(-x)$ for $k\in \NN$. Observing that, with $\phi_n$ as in Example~\ref{FRSS-ex:standardBUPU}, we have
\[
\psi_{k} \,=\, \sum_{n=2^{k-1}}^{2^k-1} \phi_n \ts ,
\]
it follows immediately that $(\psi_n)_{n \in \ZZ}$
is a partition of unity of $\RR$, whose functions have  supports that are not uniformly bounded in size. \exend
\end{example}

\begin{example}[Arbitrarily large range] For $n \in \ZZ$, consider the triangular functions $\phi_n(x)=\max \{ 0, 1- |x-n| \}$ from Example~\ref{FRSS-ex:standardBUPU} and define
\begin{displaymath}
\psi_n \,=\, -(n-1)\phi_{n-1}+\phi_n +(n+1)\phi_{n+1} \ts .
\end{displaymath}
Then, $(\psi_n)_{n\in \ZZ}$ is a partition of unity with functions having support of length 4 and having arbitrarily large range. 

In order to show $\sum_{n\in \NN}\psi_n=1$, consider arbitrary $x \in \RR$ and let $m_x =\lfloor x \rfloor$. Recall $\phi_n(x)=0$ for all $n \notin \{m_x,m_x+1\}$, which implies $\psi_n(x)=0$ for all $n \notin \{ m_x-1,m_x,m_x+1,m_x+2 \}$.
Using $\phi_{m_x+n}(x)=0$ for $n\in \{-2,-1,2,3\}$, we get
\begin{displaymath}
\begin{split}
\sum_{n \in \ZZ} \psi_n(x) &\,=\, \psi_{m_x-1}(x)+\psi_{m_x}(x) +\psi_{m_x+1}(x) +\psi_{m_x+2}(x)\\
&\,=\, m_x\phi_{m_x}(x)+(\phi_{m_x}(x) +(m_x+1)\phi_{m_x+1}(x))\\
& \hspace{3ex} +(-m_x\phi_{m_x}(x)+\phi_{m_x+1}(x)) -(m_x+1)\phi_{m_x+1}(x)\\
&\,=\,\phi_{m_x}(x) +\phi_{m_x+1}(x)\,=\sum_{n\in \ZZ} \phi_n(x)\,=\,1 \ts .
\hspace{3.6cm} \text{\exend}
\end{split}
\end{displaymath}
\end{example}

\medskip

It would not be possible to describe the Wiener algebra using arbitrary partitions of unity. For example, properties (a) and (f) in Theorem~\ref{FRSS-thm:W} would then be violated. For a bounded uniform partition of unity $(\phi_i)_{i\in I}$, we need to describe where the functions $\phi_i$ are located. We use indexed families of points in $G$ for that purpose. Let us stress here that we do not use point sets, as we want to allow for finitely many repetitions of points, hence multi-sets. Indeed, finitely many repetitions occur naturally when refining a given bounded uniform partition of unity; see Remark~\ref{FRSS-rem:refbupu}.
For an indexed point family $X=(x_i)_{i\in I}$, we call the set
$\supp(X)\, = \, \{x_i: i\in I\}$
the \emph{supporting set of $X$}\index{support!supporting~set}.
We then have $\supp(X)=\supp(\delta^{}_{X})$ for the support of the weighted Dirac comb \index{Dirac~comb!weighted}
\[
\delta^{}_{X} \, = \,\sum_{i\in I}\delta^{}_{x_i}
\]
that is associated with the point family $X$. If we use the multiplicity function $r^{}_{X}: G\xrightarrow{\quad} \NN_0\cup \{\infty\}$, which is defined by $r^{}_{X}(x)\, = \, \card\{i\in I: x_i=x\}$,
we can write
\[
\delta^{}_{X}\, =\, \sum_{x\in \supp(\delta^{}_{X})} r^{}_{X}(x)\delta_x \ts .
\]

\medskip

In order to describe properties of indexed point families, we extend some standard notions for sets as described in \cite[Ch.~2.1]{FRSS-TAO1}.
A point family $X$ is called \emph{uniformly discrete}\index{uniformly~discrete} if the set $\supp(X)$ is uniformly discrete, i.e., if there exists an open $0$-neighbourhood $U\subseteq G$ such that
\[
(x+U)\cap(y+U) \, = \, \varnothing
\]
holds for all distinct $x,y\in \supp(X)$; compare \cite[p.~278]{FRSS-St17}. In this case, we will also refer to $X$ as being a \emph{$U$-uniformly discrete} point family.
A point family $X$ is called \emph{relatively dense}\index{relatively~dense} if the set $\supp(X)$ is relatively dense, meaning there exists a compact set $K\subseteq G$ such that
\[
\supp(X)+K\,=\, G \ts .
\]
We will sometimes call the point family $X$ a $K$-relatively dense point family in that case.
A point family $X$ is called \emph{Delone}\index{Delone} if $\supp(X)$ is a Delone set, so $\supp(X)$ is both uniformly discrete and relatively dense.

A point family $X$ is called \emph{weakly uniformly discrete}\index{weakly~uniformly~discrete} if its multiplicity function $r^{}_{X}$ is bounded and if the set $\supp(X)$ is weakly uniformly discrete, meaning that $\supp(X)$ is a finite union of uniformly discrete sets; compare \cite[p.~278]{FRSS-St17}. It can be shown \cite[Thm.~14]{FRSS-fe92} that weak uniform discreteness of $X$ is equivalent to translation boundedness of the associated weighted Dirac comb $\delta^{}_{X}$. Recall here that $\delta^{}_{X}$ is translation bounded if for every compact $K\subseteq G$, we have
\begin{equation}\label{FRSS-eq:dXK}
\|\delta^{}_{X}\|^{}_{K} \,=\, \sup_{x\in G}\delta^{}_{X}(x+K)\,<\,\infty \ts ,
\end{equation}
where $\delta^{}_{X}(A) =  \card\{i\in I: x_i\in A\}$
for $A\subseteq G$. In fact, it suffices to check \eqref{FRSS-eq:dXK} for a single relatively compact set with non-empty interior, see Lemma~\ref{FRSS-lem:precomp-equiv-tb-norm}.

\begin{remark}[Dictionary] The following names and symbols are used in \cite{FRSS-fe92, FRSS-fegr92-3, FRSS-Fei23} for an indexed point family:
\begin{itemize}
\setlength\itemsep{2pt}
  \item{} \emph{uniformly separated} for uniformly discrete,
  \item{} \emph{relatively separated} for weakly uniformly discrete,
  \item{} \emph{well-spread} for weakly uniformly discrete and relatively dense,
  \item{} \emph{shah-distribution} $\shah^{}_{X}$ for the weighted Dirac comb $\delta^{}_{X}$. \exend
\end{itemize}
\end{remark}

Wiener algebra functions are absolutely summable on uniformly discrete point sets. In fact, they may be evaluated on general translation-bounded Radon measures, see Theorem~\ref{FRSS-thm:tbW} and Section~\ref{FRSS-sec:FTm} below.

\subsection{Definition and elementary properties}

We now give the definition of bounded uniform partition of unity.

\begin{definition}[BUPU]\label{FRSS-def:BUPUCc}
An indexed family $\varPhi=(\phi_i)_{i\in I}$ of functions in $\Ccsp(G)$ is called a \emph{bounded uniform partition of unity} (or simply \emph{BUPU})\index{BUPU!in~$\Ccsp(G)$} of norm $M$, size $U$ and overlap constant $B$, if there exist points $X=(x_i)_{i\in I}$ in $G$ and a relatively compact $0$-neighbourhood $U\subseteq G$ such that
\begin{itemize}
\setlength\itemsep{2pt}
\item[(a)] $\supp(\phi_i)\subseteq x_i +U$ for all $i \in I$,
\item[(b)] $\card\{j\in I: (x_i +U)\cap (x_j+U)\ne \varnothing\}\le B<\infty$ for all $i\in I$,
\item[(c)] $\sum_{i\in I} \phi_i(x)=1$ for all $x\in G$,
\item[(d)] $\|\phi_i\|^{}_{\infty}\le M<\infty$ for all $i \in I$.
\end{itemize}
\end{definition}

\begin{remark}[Non-negative BUPUs]
Whereas BUPUs often satisfy $\phi_i\ge0$, non-negativity is not required in the definition of a BUPU. In fact, all BUPUs we construct in the next sections are non-negative or can be made non-negative.
 If $\varPhi = ( \phi_i )_{i \in I}$ is a non-negative BUPU, then $\| \phi_i \|_\infty \leq 1$. We may thus choose norm $M=1$ in this case. \exend
\end{remark}

\begin{remark}[Overlap constant]\label{FRSS-rem:cd}
Condition (b) in the definition of BUPU implies that the sum in (c) has, for any fixed $x\in G$, at most $B$ non-zero terms.
We show in Corollary~\ref{FRSS-cor:chard} that condition (b) is equivalent to translation boundedness of the weighted Dirac comb $\delta^{}_{X}$. In fact, (b) is replaced in \cite{FRSS-F83} by translation boundedness of $\delta^{}_{X}$. Note that $\|\delta^{}_{X}\|^{}_{U-U}$ provides a (not necessarily optimal) overlap constant by the following lemma. \exend
\end{remark}

\begin{lemma}\label{FRSS-lem:BUPU-tb} Let $X=(x_i)_{i\in I}$ be a point family in $G$, and let\/ $U, V$ be any two relatively compact $0$-neighbourhoods in\/ $G$. We then have
\begin{displaymath}
\begin{split}
\sup_{y\in G} \card\{i\in I: (x_i +U)\cap (y+V)\ne \varnothing\} \,&\leq\, \| \delta^{}_{X} \|^{}_{V-U} \ts , \\
\sup_{j \in I} \card\{i\in I: (x_i +U)\cap (x_j+V)\ne \varnothing\} \,&\geq \, \|\delta^{}_{X} \|^{}_{-(U\cap V)} \ts .
\end{split}
\end{displaymath}
\end{lemma}
\begin{beweis}
Note first that $(x_i +U)\cap (y+V) \neq  \varnothing$ if and only if there exist $u\in U$ and $v \in V$ such that $x_i+u =y +v$ or, equivalently, if and only if $x_i \in y +(V-U)$. We thus have
\[
 \card\{i\in I: (x_i +U)\cap (y+V)\ne \varnothing\} \,=\, \delta^{}_{X}(y+(V-U)) \,\leq\,
 \| \delta^{}_{X} \|^{}_{V-U} \ts ,
\]
and the upper inequality follows. For the lower inequality, note that we have
\[
\delta^{}_{X}(x-(U\cap V))\,=\, \card\{ i \in I: x_i \in x-(U\cap V) \} \ts .
\]
Define $I_x =  \{ i \in I : x_i \in x-(U\cap V) \}$ and fix an arbitrary $x_j\in I_x$.
Then, $i\in I_x$ implies
\[
x \in (x_i+(U\cap V)) \cap (x_j+(U\cap V)) \subseteq (x_i+U)\cap (x_j+V) \,.
\]
Therefore, we get
\begin{displaymath}
\begin{split}
\delta^{}_{X}(x-(U\cap V))&\,=\, \card(I_x) \,\leq\,  \card\{i\in I: (x_i +U)\cap (x_j+V)\ne \varnothing\} \\[1mm]
& \,\leq \, \sup_{j \in I} \card\{i\in I: (x_i +U)\cap (x_j+V)\ne \varnothing\} \ts .
\end{split}
\end{displaymath}
As the upper bound is uniform in $x\in G$, the lower inequality follows.
\qed \end{beweis}

In fact, the above lemma leads to the following characterisation of translation boundedness of a point family.

\begin{corollary}\label{FRSS-cor:chard} Let\/ $X=(x_i)_{i\in I}$ be a point family in\/ $G$. Then, the following properties are equivalent.
\begin{itemize}
\setlength\itemsep{2pt}
  \item[(i)] $\delta^{}_{X}$ is translation bounded.
  \item[(ii)]  There exists some relatively compact $0$-neighbourhood $U\subseteq G$ and some finite $B$ such that\/
  \[
  \card\{i\in I: (x_i +U)\cap (x_j+U)\ne \varnothing\} \, \le \, B
  \]
  holds for all $j\in I$.
  \item[(iii)] For any relatively compact $0$-neighbourhood $V\subseteq G$, there exists some finite\/ $B^{}_{V}$ such that\/
  \[
  \card\{i\in I: (x_i +V)\cap (x_j+V)\ne \varnothing\} \, \le \, B^{}_{V}
  \]
  holds for all $j\in I$.
\end{itemize}
Moreover, in this case we can choose $B^{}_{V}= \| \delta^{}_{X} \|^{}_{V-V}$.
\qed
\end{corollary}

\begin{remark}[BUPUs on Delone sets]
For a general BUPU, its point family $X$ is well-spread, i.e., it is weakly uniformly discrete and relatively dense. Here we argue that any given BUPU can be modified to having a Delone set as point family.
Let $\varPhi = ( \phi_i )_{i\in I}$ be a BUPU of norm $M$, size $U$ and overlap constant $B$,  with point family $X=(x_i)_{i\in I}$. Consider the point set $Y = \supp(X)$ and define for every $y\in Y$ the function
\[
\psi_y \, = \, \sum_{i \in I(y)} \phi_i \ts ,
\]
where the index set $I(y)$ is given by  $I(y)=\{i\in I: x_i=y\}$.
Then, $(\psi_y)_{y \in Y}$ is a BUPU of norm $B M$, size $U$ and overlap constant $B$, with point set $Y$. We thus can assume without loss of generality that the point family of a given
BUPU is a point set.

In fact, in this situation, we can even assume that $Y$ is a Delone set. To see this, fix some relatively compact $0$-neighbourhood $V$ in $G$ such that  $V-V\subseteq U$. Pick some $Z \subseteq Y$ which is a maximal $V$-uniformly discrete subset of $Y$. Such $Z$ exists by Zorn's Lemma and satisfies $Y \subseteq Z+U$, see Lemma~\ref{FRSS-lemma:existe-Del} below. In particular, $Z$ inherits relative denseness from $Y$. Now assign to every $y\in Y$ some $z\in Z$ such that $y\in z+U$ as follows. Define $f:Y\to Z$ by $f(y)=y$ if $y\in Y\cap Z$. If $y\in Y\setminus Z$, choose any $z\in Z$ such that $y\in z+U$ and define $f(y)=z$.
Now, for each $z \in Z$ define
\[
\varphi_z \, = \, \sum_{y \in f^{-1}(\{z\})} \psi_y \ts .
\]
Then $(\varphi_z)_{z \in Z}$ is a BUPU of norm $B^2M$ and size $U+U$. \exend
\end{remark}

We have argued above that one might restrict to BUPUs on Delone point sets. However, it is convenient to consider more general weakly uniformly discrete point families instead. Indeed, weak uniform discreteness is preserved under finite unions, which allows to define convex combinations of BUPUs. Also, it is natural to assume that $X$ allows for repetitions of points in the definition of BUPU. Indeed, when refining a BUPU by another BUPU as explained below, finitely many repetitions of centre points from $X$ will arise. This is easily accounted for by working with point families instead of point sets.

\begin{remark}[Refining a BUPU]\label{FRSS-rem:refbupu}
Let $\varPhi,\varPsi$ be two BUPUs of norm $M_\varPhi$ and $M_\varPsi$, size $U$ and $V$ with point families $X=(x_i)_{i\in I}$ and $Y=(y_j)_{j\in J}$.
Consider the index set
\[
K\, = \,\{(i,j)\in I\times J: (x_i+U)\cap (y_j+V)\ne \varnothing\} \ts .
\]
The refinement\index{refinement~of~BUPU}
$\Theta=\varPhi|\varPsi$ of $\varPhi$ by $\varPsi$ is the following BUPU of norm $M_\varPhi M_\varPsi$, size $U$ with index set $K$:

Set $Z=(z_k)_{k\in K}$ with $z_{(i,j)}=x_i$ for $(i,j)\in K$. Note that the number of repetitions of any point in $Z$ is uniformly bounded due to the finite overlap property. Indeed, for fixed $i\in I$ we have
\begin{displaymath}
\begin{split}
\card\{j\in J:(i,j)\in K\} &\,=\, \card\{j\in J: (x_i+U)\cap (y_j+V)\ne\varnothing\} \\[1mm]
&\,=\,\delta^{}_{Y}(x_i+U - V)\,\le\, \|\delta^{}_{Y}\|^{}_{U-V} < \infty \ts ,
\end{split}
\end{displaymath}
compare Lemma~\ref{FRSS-lem:BUPU-tb}. Let now $\theta_{(i,j)}\,:=\, \phi^{}_{i}\cdot \psi^{}_{j}$ for $(i,j)\in K$. Then, $\Theta=(\theta_{k})_{k\in K}$ is a BUPU of size $U$: Note that property (c) in Definition~\ref{FRSS-def:BUPUCc} holds trivially, and that we have
\[
\|\theta_k\| \, \le \, \|\phi_i\|\cdot \|\psi_j\| \, \le \, M_\varPhi\cdot M_\varPsi \ts ,
\]
as the norm is multiplicative. Moreover, we have $\supp(\theta_{(i,j)})\subseteq x_i +U$ for all $(i,j)\in K$, and we may set $B^{}_{\Theta}= B^{}_{\varPhi}\cdot \|\delta^{}_{Y}\|^{}_{U-V}$. Indeed, for $z_\ell\in K$, we can estimate
\begin{displaymath}
\begin{split}
\card&\{k\in K : (z_\ell+U)\cap (z_k+U)\ne\varnothing\}\\[1mm]
&\,=\, \card\{(i,j)\in I\times J: (x_\ell+U)\cap (x_i+U)\ne\varnothing \wedge (i,j) \in K \}\\[1mm]
&\,\le\, \card\{i\in I: (x_\ell+U)\cap (x_i+U)\ne \varnothing\} \cdot \|\delta^{}_Y\|^{}_{U-V} \le B^{}_{\varPhi}\cdot \|\delta^{}_{Y}\|^{}_{U-V} \ts .
\end{split}
\end{displaymath}
For any continuous and bounded function $f:G\to\mathbb C$ consider now the expressions $\|f\|^{}_{\varPhi}:=\sum_{i\in I} \|f\phi_i\|$ and $\|f\|^{}_{\varPsi}:=\sum_{j\in J} \|f\psi_j\|$. Then multiplicativity of the norm implies that
\begin{equation}\label{FRSS-eq:refined}
\begin{split}
\max\{\|f\|^{}_{\varPhi}, &\|f\|^{}_{\varPsi}\}\,\le\,\|f\|^{}_{\varPhi|\varPsi}\,= \sum_{(i,j)\in K} \|\phi_i\psi_j f\| \\
&\,\le\, \|\delta^{}_{Y}\|^{}_{U-V}  \cdot M^{}_{\varPsi} \cdot \sum_{i\in I} \|\phi_if\|
\,=\, \|\delta^{}_{Y}\|^{}_{U-V}  \cdot M^{}_{\varPsi}\cdot \|f\|^{}_{\varPhi} \ts ,
\end{split}
\end{equation}
an estimate that will be useful a few times later on. \exend
\end{remark}

We now discuss BUPUs in a product of two LCA groups. As usual, for $\phi\in \COsp(G)$ and $\psi\in \COsp(H)$, we denote by $\phi\otimes\psi\in \COsp(G{\times}H)$ their \emph{tensor product}\index{tensor~product} defined by
\[
(\phi\otimes\psi)(x,y) \, = \, \phi(x)\cdot \psi(y) \,.
\]
This product will be discussed in more detail in Section~\ref{FRSS-sec:Pp} below.

\begin{lemma}[Product BUPU]\label{FRSS-lem:prodbupu} Consider two LCA groups, $G$ and $H$. Let $\varPhi=(\phi_i)_{i\in I}$ be a BUPU of size\/ $U$ and norm\/ $M$ with point family $X=(x_i)_{i\in I}$ in $G$. Let $\varPsi=(\psi_j)_{j\in J}$ be a BUPU of size $V$ and norm $N$ with point family $Y=(y_j)_{j\in J}$ in $H$. Define
\[
\varPhi\otimes \varPsi=\{\phi_i \otimes \psi_j: (i,j)\in I{\times}J\}\ts .
\]
Then, $\varPhi\otimes \varPsi$ is a BUPU of size $U{\times}V$ and norm $M\cdot N$ with corresponding point family in $G{\times}H $ given by
\[
X{\times}Y \, = \, ((x_i, y_j))_{(i,j)\in I{\times}J} \ts .
\]
\end{lemma}

\begin{beweis}
As we have
\[
\|\delta^{}_{X\times Y}\|^{}_{K\times L} \, = \, \|\delta^{}_{X}\|^{}_{K}\cdot \|\delta^{}_{Y}\|^{}_{L} \ts ,
\]
the weighted Dirac comb $\delta^{}_{X\times Y}$ is translation bounded.
Condition (b) in the BUPU Definition~\ref{FRSS-def:BUPUCc} is thus satisfied due to Corollary~\ref{FRSS-cor:chard}. Moreover, the set $U\times V$ is a relatively compact $0$-neighbourhood in $G\times H$.
We clearly have $\phi_i\otimes \psi_j\in \Ccsp(G{\times}H)$ and
\[
\supp(\phi_i\otimes \psi_j) \, \subseteq \,  (x_i,y_j) + (U\times V) 
\]
for all $(i,j)\in I\times J$. We also have
\begin{displaymath}
\sum_{(i,j)\in I{\times}J} (\phi_i\otimes \psi_j) ((x,y))\,= \sum_{(i,j)\in I{\times}J} \phi_i(x) \cdot \psi_j(y)
\,= \sum_{i\in  I} \phi_i(x) \cdot \sum_{j\in  J} \psi_j(y) \,=\, 1 \ts .
\end{displaymath}
As the norm estimate
\[
\|\phi_i\otimes \psi_j\|^{}_\infty \, \le \,  \|\phi_i\|^{}_\infty\cdot \|\psi_j\|^{}_\infty\le M\cdot N
\]
holds for all $(i,j)\in I\times J$, the claim follows.
\qed \end{beweis}

Recall that Example~\ref{FRSS-ex:standardBUPU} gives BUPUs of arbitrary small size in $G=\RR$. Let us next construct BUPUs in general LCA groups.

\section{Constructing BUPUs}\label{FRSS-cbupu}

In this section, we review constructions of BUPUs of arbitrary small size from \cite{FRSS-Fei81b, FRSS-HGF87}. 
We first describe a smoothing procedure that allows to construct a BUPU from a family of functions that are not necessarily continuous but satisfy the other properties of a BUPU. In particular, this applies to  the motivating tiling example of Section~\ref{FRSS-bupumot} for the case of a single prototile. It will also allow us to construct BUPUs in the Fourier algebra of $G$, compare Definition~\ref{FRSS-def:BUPUA} below.
We then review two constructions of BUPUs with prescribed locations of their function centers. One uses Delone sets, the other uses lattices of arbitrary density, whose existence is granted by the structure theorem of LCA groups. The latter two constructions avoid measure-theoretic arguments, see also \cite{FRSS-Fei22} for a recent systematic approach in that direction.

\subsection{BUPUs via smoothing}

Given a BUPU, we explain how to obtain a BUPU of functions with stronger analytic properties by a certain smoothing procedure.

In the following lemma, the \emph{support} of a measurable function $f$ on $G$ is defined as the support of the induced measure on $G$. This means that we have $\supp(f)=\supp(f\cdot m_G)$\index{support!of~a~measurable~function}. For continuous $f$, this coincides with its topological support $\overline{\{x\in G: f(x)\ne 0\}}$\index{support!topogical}.

\begin{lemma}[Smoothing]\label{FRSS-lem:smooting}
Fix any relatively compact open $0$-neighbourhood $U$ in $G$, and take an open $0$-neighbourhood $V$ in $G$ such that $V+V\subseteq U$.
Take a natural number $n$ such that $n$ translates of\/ $-V$ cover\/ $U-U$. Let $\psi \in \Lisp(G)$ be such that  $\supp(\psi) \subseteq V$ and $m_G(\psi) =1$.
Let\/ $\varPhi=(\phi_i)_{i\in I}$ be a family of functions in $\Lsp^\infty(G)$.
Let $M,B$ be finite constants. Let $X=(x_i)_{i\in I}$ be points in $G$ with the following properties.
\begin{itemize}
\setlength\itemsep{2pt}
\item[(a)] $\supp(\phi_i)\subseteq x_i +V$ for all $i \in I$,
\item[(b)] $\card\{j\in I: (x_i +V)\cap (x_j+V)\ne \varnothing\}\le B/n$ for all $i\in I$.
\item[(c)] $\sum_{i\in I} \phi_i(x)=1$ for almost all $x\in G$,
\item[(d)] $\|\phi_i\|^{}_{\infty}\le M/\|\psi\|^{}_{1}$ for all $i \in I$,
\end{itemize}
Then, $\varPsi= (\phi_i*\psi)_{i \in I}$ is a BUPU of norm $M$, size $U$ and overlap constant $B$ with point family $X$.
\end{lemma}

\begin{remark}
(a) A natural number $n$ with the above properties indeed exists, since $V$ is open and $U$ is relatively compact.

\noindent (b) $\varPhi$ is a BUPU if and only if $\phi_i$ is continuous for all $i \in I$. \exend
\end{remark}

\begin{beweis}
We check conditions (a) to (d) in Definition~\ref{FRSS-def:BUPUCc}. First note that $\phi_i \in \Lsp^\infty(G)$ and $\psi \in \Lisp(G)$ imply $\phi_i*\psi \in \Cusp(G)$, as the norm $\|\cdot\|^{}_{1}$ is translation invariant and translation continuous; compare~\cite[Lemma~1.4.2]{FRSS-DE} or \cite[Prop.~3.5.6]{FRSS-Rei2}.
Moreover, we have
\[
\| \phi^{}_i*\psi \|^{}_{\infty} \leq \| \phi^{}_i \|^{}_\infty \cdot \| \psi \|^{}_{1} \, \leq \, M \ts ,
\]
which proves (d). 

\noindent Condition (a) is obvious as
\[
\supp(\phi_i*\psi) \, \subseteq \, x_i+V+V\subseteq x_i+U \ts ,
\]
and we also get $\phi_i*\psi \in \Ccsp(G)$.

\smallskip

\noindent Let us prove (b) next. By assumption, there exists a finite set $F\subseteq G$ such that $\card(F)=n$ and
$U-U  \subseteq  \bigcup_{x\in F} (x-V)$.
Then, by a standard estimate we get
\[
\|\delta^{}_{X} \|^{}_{U-U}\, \le \, n \|\delta^{}_{X} \|^{}_{-V} \ts .
\]
Lemma~\ref{FRSS-lem:BUPU-tb} thus yields
\begin{displaymath}\label{FRSS-eq:up}
\begin{split}
\sup_{i \in I} &\card\{j\in I: (x_i +U)\cap (x_j+U)\ne \varnothing\} \,\leq\, \| \delta^{}_{X} \|^{}_{U-U} \\
&\,\leq\, n \|\delta^{}_{X} \|^{}_{-V}
\,\leq\, n \sup_{i \in I} \card\{j\in I: (x_i +V)\cap (x_j+V)\ne \varnothing\} \,\leq\,  B \ts .
\end{split}
\end{displaymath}
Thus (b) is satisfied.

\smallskip

\noindent In order to prove (c), let $x \in G$ be arbitrary and consider the  set
\[
I_x \,=\, \{ i \in I : (x-V) \cap \supp(\phi_i) \ne \varnothing \} \ts .
\]
Let $i \in I_x$. Then, since $\supp(\phi_i) \subseteq x_i+V$ we have 
$(x-V) \cap (x_i+V) \, \neq \, \varnothing$. 
We thus have $x_i \in x+V+V$, which implies $\card(I_x)\le\|\delta^{}_X\|^{}_{V+V} < \infty$.
Consequently,
\begin{displaymath}
\begin{split}
1 &\,= \int_{G} \psi(x-t) \dd t  \,= \int_{x-V} \psi(x-t) \dd t
\,= \int_{x-V} \psi(x-t) \left(\sum_{i\in I}\phi_i(t)\right)  \dd t \\
&\,= \int_{x-V} \psi(x-t) \left(\sum_{i\in I_x}\phi_i(t)\right)  \dd t
\,= \sum_{i\in I_x} \int_{x-V} \psi(x-t)\phi_i(t)  \dd t \\
&\,= \sum_{i\in I} \int_{x-V} \psi(x-t)\phi_i(t)  \dd t  \,=\sum_{i\in I}  \psi*\phi_i(x) \ts ,
\end{split}
\end{displaymath}
which yields the claim.
\qed \end{beweis}

\begin{remark}\label{FRSS-rem:sba} We will later work with BUPUs in a suitable Banach algebra $(\Asp,\|\cdot\|^{}_{\Asp})$ of functions over $G$, see Definition~\ref{FRSS-def:BUPUA} below. This generalises the case $(\Asp, \|\cdot\|^{}_{\Asp})=(\COsp(G), \|\cdot\|^{}_{\infty})$ from Definition~\ref{FRSS-def:BUPUCc}.

The above smoothing technique can then be applied as follows:
Start with a BUPU $\varPhi=(\phi_i)_{i\in I}$ in $\COsp(G)$ of size $U$. Pick some $\psi \in \Lisp(G)$ supported inside a $0$-neighbourhood $V$, and normalised such that we have $m_G(\psi) =1$. If $\phi_i*\psi \in \Asp$  and if there exists a finite constant $C$ such that
\[
\|\phi_i*\psi\|^{}_{\Asp} \leq C \|\phi_i\|^{}_{\infty}
\]
for all $i\in I$, then $\varPsi= (\phi_i*\psi)_{i \in I}$ is a BUPU in $\Asp$ of size $U+V$.

For example, consider the Fourier algebra\index{Fourier~algebra} $\Asp=\mathcal F \Lisp(\widehat G)$ of $G$ with norm
\[
\|f\|^{}_{\Asp}\,=\, \|\widecheck f\|^{}_{1} \ts .
\] If the above $\psi$ additionally satisfies $\psi\in \Asp$, we then indeed have
\begin{displaymath}
\|\phi_i*\psi\|^{}_{\Asp}\,=\, \|\widecheck \phi_i\cdot \widecheck \psi\|^{}_{1}\,\le\, \|\widecheck \phi_i\|^{}_{\infty} \cdot \|\widecheck \psi\|^{}_{1} \,\le\, \|\phi_i\|^{}_{1} \cdot \|\psi\|^{}_{\Asp} \,\le\, m_G(U) \|\phi_i\|^{}_{\infty} \cdot \|\psi\|^{}_{\Asp} \ts .
\end{displaymath}
\exend
\end{remark}

\subsection{BUPUs centered at Delone sets}

Here, we give two different constructions of BUPUs based on the existence of Delone sets in $G$. Our first construction follows \cite[Thm.~2]{FRSS-Fei81b}, see also \cite[Thm.~1]{FRSS-Fei22}. In the tempered distribution setting, similar constructions appear in \cite[Thm.~1.4.6 and Lemma~1.4.9]{FRSS-H03}.
The following well-known result (see \cite[Lemma~1]{FRSS-Fei81b} and references therein) states that there exist $\overline{U}$-relatively dense Delone sets, for any relatively compact $0$-neighbourhood $U$ in $G$. The proof uses that if $\varLambda$ is a maximal $V$-uniformly discrete subset of $G$, then
\[
\varLambda+(V-V) \,=\,G \ts .
\]
Existence of such a maximal set is ensured by Zorn's lemma.

\begin{lemma}\label{FRSS-lemma:existe-Del} Let $U$ be any relatively compact $0$-neighbourhood in $G$ and let $V$ be an open $0$-neighbourhood in $G$ such that $V-V\subseteq U$. Then, for any $\varGamma\subseteq G$ the following hold.
\begin{itemize}
\setlength\itemsep{2pt}
  \item[(a)] There exists a set $\varLambda\subseteq \varGamma$ with the following properties.
  \begin{itemize}
  \item $\varLambda$ is $V$-uniformly discrete;
  \item If $\varLambda\subsetneq \varLambda'\subseteq \varGamma$, then $\varLambda'$ is not $V$-uniformly discrete.
  \end{itemize}
  \item[(b)] Any set $\varLambda\subseteq \varGamma$ with the properties from {\rm (a)} satisfies $\varGamma\subseteq \varLambda+U$. 
\end{itemize}
In particular, for $\varGamma=G$ we have $\varLambda+U=G$ in {\rm (b)}. In that case, $\varLambda$ is $V$-uniformly discrete and $\overline{U}$-relatively dense.
\end{lemma}

\begin{beweis}

(a): Consider the collection $\mathcal P$ of all $V$-uniformly discrete subsets of $\varGamma$, together with set inclusion as partial order. Let $\mathcal C \, = \, \{\varGamma_i: i \in I\}$ be any chain in $\mathcal P$ and define the upper bound
\[
\varGamma_{\mathcal C} = \bigcup_{i \in I} \varGamma_i \subseteq \varGamma
\]
for $\mathcal C$. We show that $\varGamma_{\mathcal C}$ is $V$-uniformly discrete, which implies $\varGamma_{\mathcal C}\in \mathcal P$.
Let $x,y \in \varGamma_{\mathcal C}$ be such that $x\neq y$. Then, there exist some $i,j \in I$ such that $x \in \varGamma_i$ and $y \in \varGamma_j$. Let $k = \max \{i,j\}$. Then, $x,y \in \varGamma_k$ since $\mathcal C$ is a chain. As $\varGamma_k \in \mathcal P$, we get $(x+V) \cap (y+V) =\varnothing$, which shows $\varGamma_{\mathcal C} \in \mathcal P$.
Zorn's lemma now yields that $\mathcal P$ has some maximal element $\varLambda \in \cP$. This element satisfies the conditions of (a) by maximality.

\noindent (b): Assume that $\varLambda\subseteq \Gamma$ has the properties in (a). 
Clearly, every $x\in \varLambda$ satisfies $x\in \varLambda+U$ as $0\in U$. Consider an arbitrary $x \in \varGamma\setminus\varLambda$.
Then, $\varLambda \cup \{ x \}$ is not $V$-uniformly discrete. Since $\varLambda$ is $V$-uniformly discrete, there exists some $y \in \varLambda$ such that $(x+V) \cap (y+V) \, \neq \, \varnothing$.
Let $u,v \in V$ be such that $x+u=y+v$. We ghen have
\[
x\,=\,y+v-u \,\in\, y+V-V \,\subseteq\, y+U \,\subseteq\,   \varLambda+U \ts ,
\]
which shows $\varGamma\subseteq \varLambda+U$.
\qed \end{beweis}

Using a maximal uniformly discrete Delone set, we can construct a BUPU of given size $U$ as follows.

\begin{proposition}\label{FRSS-prop:positive-BUPU}
Let $U\subseteq G$ be a relatively compact $0$-neighbourhood and let $V\subseteq G$ be an open $0$-neighbourhood such that $V-V\subseteq U$. Let $\varLambda\subseteq G$ be a Delone set such that $\varLambda+V=G$.
Fix $\psi \in \Ccsp(G)$ supported within $U$ such that $\psi \geq 1^{}_{V}$, and define
$f = \delta^{}_{\!\varLambda}* \psi\in \Cusp(G)$, where
\[
(\delta^{}_{\!\varLambda}* \psi)(x) \, =\, \sum_{\lambda\in \varLambda} \psi(x -\lambda) \ts .
\]
Then, $f\ge 1$, and the functions $(\phi_\lambda)_{\lambda\in \varLambda}$, given by $\phi_\lambda(x)= \psi(x-\lambda)/f(x)$ constitute a BUPU of positive functions in $\Ccsp(G)$ having norm $1$ and size $U$ with point family $\varLambda$.
\end{proposition}

\begin{beweis}
Existence of a Delone set $\varLambda$ such that $\varLambda+V=G$ is granted by Lemma~\ref{FRSS-lemma:existe-Del}. As $\varLambda$ is Delone, its Dirac comb $\delta^{}_{\!\varLambda}$ is translation bounded. Hence, by Remark~\ref{FRSS-rem:cd}, condition (b) in the definition of BUPU is satisfied with $B=\|\delta^{}_{\!\varLambda}\|^{}_{U-U}$.
Next, for all $x \in G$ we have
\[
f(x) \, = \,  \delta^{}_{\!\varLambda} * \psi (x) \, = \, \sum^{}_{\lambda \in \varLambda}  \psi(x-\lambda)
\ge \, \sum^{}_{\lambda \in \varLambda}  1^{}_{V}(x-\lambda) \, \ge 1 \ts ,
\]
where the last inequality follows from $\varLambda+V=G=G$.

\noindent Also, $f\in \Cusp(G)$ as $\delta^{}_{\!\varLambda}$ is translation bounded; compare Section~\ref{FRSS-sec:tbmintro} or \cite[Lemma~4.9.19]{FRSS-MoSt}. Indeed, standard estimates show
\begin{displaymath}
|(\delta^{}_{\!\varLambda}*\psi)(x)-(\delta^{}_{\!\varLambda}*\psi)(y)| \,\le\, 2 \cdot \|T_{x-y}\psi-\psi\|^{}_{\infty} \cdot \|\delta^{}_{\!\varLambda}\|^{}_{U} \ts ,
\end{displaymath}
where $(T_x\psi)(z)=\psi(z-x)$. Now the claim follows by uniform continuity of $\psi\in \Ccsp(G)$.
As a consequence, we have $\phi_\lambda \in \Ccsp(G)$ and $\phi_\lambda\ge0$ for all $\lambda\in\varLambda$.

\noindent Now, for any fixed $x \in G$, we have
\begin{displaymath}
\sum_{\lambda \in \varLambda} \phi_\lambda(x)
 \,=\,\frac{1}{f(x)}\sum_{\lambda \in \varLambda } \psi(x-\lambda)
 \,=\,\frac{1}{f(x)} (\delta^{}_{\!\varLambda}*\psi)(x)
 \,=\,1 \ts ,
\end{displaymath}
where the above sums are all finite due to translation boundedness of $\delta^{}_{\!\varLambda}$.
This proves condition (a) in the definition of a BUPU.

\noindent Moreover, we clearly have $0 \leq \phi_{\lambda} \leq 1$, which implies $\|\phi_{\lambda} \|^{}_{\infty} \leq 1$ for all $\lambda \in \varLambda$. Finally, we observe $\supp(\phi_{\lambda})=\lambda+\supp(\psi) \subseteq \lambda+U$, which completes the proof.
\qed \end{beweis}

The second construction of a BUPU uses that the given Delone set admits a total order, as can be seen using Zorn's lemma \cite[p.~5]{FRSS-Fo99}. In comparison to Proposition~\ref{FRSS-prop:positive-BUPU}, it is better adapted to the underlying function algebra structure, as it needs no division. Our presentation is extracted from \cite[Thm.~4.2]{FRSS-HGF87}, where a construction of so-called bounded adapted partitions of unity (BAPUs) is given for discrete point sets that are relatively dense, but may fail to be uniformly discrete. The following construction has also been used in the tempered distribution setting, compare~\cite[Lemma~1.4.9]{FRSS-H03} or \cite[Thm.~6.20]{FRSS-RUD3}.

\begin{proposition}
Fix any relatively compact $0$-neighbourhood $U$ in $G$ and choose some open $0$-neighbourhood $V$ such that $\overline{V}+\overline{V}+\overline{V}\subseteq U$. Take a Delone set $\varLambda$ which is $V$-uniformly discrete and $K$-relatively dense, where $K=\overline{V}+\overline{V}$. Let $\le$ be any total order on $\varLambda$.
Take $\psi\in \Ccsp(G)$ supported within $K+\overline{V}\subseteq U$, satisfying  $0\le \psi\le 1$, such that $\psi|^{}_K=1$. Define now
\[
\psi_\lambda(x) \, = \, \psi(x-\lambda) \, , \qquad \phi_\lambda \, = \, \psi_\lambda\prod_{\lambda'<\lambda}(1-\psi_{\lambda'}) \ts .
\]
Then, $(\phi_\lambda)_{\lambda\in \varLambda}$ is a BUPU of non-negative functions in $\Ccsp(G)$ of size $U$.
\end{proposition}

\begin{beweis}
Existence of a Delone set $\varLambda$ with the claimed properties follows from Lemma~\ref{FRSS-lemma:existe-Del}.
Note that, for each $\lambda$, the set
\[
F_\lambda \,=\, \{ \lambda' \in \varLambda: \lambda' < \lambda, \supp(\psi_{\lambda'})\cap \supp(\psi_{\lambda}) \neq \varnothing \}
\]
is finite as $\varLambda$ is uniformly discrete. It follows that
\[
\phi_\lambda \, = \, \psi_\lambda\prod_{\lambda'<\lambda}(1-\psi_{\lambda'}) \, = \, \psi_\lambda \prod_{\lambda' \in F_\lambda}(1-\psi_{\lambda'}) \in \Ccsp(G) \ts ,
\]
and $\supp(\phi_\lambda)  \subseteq \supp(\psi_\lambda) \subseteq \lambda+U$.
Moreover we have $0 \leq \phi_\lambda \leq 1$, as $0 \leq \psi_\lambda \leq 1$ and $0 \leq 1-\psi_{\lambda'} \leq 1$.
In particular, for all $\lambda \in \varLambda$ we have $\| \phi_\lambda \|_\infty \leq  1$.
To complete the claim, we have to prove that
\[
\sum_{\lambda\in \varLambda} \phi_\lambda \,= \, 1 \ts .
\]
Fix arbitrary $x\in G$ and consider $F_x = \{ \lambda \in \Lambda : x \in \lambda+U \}$. Then, $F_x$ is finite as $\varLambda$ is uniformly discrete and as $U$ is relatively compact. As $\leq$ is a total order, we can write
\[
F_{x} \, = \, \{ \lambda_1, \lambda_2, \ldots, \lambda_n \} \ts , \qquad \lambda_1< \lambda_2 < \ldots < \lambda_n \ts .
\]
Since $\varLambda$ is $K$-relatively dense and $K \subseteq U$, there exists $\lambda_x\in \varLambda$ such that $x\in\lambda_x+K\subseteq \lambda_x+U$. Hence $\lambda_x=\lambda_k$ for some $k \in \{1,\ldots, n\}$, and we have $\psi_{\lambda_k}(x)  =  1$.
This implies that for all $\lambda\in \varLambda$ such that $\lambda>\lambda_k$ we have
\[
\phi_{\lambda}(x)=\psi_\lambda(x) \prod_{\lambda'<\lambda}(1-\psi_{\lambda'}(x)) =0 \,.
\]
Also note that $\phi_\lambda(x)=0$ for all $\lambda \in \varLambda \backslash F_x$, as we have $\psi_\lambda(x)=0$ in this case.
Therefore,
\begin{displaymath}
\begin{split}
  1-\sum_{\lambda\in \varLambda} \phi_\lambda(x) &\,=\, 1- \sum_{\lambda\in F_x} \phi_\lambda(x) 
  \,=\, 1- \sum_{j=1}^n \phi_{\lambda_j}(x) \,=\, 1- \sum_{j=1}^k \phi_{\lambda_j}(x) \\
  &\,=\, 1- \sum_{j=1}^k \psi_{\lambda_j}(x) \prod_{i=1}^j(1-\psi_{\lambda_j}(x)) \,=\,\prod_{j=1}^k (1-\psi_{\lambda_j}(x))\,=\,0 \ts ,
\end{split}
\end{displaymath}
where the last equality follows from $\psi_{\lambda_k}(x)= 1$. This proves the claim.
\qed \end{beweis}

\begin{remark}
The above two BUPU constructions also work in suitable Banach algebras $(\Asp,\|\cdot\|^{}_{\Asp})$ beyond $\COsp(G)$, as mentioned in Remark~\ref{FRSS-rem:sba}.

\noindent Consider for example the Fourier algebra
\[
\Asp=\Asp(G)=\mathcal F \Lisp(\widehat G)
\]
of $G$ with norm $\|f\|^{}_{\Asp}=\|\widecheck f\|^{}_{1}$.
From the first construction, a BUPU in $\Asp$ is obtained by smoothing as explained in Remark~\ref{FRSS-rem:sba}. A more general statement is \cite[Thm.~2]{FRSS-Fei81b}.
For the second construction, we note that there exists $\psi\in \Ccsp(G)\cap \Asp(G)$ with the desired properties by \cite[Lemma~4.3(iv)]{FRSS-Jak}. Since $\Asp(G)$ is a Banach algebra with respect to pointwise multiplication, the construction proceeds as in the case of $\COsp(G)$.
\exend
\end{remark}

\subsection{BUPUs via the structure theorem}

Arbitrarily fine BUPUs on an LCA group $G$ can also be constructed using the structure theorem \cite[Thm.~4.2.29]{FRSS-Rei2}, which states that every compactly generated LCA group $G$ is isomorphic to a product $\RR^d\times \ZZ^m\times \KK$ for some integers $d,m\ge 0$ and some compact group $\KK$. We split the construction as follows. The first step uses Example~\ref{FRSS-ex:scaledBUPU}.
\begin{lemma}\label{FRSS-lem-BUPu-in-R} $G=\RR$ admits a BUPU $\varPhi$ of positive functions with norm $M=1$ and size $U=[-r,r]$ for any $r>0$ in $\RR$, with overlapping constant $B=3$ and lattice $L=r \ZZ$ as point family.
Moreover, $\varPhi$ consists of translates of a single function. \qed
\end{lemma}

The second step concerns discrete groups. The following result is trivial.

\begin{lemma}\label{FRSS-prop-BUPu-in-D} Let $G=D$ be a discrete LCA group. Let $\phi_a=1_{\{a\}}$ for $a\in D$. Then, $\varPhi = \{ \phi_a \}_{a \in D}$ is a BUPU of positive functions with norm $M=1$ and size $U=\{0\}$, with overlapping constant $B=1$ and lattice $D$ as point family. Moreover, all functions in $\varPhi$ are translates of one another.
\qed
\end{lemma}

The third result concerns compact groups. The proof is analogous to that of Proposition~\ref{FRSS-prop:positive-BUPU}.

\begin{lemma}\label{FRSS-prop-BUPu-in-K} Let $G=\KK$ be a compact LCA group and let $U$ be any open $0$-neighbourhood in $\KK$. Then, there exists a BUPU $\varPhi$ of positive functions with norm $M=1$, size $U$, and with some finite set $F$ as point family.
\end{lemma}

\begin{beweis} Pick an open $0$-neighbourhood $V$ such that $\overline{V} \subseteq U$, and pick some function $\varphi\in \Ccsp(\KK)$ such that $\varphi\ge 1_{\overline{V}}$ and $\supp(\varphi)\subseteq U$. Take a finite set $F$ such that $F+V=\KK$.
Then, the function
\[
f(x)\,=\,\sum_{j \in F} \varphi(x-j)
\]
satisfies $f\in \Ccsp(\KK)$ and $f \geq 1$. Now, defining
\[
\phi_j(x) \,=\,\varphi(x-j)/f(x)
\]
yields a BUPU $\varPhi=(\phi_j)_{j \in F}$ with the claimed properties.
\qed \end{beweis}

Combining the previous results, we infer that any compactly generated LCA group admits a BUPU of the following type.

\begin{proposition}\label{FRSS-prop:bupuCG}
Let $G$ be a compactly generated LCA group. Then, for any relatively compact open $0$-neighbourhood $U$ in $G$, there exists a BUPU $\varPhi$ of size $U$ of positive functions with norm $M=1$ with the following properties:
\begin{itemize}
  \item[(i)] The point family is of the form $L+F$ for a lattice $L$ and a finite set $F$ in $G$,
  \item[(ii)] There exists a finite set $\cF\subseteq \Ccsp(G)$ such that $\varPhi$ consists of translates of functions from $\cF$.
\end{itemize}
\end{proposition}

\begin{beweis}
By the structure theorem of compactly generated LCA groups \cite[Thm.~4.2.29]{FRSS-Rei2}, we can identify $G$ by $\RR^d \times \ZZ^m \times \KK$ for some $d,m\ge0$ and some compact group $\KK$.

\noindent Pick a positive number $r>0$ and take some $0$-neighbourhood $V$ in $\KK$ such that $[-r,r]^d \times \{ 0\} \times V \subseteq U$.
By Lemma~\ref{FRSS-prop-BUPu-in-D} and by repeated application of Lemma~\ref{FRSS-lem-BUPu-in-R} and Lemma~\ref{FRSS-lem:prodbupu}, there exists a BUPU in $\RR^d \times \ZZ^m$ of positive functions with  size $[-r,r]^d \times \{ 0 \}$ and norm $M=1$, having a lattice $L' \subseteq \RR^d \times \ZZ^m$ as point family. Moreover, the functions of the BUPU are translates of one another.

\noindent Applying Lemma~\ref{FRSS-prop-BUPu-in-K} and Lemma~\ref{FRSS-lem:prodbupu} again, we get a BUPU $\varPhi$ of positive functions with norm $M=1$ and size $U$, where the point family is the product of the lattice $L' \subseteq \RR^d \times \ZZ^m$ and some finite set $F^{\ts\prime} \subseteq \KK$.  It follows that the point family is of the form $L+F$, where $L = L' \times \{ 0\}$ and $F = \{ 0 \} \times \{ 0 \} \times F^{\ts\prime}$.

\noindent Finally, each element in the BUPU $\varPhi$ is a product of the form $\varphi \otimes \psi$ where $\varphi\in \Ccsp(\RR^d\times\ZZ^m)$ is in a family of translates of a single function, and where $\psi\in \Ccsp(\KK)$ belongs to a finite family. Therefore, all functions in the BUPU are translates of a finite set of functions.
\qed \end{beweis}

For a general LCA group, we use the following result, whose proof is obvious. Here and in the following, we will denote by $T_xf$ the shift of the function $f:G\to \mathbb C$ by argument $x\in G$, i.e., we have $(T_xf)(y)=f(y-x)$.

\begin{lemma}\label{FRSS-lem:bupu-open} Let $G'$ be an open subgroup in some LCA group $G$. Fix any open relatively compact $0$-neighbourhood $U \subseteq G'$. Let $R\subseteq G$ be a system of
representatives for $G/G'$, that is, any coset $x+G'$ contains exactly one element in $R$. Let $\varPsi=(\psi_i)_{i\in I}$ be any BUPU in $G'$ of size $U$ and norm $M$ with point family $Y=(y_i)_{i\in I}$.
Then,
\[
\varPhi \, = \,  (T_{r}\psi_{i})_{(i,r) \in I \times R} \ts
\]
is a BUPU in $\Ccsp(G)$ of size $U$ and norm $M$ with associated point family $X=(y_i+r)_{(i,r) \in I \times R}$. \qed
\end{lemma}

We can now state the promised result about BUPUs constructed via the structure theorem.

\begin{theorem} Let $U\subseteq G$ be any relatively compact $0$-neighbourhood in some LCA group $G$.
Then, there exists a BUPU $\varPhi$ of positive functions with norm $M=1$ and size $U$, such that the point family is a Delone set, and such that any function in $\varPhi$ is a translate of some function from some finite set $\cF\subseteq \Ccsp(G)$.
\end{theorem}

\begin{beweis}
It is obvious from the construction in Lemma~\ref{FRSS-lem:bupu-open} that if $\varPsi$ consists of positive functions, then so does $\varPhi$. Moreover, the functions in $\varPhi$ are translates of functions in $\varPsi$. Therefore, if the functions in $\varPsi$ are translates of a finite set of functions, so are the functions in $\varPhi$. Now, given any open relatively compact $0$-neighbourhood $U$ in $G$, the group $G'$ generated by $\overline{U}$ is compactly generated and open in $G$. Consequently, Proposition~\ref{FRSS-prop:bupuCG} and Lemma~\ref{FRSS-lem:bupu-open} yield the claim of the theorem.
\qed \end{beweis}

\section{The Wiener algebra}\label{FRSS-sec:WA}

We review the Wiener algebra on an LCA group $G$, following the approach from \cite{FRSS-F83} based on BUPUs. The dual of the Wiener algebra coincides with the space of translation-bounded measures, as shown in Section~\ref{FRSS-sec:tbmintro}. In the following, we will sometimes suppress the underlying group.

\subsection{Definition and elementary properties}\label{FRSS-sec:DEP}

During the 1930's, Wiener developed Fourier analysis for time signals \cite{FRSS-W30}, compare Section~\ref{FRSS-sec:app}. In conjunction with his Tauberian theorems, he introduced a certain algebra $M_1$ of continuous integrable functions  on the line, see \cite[Ch.~I.3]{FRSS-W32}, \cite[p.~73]{FRSS-W59} and \cite[(39.33)]{FRSS-HR70} for a discussion. This was extended to locally compact groups $G$ in \cite{FRSS-Rei1, FRSS-LRW74}; the function algebra was then denoted by $M_1(G)$. A characterisation using Segal algebra theory was given in \cite{FRSS-fe77-3}, see Section~\ref{FRSS-sec:SegAlgChar} for details.
The identification as amalgam space  $\WCOli(G)$ appears in \cite{FRSS-F83}.

\smallskip

Here, we introduce the Wiener algebra $\Wsp(G)$ using BUPUs, see also \cite{FRSS-Fei22} for a recent perspective. The identification $\Wsp(G)=\WCOli(G)=M_1(G)$ then follows from \cite[Thm.~2]{FRSS-F83}.  For concreteness, we derive some of its properties directly from our definition of $\Wsp(G)$.
\begin{definition} Let $\varPhi=(\phi_i)_{i\in I}$ be any BUPU in $\COsp(G)$. Then the \emph{Wiener algebra}\index{Algebra!Wiener}\index{Wiener~algebra} is defined as
\[\index{$\| \, \|^{}_{\Wsp,\varPhi}$}
\Wsp(G)\,=\, \{ f \in \COsp(G) : \|f \|^{}_{\Wsp,\varPhi}\, = \sum_{i \in I} \|f \cdot \phi_i \|^{}_{\infty} < \infty \} \ts .
\]
\end{definition}
It is clear that $\|\cdot\|^{}_{\Wsp, \varPhi}$ defines a norm on $\Wsp(G)$. We will show in Theorem~\ref{FRSS-thm:W} below that different BUPUs in $\COsp(G)$ lead to equivalent norms. The next lemma lists some of the norm $\|\cdot\|^{}_{\Wsp, \varPhi}$.
In the following, we often write $\|\cdot\|^{}_{\varPhi}$ instead of $\|\cdot\|^{}_{\Wsp,\varPhi}$ for simplicity.

\begin{lemma}\label{FRSS-lem:normpropw}
\begin{itemize}
\item[(a)] The norm $\|\cdot\|^{}_{\varPhi}$ satisfies $\|\cdot\|^{}_{\infty}\le \|\cdot\|^{}_{\varPhi}$.
\item[(b)] Consider any relatively compact $0$-neighbourhood $V\subseteq G$. Then there exists a finite constant $C$ such that $\|f \|^{}_{\varPhi}\le C \|f \|^{}_{\infty}$ for all $f\in \Wsp(G)$ supported within some translate of $V$.
\item[(c)]  The norm $\|\cdot\|^{}_{\varPhi}$ is translation bounded, meaning that there exists some finite constant $C$ such that $\|T_x f\|^{}_{\varPhi}\le C \|f\|^{}_{\varPhi}$ holds for all $f\in \Wsp(G)$. In particular, $T_xf\in \Wsp(G)$ for all $f\in \Wsp(G)$.
\item[(d)]  The norm $\|\cdot\|^{}_{\varPhi}$ is translation continuous, meaning that we have $\lim_{x\to0} \|T_xf-f\|^{}_{\varPhi}=0$ for every $f\in\Wsp(G)$.
\end{itemize}
\end{lemma}

\begin{beweis}
(a): Note that we have for all $f\in \Wsp(G)$ the estimate
\begin{displaymath}
\|f\|^{}_{\infty}\,=\, \|\sum_{i\in I} f\phi_i \|^{}_{\infty} \le \sum_{i\in I} \|f\phi_i\|^{}_{\infty}= \|f\|^{}_{\varPhi} \ts .
\end{displaymath}

\noindent (b): Fix any relatively compact $0$-neighbourhood $V\subseteq G$.  Consider any $f\in \Wsp(G)$ such that $\supp(f)\subseteq y+V$ for some $y\in G$. Using Lemma~\ref{FRSS-lem:BUPU-tb}, we can then estimate
\begin{displaymath}
\|f\|^{}_{\varPhi}\,=\, \sum_{i\in I} \|f\phi_i\|^{}_{\infty}\,\le\, \|\delta^{}_{X}\|_{V-U} \cdot M^{}_\varPhi\cdot \|f\|^{}_{\infty} \ts .
\end{displaymath}

\noindent (c): Let $f\in \Wsp(G)$ and $x\in G$ be arbitrary. Then $T_xf\in \COsp(G)$, and we can use  Eq.~\eqref{FRSS-eq:refined} to estimate
\begin{displaymath}
\begin{split}
\|T_xf\|^{}_{\varPhi} \, &=\, \|f\|^{}_{T_{-x}\varPhi}
\, \le \, \|\delta^{}_{T_{-x}X}\|_{U-U} \cdot M^{}_{T_{-x}\varPhi}\cdot   \|f\|^{}_{\varPhi}
 = \|\delta^{}_{X}\|^{}_{U-U} \cdot M^{}_{\varPhi}\cdot \|f\|^{}_{\varPhi} \ts .
\end{split}
\end{displaymath}
Here $T_{-x}\varPhi=(T_{-x}\phi_i)_{i\in I}$ is the BUPU of size $U$ with point family given by  $T_{-x}X=(x_i-x)_{i\in I}$.
This shows the claimed norm estimate and, in particular, that $T_xf\in \Wsp(G)$.

\noindent (d): Let $f\in \Wsp(G)$ be arbitrary and consider for arbitrary $J\subseteq I$ the function
\[
f_J \,=\,\sum_{i\in J} f\phi_i\in\COsp(G) \ts .
\]
Then an estimate as in Eq.~\eqref{FRSS-eq:refined} yields
\begin{equation}\label{FRSS-eq:fJ}
\|f_J\|^{}_{\varPhi} \, \le \, \|\delta^{}_{X}\|^{}_{U-U} \cdot M^{}_{\varPhi}\cdot \sum_{i\in J}\| f\phi_i\|^{}_{\infty} \ts .
\end{equation}
In particular, we have $f_J\in \Wsp(G)$.
Let now $\varepsilon>0$ be arbitrary. By Eq.~\eqref{FRSS-eq:fJ}, we can find an index set $J \subseteq I$ such that $I \backslash J$ is finite and
\[
\|f_J\|^{}_{\varPhi} \,<\, \frac{\varepsilon}{2(C+1)} \ts ,
\]
where $C$ is the constant from (c). We thus have for all $x \in G$ the estimate
\[
\|f_J-T_x f_J\|^{}_{\varPhi} \, \le \, \|f_J\|^{}_{\varPhi}+\|T_x f_J\|^{}_{\varPhi} \, \le \, (C+1)\cdot \|f_J\|^{}_{\varPhi} \,<\, \frac{\varepsilon}{2} \ts .
\]
We next derive a similar estimate on the function
\[
f_0\,= \, f-f_{J}= \sum_{i\in I \backslash J} f\phi_i \ts .
\]
Note that $f_0\in \Ccsp(G)$, as the sum is finite. Fix any open relatively compact $0$-neighborhood $V_0\subseteq G$ and take a  relatively compact $0$-neighborhood $V\subseteq G$ such that $\supp(f_0-T_xf_0) \subseteq V$ for all $x\in V_0$. 
Recalling (b), we may take a finite constant $C_V$ such that we have
\begin{displaymath}
\|f_0-T_xf_0\|^{}_{\varPhi} \,<\,  C_V \cdot \| f_0-T_xf_0 \|^{}_{\infty}
\end{displaymath}
for all $x\in V_0$.
Now, since $f_0 \in \Ccsp(G)$ is translation continuous, there exists some 0-neighbourhood $U_0\subseteq V_0$ such that for all $x \in U_0$ we have
\[
\| f_0-T_x f_0 \|^{}_\infty \,<\, \frac{\eps}{2C_V+1} \ts .
\]
We thus have $\|f_0-T_x f_0\|^{}_{\varPhi}<\eps/2$ for all $x\in U_0$.
Combining the above estimates on $f_J$ and $f_0$, it follows that for all $x \in U_0$ we have
\[
\|f-T_x f\|^{}_{\varPhi} \, \le \, \|f_J-T_x f_J\|^{}_{\varPhi} + \|f_0-T_x f_0\|^{}_{\varPhi} \,<\, \eps \ts .
\]
This shows translation continuity of $f$.
\qed \end{beweis}

\smallskip
\begin{remark}[Index sets of BUPUs in $\Wsp(G)$]\label{FRSS-rem:Wiscountable}
Consider any BUPU $\varPhi=(\phi_i)_{i\in I}$ in $\COsp(G)$. In general, we cannot assume that $I$ is countable. Indeed, countability of $I$ implies that $G$ is $\sigma$-compact. However, when dealing with functions $f \in \Wsp(G)$, we may assume that corresponding sums are countable. Indeed, for $f\in \Wsp(G)$ consider the index set  $J  =   \{ i \in I : f \phi_i \neq 0 \}\subseteq I$. We then have $J= \bigcup_{n \in \NN} J_n$, where
\[
J_n \, = \,  \{ i \in I : \| f \phi_i \|_\infty \geq 1/n\} \ts .
\]
Since $\sum_{i \in I} \|f \cdot \phi_i \|^{}_{\infty}  < \infty$, each $J_n$ is finite, and hence $J$ is countable. Note here that the definition of $J$ and the fact that $\varPhi$ is a BUPU imply that
\[
f \, = \, \sum_{i \in J} f \phi_i \ts ,
\]
both in $(\Wsp(G), \|\cdot\|^{}_{\Wsp, \varPhi})$ and in $(\COsp(G), \|\cdot\|^{}_{\infty})$, where the right hand side of the above equation is a countable sum. We will use this representation of $f$ often below. \exend
\end{remark}

The following properties of the Wiener algebra are not difficult to verify. For a function $f$, we denote by $f^\dag$\index{$f^\dag$} its reflection, meaning  $f^\dag(x)=f(-x)$, and by $\widetilde{f}(x)=\overline{f(-x)}$.
\begin{theorem}\label{FRSS-thm:W}
The Wiener algebra $\Wsp(G)$ has the following properties.
  \begin{itemize}
  \setlength\itemsep{2pt}
   \item[(a)] Different choices of $\varPhi$ define equivalent norms on $\Wsp(G)$. In particular, the definition of $\Wsp(G)$ is independent of the choice of $\varPhi$.
   \item[(b)] $f \in \Wsp(G)$ implies $\overline{f}, |f|, f^\dag, \widetilde{f}, T_x f, \chi f \in \Wsp(G)$
   for arbitrary $x\in G$ and $\chi \in \widehat{G}$. Moreover, for any $\varPhi$  there is a finite constant $C^{}_{\varPhi}$ such that 
\begin{displaymath}
\|T_x f\|^{}_{\varPhi} \, \leq  \, C^{}_{\varPhi} \| f\|^{}_{\varPhi} \ts , \qquad 
\| f^\dag\|^{}_{\varPhi} \, = \, \| \widetilde{f} \|^{}_{\varPhi} \leq  \, C^{}_{\varPhi} \| f\|^{}_{\varPhi} \ts ,
\end{displaymath}
\begin{displaymath}
 \|\bar{f}\|^{}_{\varPhi}\,=\, \||f|\|^{}_{\varPhi}\,=\,\|\chi f\|^{}_{\varPhi} \,=\,\| f\|^{}_{\varPhi} \ts , \\
 \end{displaymath}
for arbitrary $f \in \Wsp(G)$, $x\in G$ and $\chi\in\widehat{G}$.

  \item[(c)] If $f \in \Wsp(G)$ and $g \in \Cbsp(G)$, then $fg \in \Wsp(G)$ and
  \[
  \| fg \|^{}_{\varPhi} \, \le \,  \| f \|^{}_{\varPhi}\cdot \| g \|^{}_\infty \ts .
  \]

  \item[(d)] $(\Wsp(G), \| \cdot \|^{}_{\varPhi})$ is a Banach algebra with respect to both pointwise multiplication and convolution.
  \item[(e)]   $\Ccsp(G)$, equipped with the $\|\cdot\|^{}_{\infty}$-inductive limit topology, is dense and continuously embedded in $(\Wsp(G), \|\cdot\|^{}_{\varPhi})$.
  \item[(f)]  $(\Wsp(G), \|\cdot\|^{}_{\varPhi})$ is dense and continuously embedded in $(\COsp(G), \|\cdot\|^{}_{\infty})$. Moreover, it is a  Banach ideal in $(\COsp(G), \|\cdot\|^{}_{\infty})$ with respect to pointwise multiplication.
 \item[(g)]  $(\Wsp(G), \|\cdot\|^{}_{\varPhi})$ is dense and continuously embedded in $(\Lisp(G), \|\cdot\|^{}_{1})$. Moreover, it is a Banach ideal in $(\Lisp(G), \|\cdot\|^{}_{1})$ with respect to convolution.
  \end{itemize}
\end{theorem}

\begin{remark}
Part (a) and (b) of Theorem~\ref{FRSS-thm:W} imply absolute summability of $f\in \Wsp(G)$ on any uniformly discrete point set. In fact, there are equivalent norms on more general so-called BAPUs \cite{FRSS-fegr85, FRSS-HGF87}. This leads to absolute summability beyond uniformly discrete point sets. \exend
\end{remark}

\begin{beweis}
\noindent (a): By Eq.~\eqref{FRSS-eq:refined}, any two BUPUs $\varPhi,\varPsi$ of sizes $U,V$ with point families $X,Y$, respectively, satisfy
\begin{displaymath}
\|f\|^{}_{\varPsi}\,\le\,  \|\delta^{}_{Y}\|^{}_{U-V}\cdot  M_{\varPsi}\cdot \|f\|^{}_{\varPhi} \ts , \qquad
\|f\|^{}_{\varPhi}\,\le\,  \|\delta^{}_{X}\|^{}_{V-U}\cdot  M_{\varPhi}\cdot \|f\|^{}_{\varPsi} \ts .
\end{displaymath}
Thus, the norms $\|\cdot\|^{}_{\varPhi}$ and $\|\cdot\|^{}_{\varPsi}$ are equivalent.

\noindent (b): Let $\varPhi$ be any BUPU in $\COsp(G)$. The equality
\[
\|\bar{f}\|^{}_{\varPhi}\,=\,\||f|\|^{}_{\varPhi}\,=\,\|\chi f\|^{}_{\varPhi}\,=\,\| f\|^{}_{\varPhi} 
\]
follows from the definition of the norm. We also have $\|f\|^{}_{\varPhi}=\|f^\dag\|^{}_{\varPhi^\dag}$,
where $\varPhi^\dag=(\phi_i^\dag)_{i\in I}$ is the BUPU of size $-U$ with point family $(-x_i)_{i\in I}$. Therefore, by (a), there exists a finite constant $c_1$ such that 
\[
\|f^\dag\|^{}_{\varPhi} \, \leq \, c_1 \|f^\dag\|^{}_{\varPhi^\dag} \,=\, c_1 \|f\|^{}_{\varPhi}
\]
for all $f \in \Wsp(G)$. Since $\widetilde{f}=\overline{f^\dag}$ we also have $\|\widetilde{f}\|^{}_{\varPhi} = \|f^\dag\|^{}_{\varPhi}$.
Finally, by Lemma~\ref{FRSS-lem:normpropw}(c), there exists a constant $c_2$ such that
\[
\|T_x f\|^{}_{\varPhi} \, \leq \, c_2 \|f\|^{}_{\varPhi} 
\]
for all $f \in \Wsp(G)$ and $x \in G$. The claim follows by setting $C^{}_{\varPhi}= \max \{c_1,c_2\}$.


\noindent (c): Since $f \in \Wsp(G) \subseteq \COsp(G)$ and $g \in \Cbsp(G)$, we have $fg \in \COsp(G)$. Moreover,
\[
 \| fg \|^{}_{\varPhi} \, = \, \sum_i \|fg \phi_i \|^{}_{\infty} \,\leq\, \| g\|^{}_{\infty} \sum_i \|f \phi_i \|^{}_{\infty}
 \, = \, \| f \|^{}_{\varPhi} \cdot \| g \|^{}_{\infty} \ts .
\]

\noindent (d): The fact that $\Wsp(G)$ is an algebra with respect to pointwise multiplication follows from (c). We will show in (g) that $\Wsp(G)$ is an algebra with respect to convolution. We now show completeness.

\noindent Recall that, in a normed vector space, completeness is equivalent to the property that every absolutely convergent series is norm convergent. Consider any series $\sum_{j\in \mathbb N} f_j$ that is absolutely convergent in $(\Wsp(G), \|\cdot\|^{}_{\varPhi})$. As we have $\|\cdot\|^{}_{\infty}\le \|\cdot\|^{}_{\varPhi}$, the series $\sum_j f_j$ is also absolutely convergent in the Banach space $(\COsp(G), \|\cdot\|^{}_{\infty})$. Its limit $f$ in $\COsp(G)$ then satisfies
\begin{displaymath}
\|f\|^{}_{\varPhi} \, = \, \|\sum_j f_j \|^{}_{\varPhi} \,\le\,  \sum _j \|f_j\|^{}_{\varPhi} < \infty \ ,
\end{displaymath}
meaning that $f\in \Wsp(G)$. In fact, $\sum_j f_j$ converges to $f$ in $(\Wsp(G), \|\cdot\|^{}_{\varPhi})$. This follows from the estimate
\begin{displaymath}
\begin{split}
\|f-\sum_{j\le n} f_j\|^{}_{\varPhi} \,=\,
\|\sum_{j>n} f_j\|^{}_{\varPhi} \,\le\,
\sum_{j>n} \|f_j\|^{}_{\varPhi} \ts ,
\end{split}
\end{displaymath}
which vanishes as $n\to\infty$ due to absolute convergence.

\noindent (e): We use the identity as embedding map. Obviously $\Ccsp(G)\subseteq \Wsp(G)$, and continuity follows from Lemma ~\ref{FRSS-lem:normpropw}(b). In order to show denseness, consider arbitrary $f\in \Wsp(G)$ and note that $\sum_{j\in I} f\phi_j\in \Wsp(G)$ by Eq.~\eqref{FRSS-eq:fJ}. As $f=\sum_{j\in I} f\phi_j$ holds pointwise, the claim follows.

\noindent (f): We use the identity as embedding map.  Denseness follows from denseness of $\Ccsp(G)$ in $(\COsp(G), \|\cdot\|^{}_{\infty})$ as $\Ccsp(G)\subseteq \Wsp(G)\subseteq \COsp(G)$. Also note that $(\Wsp(G), \|\cdot\|^{}_{\varPhi})$ is continuously embedded in $(\COsp(G), \|\cdot\|^{}_{\infty})$ as $\|\cdot\|^{}_{\infty}\le \|\cdot\|^{}_{\varPhi}$. The ideal property  follows from (c).

\noindent (g): Take a BUPU $\varPhi$ of size $U$. Then, for any $f\in \Wsp(G)$ we have that
\begin{displaymath}
\|f\|^{}_{1} \,\le \sum_{i\in I} \|f\phi_i\|^{}_{1} \,\le\, m_G(U)\cdot  \sum_{i\in I} \|f\phi_i\|^{}_{\infty}
\,=\, m_G(U) \cdot \|f\|^{}_{\varPhi} \,<\, \infty \ts.
\end{displaymath}
As different functions in $\Wsp(G)\subseteq \Csp(G)$ define different classes in $\Lisp(G)$, we conclude that $(\Wsp(G), \|\cdot\|^{}_{\varPhi})$ is continuously embedded into $(\Lisp(G), \|\cdot\|^{}_{1})$. Denseness follows from denseness of $\Ccsp(G)$ in $(\Lisp(G), \|\cdot\|^{}_{1})$, as we have $\Ccsp(G)\subseteq \Wsp(G)$.

\noindent To show that $f*g\in \Wsp(G)$ for $f\in \Wsp(G)$ and $g\in \Lisp(G)$, note first $f\in \Lisp(G)$, such that the convolution is well defined. Now translation boundedness of $\|\cdot\|^{}_{\varPhi}$, compare Lemma~\ref{FRSS-lem:normpropw}(c), leads to the estimate
\begin{displaymath}
\begin{split}
\|f*g\|^{}_{\varPhi}& \,=\, \sum_{i\in I} \sup_{x\in G} \left|\int_G (T_yf)(x) g(y) \phi_i(x)\, {\rm d}y\right|\\
&\,\le\, \sum_{i\in I} \int_G \|(T_yf)\phi_i\|^{}_{\infty}\cdot |g(y)| \, {\rm d}y
\,= \int_G \|T_yf\|^{}_{\varPhi} \cdot |g(y)| \, {\rm d}y \\
&\,\le\, C \cdot \|f\|^{}_{\varPhi} \cdot \|g\|^{}_{1} \ts ,
\end{split}
\end{displaymath}
where we used monotone convergence for the second equality. More conceptually, we could have used Bochner integral calculus for the vector-valued function $y\mapsto T_yg$, compare~\cite[Prop.~3.6.14]{FRSS-Rei2} or \cite[Ch.~2.3]{FRSS-R02}.
\qed \end{beweis}

\subsection{Segal algebra characterisation of $\Wsp(G)$}\label{FRSS-sec:SegAlgChar}

For the Wiener algebra $(\Wsp(G), \|\cdot\|^{}_{\varPhi})$, we can give an equivalent norm which is translation invariant and translation continuous.

\begin{lemma}\label{FRSS-lem:hnW}
Let $\varPhi$ be any BUPU in $\COsp(G)$ and define for $f\in \Wsp(G)$ the norm $\|f\|^{}_{\Wsp}=\sup_{x\in G}\|T_xf\|^{}_{\varPhi}$\index{$\| \, \|^{}_{\Wsp}$}. Then, $\|\cdot\|^{}_{\Wsp}$ is an equivalent norm on $(\Wsp(G), \|\cdot\|^{}_{\varPhi})$ that is translation invariant and translation continuous.
\end{lemma}

\begin{beweis}
As the norm $\|\cdot\|^{}_{\varPhi}$ is translation bounded by Lemma~\ref{FRSS-lem:normpropw} (c), there exists a finite constant $C$ such that
\begin{displaymath}
\|f\|^{}_{\Wsp}\,=\,\sup_{x\in G}\|T_xf\|^{}_{\varPhi} \,\le\, C\cdot  \|f\|^{}_{\varPhi}
\end{displaymath}
for any $f\in \Wsp(G)$.
This shows that $\|f\|^{}_{\Wsp}$ is finite. As $\|\cdot\|^{}_{\Wsp}$ satisfies the norm axioms and $\|\cdot\|^{}_{\varPhi}\le \|\cdot\|^{}_{\Wsp}$ by definition, this also shows that $\|\cdot\|^{}_{\Wsp}$ and $\|\cdot\|^{}_{\varPhi}$ are equivalent norms. The norm $\|\cdot\|^{}_{\Wsp}$ is translation invariant by definition, and translation continuity is inherited from $\|\cdot\|^{}_{\varPhi}$ as
\begin{displaymath}
\|T_xf-f\|^{}_{\Wsp}\,\le\,C \cdot\|T_xf-f\|^{}_{\varPhi} \ts ;
\end{displaymath}
compare Lemma~\ref{FRSS-lem:normpropw} (d).
\qed \end{beweis}

Lemma~\ref{FRSS-lem:hnW} implies that $(\Wsp(G), \|\cdot \|^{}_{\Wsp})$ is a Segal algebra. The Wiener algebra is prominently featured as a non-trivial example of a Segal algebra in Reiter's books \cite{FRSS-Rei1, FRSS-Rei2}. Let us recall its definition and elementary properties from Section 6.2 of \cite{FRSS-Rei2}; see Appendix A.3 of \cite{FRSS-Rei2} and the monograph~\cite{FRSS-Rei3} for more background.
\begin{definition}A  Banach space $(\Bsp, \|\cdot\|^{}_{\Bsp})$ is called a  \emph{Segal algebra}\index{Algebra!Segal} if it is contained in $(\Lisp(G),\|\cdot\|^{}_{1})$ as a dense subspace and is \textit{homogeneous}, meaning that for all $x \in G$ and $f \in \Bsp$ we have
\begin{displaymath}
  \|T_xf\|^{}_{\Bsp}\, = \, \|f\|^{}_{\Bsp} \qquad \text{and} \qquad
\lim_{y\to0} \|T_yf-f\|^{}_{\Bsp} \, =\, 0 \ts .
\end{displaymath}
\end{definition}

\smallskip

Segal algebras are convolution ideals in $(\Lisp(G), \|\cdot\|^{}_{1})$ that satisfy the inequality $\|g*f\|^{}_{\Bsp}\le \|g\|^{}_{1}\cdot \|f\|^{}_{\Bsp}$ for $g\in \Lisp(G)$ and $f\in \Bsp$. This can be shown as in the proof of Theorem~\ref{FRSS-thm:W}(g).
Any Segal algebra contains the integrable functions of finite bandwidth
\begin{displaymath}\index{$\LKsp(G)$}
\LKsp(G)\,=\,\{f\in \Lisp(G): \widehat f\in \Ccsp(\widehat G)\}
\end{displaymath}
as a dense subspace,
and $\LKsp(G)$ equals the intersection of all Segal algebras \cite[p.~26]{FRSS-Rei3}.
We will always consider $\LKsp(G)$ as a space of continuous functions.

\smallskip

The observation that $\Wsp(G)$ is a Segal algebra was already made in \cite[Thm.~5.1]{FRSS-LRW74}. Let us remark that the Wiener algebra $\Wsp(G)$ is characterised by  a certain minimality property \cite[Thm.~3]{FRSS-fe77-3}: Among all Segal algebras with the additional property of being a pointwise $\COsp(G)$-module it is the smallest one, meaning that it is continuously embedded into any such Segal algebra. See \cite[Thm.~5]{FRSS-Fei81b} for an extension of this result.

\smallskip

\subsection{Tensor product property of $\Wsp(G)$}\label{FRSS-sec:Pp}

Here we describe a certain product property of the Wiener algebra that will later become important. We first explain the notion of the projective tensor product of function spaces~\cite[Ch.~2.1]{FRSS-R02} in our context. Consider two Banach spaces $(\Asp, \| \cdot \|^{}_{\Asp})$ and $(\Bsp, \| \cdot \|^{}_{\Bsp})$ of functions on LCA groups $G$ and $H$, respectively. The \emph{tensor product}\index{tensor~product!of~functions} $g\otimes h$ of the functions $g \in \Asp$ and $h \in \Bsp$ is defined as the function $g \otimes h : G \times H\to \CC$ given by
\begin{displaymath}
(g \otimes h)(x,y) \,=\, g(x) \cdot  h(y) \ts . \index{$g{\otimes}h$|see{tensor~product}}
\end{displaymath}
The \emph{tensor product} $\Asp \otimes \Bsp$\index{tensor~product!of~function~spaces} of the spaces $\Asp$ and $\Bsp$ is the vector space of functions on $G \times H$ spanned by tensor products of functions from $\Asp$ and $\Bsp$, respectively. We thus have
\begin{displaymath}
\Asp \otimes \Bsp \,=\, \big\{ \sum_{j=1}^n g_j \otimes h_j : n \in \NN, g_j \in \Asp, h_j \in \Bsp \big\}   \ts .
\end{displaymath}
The so-called \emph{projective norm} $\|\cdot\|$ on $\Asp\otimes \Bsp$ is defined by
\begin{displaymath}
\|f\| \,=\,\inf\big\{\sum_{j=1}^n  \|g\|^{}_{\Asp}\cdot \|h\|^{}_{\Bsp} : f= \sum_{j=1}^n g_j \otimes h_j \big\} \ .
\end{displaymath}
It is the largest norm on $\Asp \otimes \Bsp$ such that $\| g \otimes h \| = \|g\|^{}_{\Asp}\cdot \|h\|^{}_{\Bsp}$ for all $g \in \Asp$ and $h \in \Bsp$, compare \cite[Prop.~2.1]{FRSS-R02}.

The \emph{projective tensor product}\index{projective~tensor~product}\index{tensor~product!projective} $\Asp \,\widehat{\otimes}\, \Bsp$ is defined as the abstract completion of $(\Asp \otimes \Bsp, \| \cdot \|)$. Any element $f \in \Asp \,\widehat{\otimes}\, \Bsp$ can be written as an absolutely convergent series
\[
f \, =\, \sum_{n\in \NN} g_n\otimes h_n \ts ,
\] 
for some sequences $(g_n)_{n\in\NN}$ in $\Asp$ and $(h_n)_{n\in\NN}$ in $\Bsp$, respectively \cite[Prop.~2.8]{FRSS-R02}.

\smallskip

In all examples below where we refer to the projective tensor product $\Asp \,\widehat{\otimes}\, \Bsp$, the space $(\Asp \otimes \Bsp, \| \cdot \|)$ embeds continuously into $(\COsp(G\times H), \| \cdot \|^{}_{\infty})$. We will thus always view $\Asp \,\widehat{\otimes}\, \Bsp$ as a subspace of $\COsp(G\times H)$.

\begin{lemma}\label{FRSS-lem:WW}
Let $G,H$ be LCA groups. We then have
\[
\Wsp(G)\, \widehat{\otimes}\, \Wsp(H) \,\subseteq \,  \Wsp(G\times H) \ts .
\]
\end{lemma}

\begin{beweis}
Let $\varPhi=(\phi_i)_{i\in I}$ be a BUPU in $\COsp(G)$, and let $\varPsi=(\psi_j)_{j\in J}$ be a BUPU in $\COsp(H)$.
Consider the product BUPU
\[
\varPhi\otimes \varPsi \, =\,(\phi_i \otimes \psi_j)_{(i,j)\in I\times J}
\]
in $\COsp(G\times H)$; compare  Lemma~\ref{FRSS-lem:prodbupu}.
For $g\in \Wsp(G)$ and $h\in \Wsp(H)$, we then have that
  \begin{displaymath}
    \begin{split}
      \|g\otimes &h\|^{}_{\Wsp(G\times H), \varPhi\otimes \varPsi} \,= \sum_{(i,j)\in I\times J} \|(\phi_i\otimes \psi_j)\cdot (g\otimes h)\|^{}_{\infty}\\
      &\,= \sum_{(i,j)\in I\times J} \|(\phi_i \cdot g)\otimes (\psi_j\cdot h)\|^{}_{\infty}
      \,= \sum_{(i,j)\in I\times J} \|\phi_i g\|^{}_{\infty} \cdot \|\psi_j h\|^{}_{\infty} \\[1mm]
      &\,=\, \|g\|^{}_{\Wsp(G), \varPhi} \cdot \|h\|^{}_{\Wsp(H), \varPsi} \,<\, \infty \ts .
    \end{split}
  \end{displaymath}
The above argument shows that we also have
\[
f \,=\, \sum_{n\in \NN} g_n \otimes h_n \in \Wsp(G\times H) \ts ,
\]
as long as $(g_n)_{n\in\NN}$ and $(h_n)_{n\in\NN}$ are sequences in $\Wsp(G)$ and $\Wsp(H)$, respectively, with a finite sum
\begin{displaymath}
\sum_{n\in \NN} \|g_n\|^{}_{\Wsp(G),\varPhi} \cdot\|h_n\|^{}_{\Wsp(H),\varPsi} \,<\, \infty \ts .
\end{displaymath}
We thus have $\Wsp(G) \,\widehat{\otimes}\, \Wsp(H)\subseteq \Wsp(G\times H)$.
\qed \end{beweis}

\subsection{The dual of the Wiener algebra}

Consider the dual $\Wspp(G)$\index{Wiener~algebra~dual} of the normed space $(\Wsp(G), \| \cdot\|^{}_{\varPhi})$. Then, $\Wspp(G)$ is a Banach space when equipped with the operator norm $\mu\mapsto \op{\mu}_{\Wspp,\varPhi}$\index{$\op{ \cdot }_{\varPhi}$}, given by
\begin{displaymath}
\op{\mu}_{\Wspp,\varPhi}\,:=\, \sup \left\{ |\mu(f)| : f \in \Wsp(G), \| f\|^{}_{\varPhi} \le 1 \right\} \ts .
\end{displaymath}
Note that $\left| \mu(f) \right| \leq \op{\mu}_{\Wspp,\varPhi} \cdot \| f\|^{}_{\varPhi}$ for all $f \in \Wsp(G)$, and that different BUPUs lead to equivalent norms on $\Wspp(G)$. Whenever there is no possibility for confusion, we will simply write $\op{\mu}_{\varPhi}$ instead of $\op{\mu}_{\Wspp,\varPhi}$.

It is well known \cite[Thm.~6.1]{FRSS-LRW74} that $\Wspp(G)$ can be identified with the space of complex Radon measures that are translation bounded. This will be discussed in Section~\ref{FRSS-sec:m}, see Theorem~\ref{FRSS-thm:tbW} below. For later reference, we note the following result on convolution.

\begin{lemma}\label{FRSS-lem:convcu}
For $\mu\in \Wspp(G)$ and $f\in \Wsp(G)$, define $\mu*f: G\to \CC$ by $(\mu * f)(x)=\mu(T_xf^\dag)$ for $x\in G$. Then $\mu*f\in \Cusp(G)$.
Moreover, for every BUPU $\varPhi$, there exists a constant $C_{\varPhi}$ such that
\[
\| \mu * f \|^{}_{\infty} \,\leq\, C_\varPhi \cdot\op{\mu}_{\varPhi}\cdot \|f\|^{}_{\varPhi}
\]
holds for all $\mu \in \Wspp(G)$ and all $f \in \Wsp(G)$.
\end{lemma}
\begin{beweis}
Consider a BUPU $\varPhi=(\phi_i)_{i\in I}$ of size $U$ and norm $M$ with point family $X$. Boundedness follows from the estimate
\begin{displaymath}
\begin{split}
|(\mu * f)(x)|\,&=\,|\mu(T_x f^\dag)|\,\le\, \op{\mu}_{\varPhi} \cdot \|T_xf^\dag\|^{}_{\varPhi}\\[1mm]
\,&\le\,   C_1\cdot \op{\mu}_{\varPhi} \cdot \|f^\dag\|^{}_{\varPhi} \,\le\, C_1 C_2\cdot  \op{\mu}_{\varPhi} \cdot \|f\|^{}_{\varPhi} \ts .
\end{split}
\end{displaymath}
The constants can be chosen as $C_1=M\cdot \|\delta^{}_{X}\|^{}_{U-U}$, compare Lemma~\ref{FRSS-lem:normpropw} (c), and as $C_2=M\cdot \|\delta^{}_{X}\|^{}_{-U-U}$, compare Theorem~\ref{FRSS-thm:W} (b) and (a).
Uniform continuity follows from translation continuity of $\|\cdot\|^{}_{\varPhi}$; see Lemma~\ref{FRSS-lem:normpropw}.
\qed \end{beweis}

In fact, we have the following converse to Lemma~\ref{FRSS-lem:convcu}:

\begin{corollary} Let $L: \Wsp(G) \to \Cusp(G)$ a bounded operator which commutes with the translation operator. Then, there exists some $\mu \in \Wspp(G)$ such that for all $f \in \Wsp(G)$ we have $L(f) =  \mu*f$.
\end{corollary}
\begin{beweis} Define $\mu :\Wsp(G) \to \CC$ via $\mu(f) =  L(f^\dag)(0)$.
Boundedness of $L$ implies that $\mu$ is a bounded linear functional, and hence $\mu \in \Wspp(G)$. Consider now any $f \in \Wsp(G)$. We then have
\[
(\mu*f)(x) \,=\, \mu(T_xf^\dag) \, = \, L( T_{-x}f)(0) \,=\, (T_{-x} L(f))(0) \,=\, L(f)(x) 
\]
for arbitrary $x \in G$. Thus $L(f) =  \mu*f$.
\qed \end{beweis}

\subsection{Atomic decomposition and continuous norms on $\Wsp(G)$}\label{FRSS-ssec:adcn}

Let us abbreviate $\Csp^{}_{\mkern-4mu K}(G)=\{f\in \Ccsp(G): \supp(f)\subseteq K\}$.
Fix any relatively compact $0$-neighbourhood $U\subseteq G$ and consider the space $\Wasp(G)$ defined by
\begin{displaymath}
\Wasp(G)=\big\{f\in \COsp(G): f=\sum_i f_i  \text{ with } f_i \in \Csp^{}_{\mkern-4mu x_i+U}(G), \sum_i\|f_i\|_\infty<\infty\big\} \ts .
\end{displaymath}
Exactly as in Remark~\ref{FRSS-rem:Wiscountable}, for any such decomposition $f= \sum_{i \in I} f_i$,
we can assume without loss of generality that $I$ is countable.

\begin{lemma}\label{FRSS-lem:Wa}
We have $\Wsp(G)=\Wasp(G)$. In particular, the definition of $\Wasp(G)$ is independent of the choice of $U$.
Moreover, if $\varPhi$ is any BUPU of size $U$, then there exists a constant $C^{}_{\varPhi}$ such that for any
$f=\sum_i f_i\in \Wasp(G)$ with $f_i \in \Csp^{}_{\mkern-4mu x_i+U}(G)$ we have
\begin{equation}\label{FRSS-eq:WWa-est}
\|f\|^{}_\varPhi \,\le \,  C^{}_{\varPhi} \sum_i \|f^{}_i\|^{}_\infty\ts .
\end{equation}
\end{lemma}

\begin{beweis}
Let $\varPhi$ be any BUPU of size $U$. Note first that $f\in \Wsp(G)$ can be written as $f=\sum_{i\in I} f_i$ where $f_i=f\phi_i$. Thus $\sum_i \|f^{}_i\|^{}_\infty=\|f\|^{}_{\varPhi}$, and $\Wsp(G)\subseteq \Wasp(G)$ follows.

\noindent In order to show $\Wasp(G)\subseteq \Wsp(G)$, recall from  Lemma~\ref{FRSS-lem:normpropw}  that there exists a finite constant $C$ such that $\|g\|^{}_{\varPhi} \,\le\, C \|g\|^{}_\infty$ for any $g\in \Ccsp(G)$ that is supported within some translate of $U$.
Let now $f=\sum_i f_i\in \Wasp(G)$ with $f_i \in \Csp^{}_{\mkern-4mu x_i+U}(G)$. Then,
%
\begin{displaymath}
\|f\|^{}_\varPhi \,\le \, \sum_i \|f^{}_i\|^{}_\varPhi \, \le \, C \sum_i \|f^{}_i\|^{}_\infty\ts .
\end{displaymath}
We conclude that $f\in \Wsp(G)$. Independence of $U$ holds since the definition of $\Wsp(G)$ is independent of the size $U$ of the BUPU by Theorem~\ref{FRSS-thm:W} (a).
\qed \end{beweis}

There is a natural norm $\|\cdot\|^{}_a$ associated to $\Wasp(G)$, which is for $f\in \Wasp(G)$ given by
\begin{displaymath}
\|f\|^{}_a=\inf\big\{\sum_i \|f_i\|_\infty: f=\sum_i f_i \text{ and } f_i \in \Csp^{}_{\mkern-4mu x_i+U}(G) \big\} \ts .
\end{displaymath}

\begin{proposition}
Let $\varPhi$ be any BUPU in $\COsp(G)$. Then $\|\cdot\|^{}_a$ is an equivalent homogeneous norm on $(\Wsp(G), \|\cdot\|^{}_{\varPhi})$.
\end{proposition}

\begin{beweis}
We first check the norm axioms.
Consider any $f \in \Wsp(G)$ and let $f_i \in \Csp^{}_{\mkern-4mu x_i+U}(G)$ be so that $f = \sum_{i} f_i$.
Then $\| f \|_\infty \, \le \,\|f\|^{}_a$ follows from $\|f \|_\infty \le \sum_{i} \| f_i \|_\infty$. 
In particular, from $f \neq 0$ we get $\| f \|^{}_a >0$, which shows positive definiteness. 
Next, for any $f \in \Wasp(G)$ and $c \in \CC$ we clearly have $\|cf\|^{}_a =  |c| \cdot \|f\|^{}_a$, which shows absolute homegeneity.
Consider now arbitrary $f,g\in \Wasp(G)$ and choose $\varepsilon>0$ and decompositions $f=\sum_i f_i$ and $g=\sum_j g_j$ such that both $\sum_i \|f^{}_i\|^{}_\infty\le \|f\|^{}_a+\varepsilon$ and $\sum_j \|g^{}_j\|^{}_\infty\le \|g\|^{}_a+\varepsilon$. We then have %
\begin{displaymath}
\|f+g\|^{}_a\,\le\, \sum_i \|f^{}_i\|^{}_\infty+\sum_j \|g^{}_j\|^{}_\infty\,\le\, \|f\|^{}_a+\|g\|^{}_a+2\varepsilon \ts .
\end{displaymath}
As $\varepsilon>0$ was arbitrary, this shows the triangle inequality. Thus $\|\cdot\|^{}_a$ is a norm on $\Wasp(G)=\Wsp(G)$. 

\noindent Translation invariance $\|T_xf\|^{}_a=\|f\|^{}_a$ follows from translation invariance of $\|\cdot\|^{}_\infty$.
In order to show translation continuity, let $\varPhi$ be any BUPU. For all $f \in \Wsp(G)$ we then clearly have
\[
\|f\|^{}_a \, \le \, \sum_i \| f \phi_i\|_\infty \, =\, \|f\|^{}_{\varPhi} \ts .
\]
Moreover, Eq.~\eqref{FRSS-eq:WWa-est} gives $\|f\|^{}_{\varPhi} \le C^{}_{\varPhi} \cdot \|f\|^{}_a$.
Thus the norms $\|\cdot\|^{}_a$ and $\|\cdot\|^{}_{\varPhi}$ are equivalent. In particular, translation continuity of $\|\cdot\|^{}_a$ is inherited from translation continuity of $\|\cdot\|^{}_{\varPhi}$; compare Lemma~\ref{FRSS-lem:normpropw}(d).
\qed \end{beweis}

We discuss continuous norms on $\Wsp(G)$ next.
For $f,g\in \COsp(G)$ define their \emph{control function} $\kappa(f,g): G\to [0,\infty)$ for $z\in G$ by
\begin{displaymath}
\kappa(f,g)(z)\,=\,\|f T_zg\|^{}_\infty \ts .
\end{displaymath}
The control function satisfies the following properties, which are immediate consequences of its definition.
\begin{lemma}\label{FRSS-Lem:cont-prop} For $f,g,h \in \COsp(G)$ the following hold.
\begin{itemize}
 \item[(a)] $\kappa(cf,g) =  \kappa(f, cg)  =  |c|\cdot\kappa(f, g)$ for every $c\in \CC$,
  \item[(b)] $\kappa(f+h,g)\, \leq \, \kappa(f, g) \, + \,  \kappa(h, g)$,
  \item[(c)] $\kappa(g,f)^\dag \, = \, \kappa(f,g)$,
  \item[(d)] $\kappa(T_xf,g) \, =\, \kappa(f,T_{-x}g)\, = \, T_{x}\kappa(f,g)$ for every $x\in G$,
  \item[(e)] $\|\kappa(T_xf,g)\|^{}_1 \,=\, \|\kappa(f,g)\|^{}_1\,=\,\|\kappa(g,f)\|^{}_1$ for every $x\in G$,
  \item[(f)] $\kappa(f,gh)\le \|h\|_\infty \cdot \kappa(f,g)$,
  \item[(g)] $\kappa(|f|, |g|)\,= \, \kappa(f,g)$,
  \item[(h)] $\kappa(f,g) \, \le \,   \kappa(h,g)$ if $|f| \leq |h|$. \qed
\end{itemize}
\end{lemma}
\noindent The following lemma shows, in particular, that  $\|\kappa(f,g)\|^{}_1$ is finite for any $f,g\in \Wsp(G)$.

\begin{lemma}\label{FRSS-lem:contnorm} Let $f,g \in \Wsp(G)$. Let $\varPhi, \varPsi$ be any BUPUs in $\COsp(G)$ of size $U$ and $V$, respectively. We then have
\begin{displaymath}
\|f\|^{}_{1} \cdot \|g\|^{}_\infty  \, \le \,  \|\kappa(f,g)\|^{}_1 \,\le\,  m_G(U-V) \cdot \|f\|^{}_\varPhi \cdot\|g\|^{}_\varPsi \ts .
\end{displaymath}
\end{lemma}
\begin{beweis}
First, fix some $x \in G$. Note that by definition of the control function we have
\[
|f(x+z)g(x)| \, \leq \,  \kappa(f,g)(z) \ts 
\]
for arbitrary $z\in G$. Thus, integration with respect to $z$ yields
\[
\|f\|^{}_1 \cdot |g(x)| \, \le \,  \| \kappa(f,g)\|^{}_1 .
\]
Taking the supremum over $x \in G$, we get the lower inequality.

\noindent For the upper inequality, let $\varPhi=(\phi_i)_{i \in I}$ and $\varPsi = (\psi_j)_{j \in J}$ be any two BUPUs in $\COsp(G)$ of sizes $U$ and $V$, respectively. For each $i\in I$ and $j \in J$, let $x_i,y_j \in G$ be such that $\supp(\phi_i) \in x_i+U$ and $\supp(\psi_j) \subseteq y_j+U$. Then, for all $z \notin (x_i+U)-(y_j+V)$ we have 
\begin{displaymath}
(f\phi_i) T_z(g \psi_j) \,=\, 0\ts .
\end{displaymath}
Thus, the control function on $f\phi_i$ and $g \psi_j$ satisfies
\begin{displaymath}
\begin{split}
\| \kappa(f\phi_i,g\psi_j)\|^{}_1
\,&=\, \int_{x_i-y_j+U-V} \|(f\phi_i) T_z(g \psi_j) \|^{}_{\infty} \, {\rm d}z \\[1mm]
\, &\le \, m_G(U-V) \cdot \|(f\phi_i)\|^{}_\infty \cdot\| g \psi_j \|^{}_\infty  \ts .
\end{split}
\end{displaymath}
Now we can estimate
\begin{displaymath}
\begin{split}
\|\kappa(f,g)\|^{}_1 \,&\le \, \sum_{i\in I}  \sum_{j \in J}\|\kappa(f\phi_i,g\psi_j)\|^{}_1 \, \le \, m_G(U-V)  \sum_{i\in I}\sum_{ j \in J} \|f\phi_i\|^{}_\infty \cdot \|g\psi_i\|^{}_\infty \\[1mm]
\, &= \, m_G(U-V)\cdot \|f\|^{}_\varPhi  \cdot \|g\|^{}_\varPsi \ts ,
\end{split}
\end{displaymath}
which shows the upper inequality.
\qed \end{beweis}

\begin{proposition}\label{FRSS-prop:contnorm}
Let $g\in \Wsp(G)\setminus \{0\}$. We then have
\begin{displaymath}
\Wsp(G)=\{f\in \COsp(G): \|\kappa(f,g)\|^{}_1<\infty\} \ts .
\end{displaymath}
In fact, $\|\kappa(\cdot,g)\|^{}_1$ is an equivalent norm on $(\Wsp(G), \|\cdot\|^{}_{\Wsp})$.
\end{proposition}

\begin{beweis}
Let $g\in \Wsp(G)\setminus \{0\}$ be arbitrary, and let $\varPhi=(\phi)_{i\in I}$ be any BUPU in $\COsp(G)$  of size $U$ with point family $X=(x_i)_{i\in I}$. Recalling Lemma~\ref{FRSS-Lem:cont-prop} and Lemma~\ref{FRSS-lem:contnorm}, it is easily seen that $\|\kappa(\cdot,g)\|^{}_1$ defines a norm on $\Wsp(G)$. Translation invariance is obvious, and translation continuity is inherited from $\|\cdot\|^{}_\varPhi$ by Lemma~\ref{FRSS-lem:contnorm}; compare Lemma~\ref{FRSS-lem:normpropw}(d).

\noindent We now argue that the norm $\|\kappa(\cdot,g)\|^{}_1$ is equivalent to the norm $\|\cdot\|^{}_\varPhi$, which yields the claim by Lemma~\ref{FRSS-lem:hnW}.
Due to Lemma~\ref{FRSS-lem:contnorm}, we only need to give a finite constant $C^{}_{\varPhi}$ such that 
\begin{equation}\label{FRSS-eq:kapfk}
\|f\|^{}_{\varPhi} \le C^{}_{\varPhi} \cdot  \|\kappa(f,g)\|^{}_1
\end{equation}
for all $f\in \Wsp(G)$.

\noindent
By Lemma~\ref{FRSS-Lem:cont-prop}(f), 
we can assume without loss of generality that the set $K=\supp(g)$ is compact. Using Lemma~\ref{FRSS-Lem:cont-prop}(f) with $h= \overline{g}$, we can also assume without loss of generality that $g \geq 0$. For ease of presentation, we further assume $\|g\|_1=1$.

\noindent Take now $p\in \Ccsp(G)$ such that $p\ge0$ and $p|_{U-U-K}=1$. Then, the function $h=p*g\in \Ccsp(G)$ satisfies $h|_{U-U}=1$ by construction.
We thus have $\phi_i=\phi_i\cdot T_xh$ for $x\in x_i+U$, due to $(T_xh)|^{}_{\supp(\phi_i)}=1$ for $x\in x_i+U$. Also note that we have 
\begin{displaymath}
\card(\{i\in I: x\in x_i+U\})\,\le\, \|\delta^{}_{X}\|^{}_{U-U}
\end{displaymath}
for all $x\in G$ by Lemma~\ref{FRSS-lem:BUPU-tb}. We can thus estimate
\begin{displaymath}
\begin{split}
m(U)\,\|f\|^{}_\varPhi \,&=\,\int_G \sum_{i\in I} \|f\phi_i\|^{}_\infty 1_{x_i+U}(x) \,{\rm d}x
\,\le \,M \!\!\int_G \sum_{i\in I} \|fT_xh\|^{}_\infty 1_{x_i+U}(x) \,{\rm d}x \\
\,&\le \, M \|\delta^{}_{X}\|^{}_{U-U}\int_G  \|fT_xh\|^{}_\infty \, {\rm d}x
\,=\,  M \|\delta^{}_{X}\|^{}_{U-U}  \cdot \|\kappa(f,h)\|^{}_1 \ts .
\end{split}
\end{displaymath}
We can now write $h=p*g=\int p(y) \, T_yg\,{\rm d}y$ and use elementary estimates to obtain
\begin{displaymath}
\kappa(f, h)\,\le \, \int p(y)\, \kappa(f, T_yg)\,{\rm d}y \,=\, \int p(y) \, T_{-y}\kappa(f,g)  \, {\rm d}y\,=\, p^\dag*\kappa(f,g)  \ts .
\end{displaymath}
Note that the above estimate follows immediately using Bochner integral calculus for the vector-valued function $y\mapsto T_y g$, compare \cite[Prop.~3.6.14]{FRSS-Rei2} or \cite[Ch.~2.3]{FRSS-R02}.

\noindent The above estimate implies $\|\kappa(f, h)\|_1\le \|p\|_1\cdot \|\kappa(f, g)\|_1$. 
We conclude that Eq.~\eqref{FRSS-eq:kapfk} holds with $C^{}_\varPhi=M \|\delta^{}_{X}\|^{}_{U-U}\|p\|_1/m(U)$, and the claim follows.
\qed \end{beweis}

\section{Feichtinger's algebra}\label{FRSS-sect:appB}

In this section, we discuss the Feichtinger algebra $\So(G)$. This function algebra has been introduced in \cite{FRSS-Fei81} and studied intensively ever since; see \cite{FRSS-Jak} for a recent review. We will introduce $\So(G)$ using BUPUs in the Fourier algebra $\Asp(G)$, in analogy to our treatment of the Wiener algebra.

\subsection{BUPUs for Banach algebras}

We recall the general definition of bounded uniform partition of unity from \cite{FRSS-F83}.  In the sequel, boundedness is understood with respect to a suitable Banach algebra $(\Asp,\|\cdot\|^{}_{\Asp})$ of continuous functions over $G$ (with respect to pointwise multiplication).

\begin{definition}[BUPU in $\Asp$]\label{FRSS-def:BUPUA}
Let $(\Asp,\|\cdot\|^{}_{\Asp})$ be a pointwise Banach algebra of continuous functions on $G$. An indexed family $\varPhi=(\phi_i)_{i\in I}$ in $A$ is called a \emph{bounded uniform partition of unity ($\Asp$-BUPU)}\index{BUPU!in~$\Asp$} of  norm $M$, size $U$ and overlap constant $B$, if there exist points $X=(x_i)_{i\in I}$ in $G$ and a relatively compact $0$-neighbourhood $U\subseteq G$ such that
\begin{itemize}
\setlength\itemsep{2pt}
\item[(a)] $\supp(\phi_i)\subseteq x_i +U$ for all $i\in I$,
\item[(b)] $\card\{j\in I: (x_i +U)\cap (x_j+U)\ne \varnothing\}\le B<\infty$ for all $i\in I$,
\item[(c)] $\sum_{i\in I} \phi_i(x)=1$ for all $x\in G$,
\item[(d)] $\|\phi_i\|^{}_{\Asp}\le M<\infty$ for all $i \in I$.
\end{itemize}
\end{definition}

Instead of working with the general definition of a suitable Banach algebra, which we will not give here, we restrict to two examples in this chapter.
For the Wiener algebra, we have considered $(\Asp,\|\cdot\|^{}_{\Asp}) =(\COsp(G), \|\cdot\|^{}_{\infty})$ in the previous section, where $\COsp(G)$ denotes the space of continuous functions on $G$ vanishing at infinity, and where $\|\cdot\|^{}_{\infty}$ denotes the supremum norm.

\smallskip

For Feichtinger's algebra, we will consider the space $(\Asp(G), \|\cdot\|^{}_{\Asp(G)})$, where $\Asp(G)=\mathcal F\Lisp(\widehat G)$ is the Fourier algebra\index{Algebra!Fourier} of $G$, together with the norm $\|f\|^{}_{\Asp(G)}=\|\widecheck f\|^{}_{\Lisp(\widehat G)}$\index{$\| \, \|^{}_{\Asp(G)}$}.
An important subspace of the Fourier algebra is given by
\begin{displaymath}\index{$\KLsp(G)$}
\KLsp(G)\,=\,\{ f \in \Ccsp(G) : \widehat{f} \in \Lisp(\widehat{G}) \}\,=\,\Ccsp(G)\cap \Asp(G) \ts .
\end{displaymath}
Let us list some properties of the Fourier algebra; compare~\cite[Ch.~1]{FRSS-B75} for a more comprehensive account.
\begin{lemma}\label{FRSS-lem:AGC0}
The Fourier algebra $\Asp(G)$ is dense in $(\COsp(G), \|\cdot\|^{}_{\infty})$. In fact, the space
\begin{displaymath}
\Ksp^{}_{\!2}(G) \, = \, \text{\rm span}\{f*\widetilde{f}:f\in \Ccsp(G)\}
\end{displaymath}
satisfies $\Ksp^{}_{\!2}(G)\subseteq \KLsp(G)$ and is dense in $(\COsp(G), \|\cdot\|^{}_{\infty})$.
Moreover, the Fourier algebra $(\Asp(G), \|\cdot\|^{}_{\Asp(G)})$ is continuously embedded in $(\COsp(G), \|\cdot\|^{}_{\infty})$.
\end{lemma}

\begin{beweis}
We have $\Asp(G)\subseteq \COsp(G)$ by the Riemann--Lebesgue lemma and $\|f\|^{}_{\infty}\le \|f\|^{}_{\Asp(G)}$ for $f\in \Asp(G)$. The identity as embedding map is thus clearly continuous.

\noindent We also have $\Ksp^{}_{\!2}(G)\subseteq \Asp(G)$. To see this, consider any $f,g\in \Ccsp(G)$ and note that $\widehat{f*g}=\widehat f \cdot \widehat g$. We have $\widehat f \cdot \widehat g\in \Lisp(\widehat G)$, which follows from $\widehat f, \widehat g\in \Ltsp(\widehat G)$ due to Plancherel's theorem \cite[Thm.~4.4.14]{FRSS-Rei2}. Thus $f*g\in \Asp(G)$ results from Fourier inversion \cite[Thm.~4.4.5]{FRSS-Rei2}.
Altogether we have
\[
\Ksp^{}_{\!2}(G) \, \subseteq \, \Ccsp(G)\cap \Asp(G)=\KLsp(G) \,.
\]
We argue that $\Ksp^{}_{\!2}(G)$ is dense in $(\Ccsp(G), \|\cdot\|^{}_{\infty})$. Let $f\in \Ccsp(G)$ be arbitrary.
Take a Dirac net $(v_i)_i$ in $G$, that is, $v_i\in \Ccsp(G)$ is non-negative, $\|v_i\|^{}_{1}=1$ and the support of $v_i$ is contained in some $0$-neighbourhood which shrinks to $0$ with increasing $i$; compare \cite[Lemma~1.6.5]{FRSS-DE}.  Then, $(f*v_i)_i$ converges uniformly to $f$; compare \cite[Lemma~3.4.5]{FRSS-DE}. Moreover $f*v_i\in \Ksp^{}_{\!2}(G)$ by polarisation \cite[p.~244]{FRSS-MoSt}. Thus $\Ksp^{}_{\!2}(G)$ is dense in $(\Ccsp(G), \|\cdot\|^{}_{\infty})$. Since $\Ccsp(G)$ is dense in $(\COsp(G), \|\cdot\|^{}_{\infty})$, it follows that $\Ksp^{}_{\!2}(G)$ is dense in $(\COsp(G), \|\cdot\|^{}_{\infty})$.
\qed \end{beweis}

\begin{remark}\label{FRSS-absAsp}
Note that $f \in \Asp(G)$ does not imply $|f| \in \Asp(G)$. Indeed, Kahane has given an explicit counterexample on the torus $\TT=\RR/ 2 \pi \ZZ$; see \cite{FRSS-Kah} and  \cite[Footnote~7]{FRSS-RudArt}, which satisfies $f \in \KLsp(\TT)$ but $|f| \notin \KLsp(\TT)$. \exend
\end{remark}

\subsection{Feichtinger's algebra}

As any BUPU in $\Asp(G)$ is a BUPU in $\COsp(G)$ due to $\|\cdot\|^{}_{\infty}\le \|\cdot\|^{}_{\Asp(G)}$, it is natural to consider the following subspace of $\Wsp(G)$.

\begin{definition}\label{FRSS-def:FeichAlg} Let $\varPhi=(\phi_i)_{i\in I}$ be any BUPU in $\Asp(G)$. The \emph{Feichtinger algebra}\index{Algebra!Feichtinger}\index{Feichtinger~algebra} is given by
\[\index{$\| \, \|^{}_{\So,\varPhi}$}
\So(G)\,=\, \{ f \in \Asp(G) : \|f \|^{}_{\So,\varPhi} \,= \sum_{i \in I} \|f \cdot \phi_i \|^{}_{\Asp(G)} \,<\, \infty \} \ts .
\]
\end{definition}
It is clear that $\|\cdot\|^{}_{\So, \varPhi}$ defines a norm on $\So(G)$, which is isometrically invariant under multiplication with continuous characters, i.e., we have
\[
\|\chi\cdot f\|^{}_{\So,\phi} \, = \, \|f\|^{}_{\So,\phi}
\]
for all $f\in \So(G)$ and all $\chi\in \widehat G$. In \cite{FRSS-Fei81}, the term \emph{strong character invariance} is used for this property.
Moreover, for $f \in \So(G)$ we clearly have
\[
\|f\|^{}_{\Wsp,\phi} \, \leq \, \|f\|^{}_{\So,\phi}  \ts .
\]
In particular, we indeed have $\So(G) \subseteq  \Wsp(G)$.

\begin{remark} In general, $\So(G)$ is strictly smaller than $\Wsp(G)$. Indeed, we have $\Ccsp(G) \subseteq \Wsp(G)$. Moreover, it is easy to see from the definition that
\[
\Ccsp(G) \cap \So(G) \,=\,  \KLsp(G) \ts ,
\]
which is in general strictly included in $\Ccsp(G)$. On the other hand, for a discrete group $G$ it is easy to argue that $\Wsp(G)= \Lisp(G)=\So(G)$. \exend
\end{remark}

We list some properties of $\|\cdot\|^{}_{\So, \varPhi}$ in the next lemma, whose proof is exactly as in Lemma~\ref{FRSS-lem:normpropw}. For translation continuity, one uses that the topology on $\widehat G$ is the topology of convergence on compact sets.

\begin{lemma}\label{FRSS-lem:normso}
\begin{itemize}
\item[(a)] The norm $\|\cdot\|^{}_{\So,\varPhi}$ satisfies $\|\cdot\|^{}_{\Asp(G)}\le \|\cdot\|^{}_{\So,\varPhi}$.
\item[(b)] Consider any relatively compact $0$-neighbourhood $V\subseteq G$. Then there exists a finite constant $C$ such that $\|f \|^{}_{\So,\varPhi}\le C \|f \|^{}_{\Asp(G)}$ for all $f\in \So(G)$ supported within some translate of $V$.
\item[(c)]  The norm $\|\cdot\|^{}_{\So,\varPhi}$ is translation bounded, meaning that there exists a finite constant $C$ such that $\|T_x f\|^{}_{\So,\varPhi}\le C \|f\|^{}_{\So,\varPhi}$ holds for all $f\in \So(G)$. In particular, $T_xf\in \So(G)$ for all $f\in \So(G)$.
\item[(d)]  The norm $\|\cdot\|^{}_{\So,\varPhi}$ is translation continuous, meaning that we have $\lim_{x\to0} \|T_xf-f\|^{}_{\So,\varPhi}=0$ for every $f\in\So(G)$. \qed
\end{itemize}
\end{lemma}

The following result describes important relations to $\Wsp(G)$ and $\COsp(G)$.
\begin{proposition}\label{FRSS-prop:S0Wdense}
$(\So(G), \|\cdot\|^{}_{\So,\varPhi})$ is dense and continuously embedded in both $(\Wsp(G), \|\cdot\|^{}_{\Wsp,\varPhi})$ and $(\COsp(G), \|\cdot\|^{}_{\infty})$.
\end{proposition}
\begin{beweis}
We have argued above that $\So(G)\subseteq \Wsp(G)\subseteq \COsp(G)$.  As we have 
$\|\cdot\|^{}_{\Wsp} \le  \|\cdot\|^{}_{\So}$
and $\|\cdot\|^{}_{\infty} \le  \|\cdot\|^{}_{\Asp}\le \|\cdot\|^{}_{\So}$, the identity map provides continuous embeddings.

\noindent Denseness of $\So(G)$ in $(\Wsp(G),\|\cdot\|^{}_{\Wsp})$ follows from denseness of $\Ccsp(G)$ in $(\Wsp(G), \|\cdot\|^{}_{\Wsp})$ and denseness of $\Ksp^{}_{\!2}(G)\subseteq \Wsp(G)$ in $(\COsp(G), \|\cdot\|^{}_{\infty})$. Indeed, consider arbitrary $f\in \Wsp(G)$ and let $\varepsilon>0$ be arbitrary. Take $g\in \Ccsp(G)$ such that
$\|f-g\|^{}_{\Wsp}<  \varepsilon$,
and let $K=\supp(g)+U$ for some $0$-neighbourhood $U$. Then, there exists a positive finite constant $C$ such that $\|h\|^{}_{\Wsp}\le C\|h\|^{}_{\infty}$  for all $h\in \Ccsp(G)$ such that $\supp(h)\subseteq K$. Take now some $h\in \Ksp^{}_{\!2}(G)$ such that $\supp(h)\subseteq K$ and $\|g-h\|^{}_{\infty}<\varepsilon/C$. We then have
\begin{displaymath}
\|f-h\|^{}_{\Wsp}\,\le\, \|f-g\|^{}_{\Wsp}+\|g-h\|^{}_{\Wsp} \,\le\, \|f-g\|^{}_{\Wsp} + C\|g-h\|^{}_{\infty} \,\le\, 2 \varepsilon \ts .
\end{displaymath}
As $\varepsilon>0$ was arbitrary, this shows denseness in $(\Wsp(G),\|\cdot\|^{}_{\Wsp})$.

\noindent Denseness in $(\COsp(G), \|\cdot\|^{}_{\infty})$ holds because $\KLsp(G)\subseteq \So(G)$ is dense in $(\COsp(G), \|\cdot\|^{}_{\infty})$ by  Lemma~\ref{FRSS-lem:AGC0}.
\qed \end{beweis}
The following result is parallel to the Wiener algebra case in Theorem~\ref{FRSS-thm:W}.
\begin{theorem}\label{FRSS-thm:S0}
The following properties hold for $\So(G)$.
  \begin{itemize}
  \setlength\itemsep{2pt}
   \item[(a)] Different choices of $\varPhi$ define equivalent norms on $\So(G)$. In particular, the definition of $\So(G)$ is independent of the choice of $\varPhi$.
   \item[(b)] $f \in \So(G)$ implies $\overline{f}, f^\dag, \widetilde{f}, T_t f, \chi f \in \So(G)$
   for arbitrary $t \in G$ and $\chi \in \widehat{G}$.
   \item[(c)] $(\So(G), \| \cdot \|^{}_{\So,\varPhi})$ is a Banach algebra with respect to both pointwise multiplication and convolution.
\item[(d)]   $\KLsp(G)=\Ccsp(G)\cap \So(G)$, equipped with the $\|\cdot\|^{}_{\Asp(G)}$-inductive limit topology, is dense and continuously embedded in $(\So(G), \|\cdot\|^{}_{\So, \varPhi})$.
 \item[(e)]  The Banach algebra $(\So(G), \|\cdot\|^{}_{\So,\varPhi})$ is dense and continuously embedded in $(\Asp(G), \|\cdot\|^{}_{\Asp(G)})$. Moreover, it is a Banach ideal in $(\Asp(G), \|\cdot\|^{}_{\Asp(G)})$ with respect to pointwise multiplication.
 \item[(f)]   The Banach algebra $(\So(G), \|\cdot\|^{}_{\So,\varPhi})$ is dense and continuously embedded in $(\Lisp(G), \|\cdot\|^{}_{1})$. Moreover, it is a Banach ideal in $(\Lisp(G), \|\cdot\|^{}_{1})$ with respect to convolution.
  \end{itemize}
\end{theorem}

\begin{remark}
Note that, in contrast to the Wiener algebra case, $f \in \So(G)$ does not imply $|f| \in \So(G)$, see Remark~\ref{FRSS-absAsp}.  \exend
\end{remark}

\begin{beweis}[Sketch of proof]
This is shown as in the proof of Theorem \ref{FRSS-thm:W}, where $\|\cdot\|^{}_{\infty}$ is replaced by $\|\cdot\|^{}_{\Asp(G)}$.

\noindent For (d), note that $\KLsp(G)\subseteq \So(G)$ since $\|f\|^{}_{\So,\varPhi}$ has only finitely many non-zero summands for $f\in \KLsp(G)$.

\noindent For (e), use that $\KLsp(G)$ is dense in $(\Asp(G), \|\cdot\|^{}_{\Asp(G)})$. This holds because $\LKsp(G)$ is dense in $(\Lisp(G), \|\cdot\|^{}_{1})$; compare~\cite[Prop.~5.4.1]{FRSS-Rei2}.

\noindent For (f), observe  that $\So(G)\subseteq \Wsp(G)\subseteq \Lisp(G)$ and recall $\|\cdot\|^{}_{\infty}\le \|\cdot\|^{}_{\Asp(G)}$. For denseness, use that $\Ccsp(G)$ is dense in $(\Lisp(G), \|\cdot\|^{}_{1})$ and that $\KLsp(G)$ is dense in $(\COsp(G), \|\cdot\|^{}_{\infty})$ by Lemma~\ref{FRSS-lem:AGC0}.
\qed \end{beweis}

\begin{remark}\label{FRSS-rem:fce}
(a)
In fact, Feichtinger's algebra is dense and continuously embedded in $(\Lsp^p(G), \|\cdot\|^{}_{p})$ for all $1 \le p < \infty$; compare~\cite[Lemma~4.19]{FRSS-Jak}. It can be identified with the Wiener amalgan space $\WFLili(G)$, as detailed in \cite[Rem.~6]{FRSS-Fei81} and \cite{FRSS-Fei81b}.

\noindent (b) Feichtinger's algebra contains the Schwartz--Bruhat space $\Scsp(G)$\index{Schwartz--Bruhat~function}, the test functions for tempered distributions\index{distribution!tempered} on $G$, as a dense subspace. This has been shown in \cite[Satz~3]{FRSS-P80} for Euclidean space and in \cite[Thm.~9]{FRSS-Fei81} for general LCA groups. 
In fact, the identity map $i:\Scsp(G)\to \So(G)$ provides a continuous embedding in $\So(G)$; see \cite[{§}~2.14]{FRSS-Rei4}, and the compactly supported Bruhat--Schwartz functions are dense in $\So(G)$, as shown in \cite{FRSS-Rei5}.
Distribution theory based on the Banach space $\So(G)$ is less technical as distribution theory based on the Schwartz--Bruhat space, see below.

\noindent (c) We have $\So(G)\, \widehat \otimes \, \So(H)=\So(G\times H)$ for the projective tensor product between $\So(G)$ and $\So(H)$; compare Lemma~\ref{FRSS-lem:WW}.
Whereas the inclusion ``$\subseteq$'' is seen as in the Wiener algebra case, the inclusion ``$\supseteq$'' has been proved using the Segal algebra characterisation of $\So(G)$ below; see \cite[Thm.~7]{FRSS-Fei81} or \cite[Thm.~7.4]{FRSS-Jak}.
 \exend
\end{remark}

\subsection{Segal algebra characterisation of $\So(G)$}\label{FRSS-sec:sac}

The following result is proved as in the previous section.
\begin{lemma}
Let $\varPhi=(\phi_i)_{i\in I}$ be any BUPU in $\Asp(G)$. For $f\in \So(G)$ define
\[
\|f\|^{}_{\So}\,=\, \sup\big\{ \|T_xf\|^{}_{\So,\varPhi}: x\in G\big\} \ts .
\]
Then, $\|\cdot\|^{}_{\So}$\index{$\|\cdot\|^{}_{\So}$} is an equivalent homogeneous norm on $(\So(G), \|\cdot\|^{}_{\So,\varPhi})$. \qed
\end{lemma}

We conclude that $(\So(G), \|\cdot\|^{}_{\So})$ is a Segal algebra. In particular $\LKsp(G)$ is dense in $(\So(G), \|\cdot\|^{}_{\So})$. Also note that $(\So(G),  \|\cdot\|^{}_{\So})$ is \emph{strongly character invariant}, i.e., we have $\|\chi f\|^{}_{\So}=\|f\|^{}_{\So}$ for all $f\in \So(G)$ and for all characters $\chi\in \widehat G$.  

In fact, the Feichtinger algebra $(\So(G), \|\cdot\|^{}_{\So})$ is characterised by a certain minimality property~\cite[Thm.~1]{FRSS-Fei81}: Among all strongly character invariant Segal algebras it is the smallest one, i.e., it is continuously embedded into any such Segal algebra. See \cite[Thm.~7.3]{FRSS-Jak} for an extension of this result. Within the context of Segal algebras, Feichtinger's algebra has been studied by Reiter in \cite{FRSS-Rei2, FRSS-Rei3, FRSS-Rei4}.

\subsection{Continuous norms on $\So(G)$}\label{FRSS-sec:contSo}

In analogy to the Wiener algebra case discussed in Section~\ref{FRSS-ssec:adcn}, also the Feichtinger algebra $\So(G)$ admits an atomic decomposition and continuous norms.
Here, we only discuss continuous norms on $\So(G)$.
For $f,g\in \Asp(G)$ define their control function $\kappa_{\Asp}(f,g): G\to [0,\infty)$ by
\begin{displaymath}
\kappa_{\Asp}(f,g)(z)=\|f T_z\overline{g}\|^{}_{\Asp(G)} \ts .
\end{displaymath}
Note that we have $\kappa_{\Asp}(f,g)\ge \kappa(f,g)$. As the Fourier algebra norm $\|\cdot\|_{\Asp(G)}$ is translation invariant and invariant under complex conjugation, the control function $\kappa_{\Asp}$ satisfies properties (a)-(f) in Lemma~\ref{FRSS-Lem:cont-prop}, with $\|\cdot\|_\infty$ replaced by $\|\cdot\|_{\Asp}$. Since the Fourier algebra norm is multiplicative, the proof of the upper inequality in Lemma~\ref{FRSS-lem:contnorm} carries over when we replace $\| \cdot \|_\infty$ by $\| \cdot \|_{\Asp}$. We thus have the following result.

\begin{lemma}\label{FRSS-p1}
Let $f,g \in \So(G)$. Let $\varPhi, \varPsi$ be any BUPUs for $\Asp(G)$ of size $U$ and $V$, respectively. We then have
\begin{displaymath}
\|f\|^{}_{1} \cdot \|g\|^{}_\infty  \, \le \,  \|\kappa_{\Asp}(f,g)\|^{}_1 \,\le\,  m_G(U-V) \cdot \|f\|^{}_{\So,\varPhi} \cdot \|g\|^{}_{\So,\varPsi} \ts . \eqno \qed
\end{displaymath}
\end{lemma}
For $g\in\So(G)\setminus \{0\}$ define now
\begin{displaymath}
\|\cdot\|_{\So(G),g}=\|\kappa_{\Asp}(\cdot,g)\|^{}_1 \ts .
\end{displaymath}
\noindent The following result is proved as in Proposition~\ref{FRSS-prop:contnorm}.

\omit{
\begin{lemma}\label{lem:contnormS}
Let $g\in \So(G)$ be arbitrary, and let $\varPhi$ be any BUPU in $\Asp(G)$. Then, there exists a finite constant $C$ such that every $f\in \So(G)$ satisfies
\begin{displaymath}
\|\kappa_{\Asp}(f,g)\|^{}_1 \,\le\, C \|f\|^{}_{\So,\varPhi} \ts . \eqno\qed
\end{displaymath}
\end{lemma}
}

\begin{proposition}\label{FRSS-prop:normSoG}
Let $g\in \So(G)\setminus \{0\}$. Then $\|\cdot\|_{\So(G),g}$ is an equivalent norm on $\So(G)$, which is homogeneous and strongly character invariant. Moreover, for all $f\in \So(G)$ we have
\begin{displaymath}
\| f \|^{}_{\So(G), g}  \,=\,  \| \overline{f} \|^{}_{\So(G), \overline{g}}\,=\,  \| f^\dag \|^{}_{\So(G), g^\dag}  \ts .  \eqno \qed
\end{displaymath}
\end{proposition}

\smallskip

\begin{remark}[Time-frequency analysis over $G$]
For arbitrary $f\in \Lisp(G)$ and $g\in\Lsp^{\infty}(G)$, consider the function $\mathcal V_gf$ on $G\times \widehat G$ defined by
\begin{displaymath}
(\mathcal V_gf)(t, \chi)\, =\int_{G} f(s) \overline{\chi(s)}\overline{g(s-t)}  \dd s \ts .
\end{displaymath}
In time-frequency analysis \cite[Ch.~3]{FRSS-G01}, the function $\mathcal V_gf$ is called the \emph{short-time Fourier transform of $f$ with window $g$}.
Note that $\mathcal V_gf\in \Cusp(G\times \widehat G)$, which follows from inversion invariance of the Haar measure, together with $\|\cdot\|^{}_{1}$-continuity with respect to  translation and character multiplication.

For $f\in \So(G)$ and $g\in \So(G)\setminus \{0\}$, the above expression for $\|f\|^{}_{\So(G), g}$ can be written as
\begin{displaymath}\index{$\| \, \|^{}_{\So(G), g}$}
\|f\|^{}_{\So(G), g} \,= \int_G \int_{\widehat{G}} \left| (\mathcal V_gf)(t, \chi)\right| \dd \chi \dd t \ts .
\end{displaymath}
This indicates that $\So(G)$ is a natural test function space for time-frequency analysis; see \cite[Ch.~12.1]{FRSS-G01} for a detailed discussion. \exend
\end{remark}

\subsection{Fourier properties of $\So(G)$}

Fourier analysis on $\So(G)$ shares many properties with that on $\Ltsp(G)$. This was studied in \cite{FRSS-Fei81} using the minimality property of $\So(G)$. Here we give two alternative direct arguments.
\begin{theorem}[{\cite[Thm.~7]{FRSS-Fei81}}]\label{FRSS-thm:FFT}
If $f \in \So(G)$, then $\widehat{f} \in \So(\widehat{G})$. As a consequence, the Fourier transform \, $\widehat{} : \So(G) \to \So(\widehat G)$ is a bijection.
\end{theorem}
\begin{beweis}
Let $\varPhi=(\phi_i)_{i\in I}$ be a BUPU in $\Asp(G)$ of size $U$ and consider arbitrary $f\in \So(G)$.
We then have $f =   \sum_{i\in I} f \phi_i$ both pointwise and in $\So(G)$, the latter due to Lemma~\ref{FRSS-lem:normso}(b). In particular, the sum is convergent in $\Lisp(G)$ and thus also in $\Ltsp(G)$, by boundedndess of $f$. This implies
\[
\widehat f= \sum_{i\in I} \widehat f * \widehat \phi_i
\]
pointwise, by continuity of $\widehat f$ and by continuity of the Fourier transform on $\Ltsp(G)$. We show that the latter sum is convergent in $\So(\widehat G)$, which proves the claim.

\noindent Pick a norm on $\So(\widehat G)$ that is character invariant.
Let $\psi\in \KLsp(G)$ such that $\psi=1$ on $U$ and define $\psi_i=T_{x_i}\psi$. This gives $\psi_i\cdot \phi_i=\phi_i$.
Note here that since $\So(\widehat{G})$ is a convolution ideal in $\Lisp(\widehat{G})$, see Theorem~\ref{FRSS-thm:S0}(f), there exists some finite constant $C$ so that
\[
\|\widehat f * \widehat \phi_i* \widehat \psi_i\|^{}_{\So}
\,\le\, C \|\widehat \psi_i\|^{}_{\So} \cdot \|\widehat f * \widehat \phi_i\|^{}_{1} \ts .
\]
We can now estimate
\begin{displaymath}
\|\widehat f * \widehat \phi_i\|^{}_{\So}\,=\,\|\widehat f * \widehat \phi_i* \widehat \psi_i\|^{}_{\So}
\,\le\, C \|\widehat \psi_i\|^{}_{\So} \cdot \|\widehat f * \widehat \phi_i\|^{}_{1}\,=\,C  \|\widehat \psi\|^{}_{\So}\cdot  \|f\phi_i\|^{}_{\Asp(G)} \ts .
\end{displaymath}
As $f\in \So(G)$, this shows
\begin{displaymath}
\|\widehat f \|^{}_{\So} \, = \, \|\sum_{i\in I} \widehat f * \widehat \phi_i \|^{}_{\So} \, \le \, \sum_{i\in I} \|\widehat f * \widehat \phi_i\|^{}_{\So}\,\le\,
C \|\widehat \psi\|^{}_{\So}\cdot\sum_{i\in I}\|f \phi_i\|^{}_{\Asp(G)} \,<\,\infty \ts .
\end{displaymath}
We conclude $\mathcal F \So(G)\subseteq \So(\widehat G)$. As $\So(G)\subseteq \Lisp(G)\cap \COsp(G)$, the bijectivity statement follows from the Fourier inversion theorem \cite[Thm.~4.4.5]{FRSS-Rei2}.
\qed \end{beweis}

\begin{remark}
 Let $\varPsi=(\psi_i)_{i\in I}$ be any BUPU in $\Asp(\widehat G)$. As a consequence of the previous result, the Feichtinger algebra is characterised on the Fourier side by
\[
\So(G)\,=\, \{ f \in \Lisp(G)\cap \COsp(G) : \|\widehat f \|^{}_{\So,\psi} \,<\, \infty \} \ts ,
\]
with a corresponding equivalent norm
\[
\|\widehat f \|^{}_{\So,\psi} \,=\, \sum_{i \in I} \|\widehat f  \psi_i \|^{}_{\Asp(\widehat G)} 
\,=\, \sum_{i \in I} \|f * \widecheck{\psi_i} \|^{}_{\Lisp(G)} \ts .
\]
In the Euclidean case, this indicates a close connection to Besov spaces. This is described in detail in \cite[Thm.~3.1]{FRSS-T04} for general modulation spaces, of which Feichtinger's algebra is an important example \cite{FRSS-Jak}.  
\exend
\end{remark}

We now show that the continuous norm $\|\cdot\|^{}_{\So(G), g}$ makes the Fourier transform an isometric bijection. 
The following proof provides our second argument for the relation $\mathcal F\So(G)=\So(\widehat G)$.
\begin{theorem}\label{FRSS-th:bij}
Fix any
$g\in \So(G) \setminus \{0\}$. For arbitrary $f\in \So(G)$, we then have
\begin{displaymath}
\| f \|^{}_{\So(G), g} \,=\,   \| \widehat{f} \|^{}_{\So(\widehat{G}), \widehat{g}} \ts ,
\end{displaymath}
where $\widehat g\in \So(\widehat G) \setminus \{0\}$. Thus, the Fourier transform is an isometric isomorphism
\begin{displaymath}
 \widehat{ \, } :  ( \So(G),  \| \cdot \|^{}_{\So(G), g} ) \,\to\, ( \So(\widehat{G}),  \| \cdot \|^{}_{\So(\widehat{G}), \widehat{g}} ) \ts .
\end{displaymath}
\end{theorem}

\smallskip

\begin{remark} 
In Euclidean space $G=\RR^d$, one may pick any non-trivial Schwartz function $g \in \Scsp(\RR^d)$ such that $\widehat{g}=g$, for example the $d$-dimensional Gaussian function.
We then have
\[
\| \widehat{f} \|^{}_{\So(\RR^d), g} \, = \,  \| f \|^{}_{\So(\RR^d), g}  
\]
for every $ \in \So(\RR^d)$.
\exend
\end{remark}
\begin{beweis}[Proof of Theorem~\ref{FRSS-th:bij}]
Consider arbitrary $f\in \So(G)$ and arbitrary $g\in \So(G) \setminus \{0\}$. 
As the assumptions on $f,g$ imply $f,g\in \Ltsp(G)$, an elementary computation using Plancherel's formula yields
\begin{displaymath}
\begin{split}
(\mathcal V_gf)(t,\chi) \,&=\, \langle f, \chi T_tg \rangle \,=\, \langle \widehat f, \widehat{\chi T_tg} \rangle \\[1mm]
\,&=\, \overline{\chi}(t) \langle \widehat f, \overline{t} T_\chi \widehat g\rangle
\,=\, \overline{\chi}(t) (\mathcal V_{\widehat g} \widehat f)(\chi, -t) \ts ,
\end{split}
\end{displaymath}
for arbitrary $t\in G$ and $\chi\in\widehat G$. The latter equation is the so-called \emph{fundamental equation of time-frequency analysis} \cite[Eq.~(3.10)]{FRSS-G01}. We thus have
\begin{displaymath}
\begin{split}
\|\kappa_\Asp(f, g)\|^{}_{1} \,&=\,
\int_{G} \int_{\widehat G} |(\mathcal V_gf)(t,\chi) |  \, {\rm d}\chi \, {\rm d}t  \\[1mm]
\,&=\, \int_{G} \int_{\widehat G} |(\mathcal V_{\widehat g}\widehat {f})(-t,\chi) |  \, {\rm d}\chi \, {\rm d}t
 \,=\, \|\kappa_\Asp(\widehat f, \widehat g)\|^{}_{1} \ts ,
\end{split}
\end{displaymath}
and the claim easily follows.
\qed \end{beweis}

As an application, we give a short proof of the following well-known result. It may alternatively be regarded as a consequence of the Segal algebra property of the Wiener algebra; compare Section~\ref{FRSS-sec:SegAlgChar}.
\begin{corollary}\label{FRSS-cor:SoinW}
If $f\in \KLsp(G)$, then $\widehat f \in \Wsp(\widehat G)$.
\end{corollary}

\begin{beweis}
As $\KLsp(G)\subseteq \So(G)$ by Theorem~\ref{FRSS-thm:S0} (d), we have $\LKsp(G)\subseteq \So(\widehat G)$ by Theorem~\ref{FRSS-thm:FFT}. Noting $\So(\widehat G)\subseteq \Wsp(\widehat G)$, the claim follows.
\qed \end{beweis}

\section{Mild distributions}\label{FRSS-sec:fmd}

Consider the dual space $\SOp(G)$ of the Banach space $(\So(G), \| \cdot\|^{}_{\So, g})$. As usual, $\SOp(G)$ is a Banach space when equipped with the operator norm $\sigma\mapsto \op{\sigma}_{\SOp, g}$, which is given by
\begin{displaymath}\index{$\op{\cdot}_{\So(G), g}$}
\op{\sigma}_{\SOp, g}\,:=\, \sup \big\{ |\sigma(f)| : f \in \So(G), \| f\|^{}_{\So, g} \le 1 \big\} \ts .
\end{displaymath}
Note that we have the estimate
\[
\left| \sigma(f) \right| \leq \op{\sigma}_{\SOp, g} \cdot \| f\|^{}_{\So, g}
\]
for all $f \in \So(G)$.

\begin{remark} If we replace the norm $\| \cdot\|^{}_{\So, g}$ by an equivalent norm on $\So(G)$, we get an equivalent norm on $\SOp(G)$. In particular, different $g \in \So(G) \backslash \{0 \}$ induce equivalent norms on $\SOp(G)$ by Theorem~\ref{FRSS-th:bij}. We will often prefer to work with a norm $ \| \cdot \|^{}_{\So}$ that is translation invariant. In that case, the induced norm $\op{\cdot}_{\SOp}$ on $\SOp(G)$ is also translation invariant and equivalent to the norm $\op{\cdot}_{\SOp,g}$. \exend
\end{remark}

\begin{definition}[\cite{FRSS-Fei22}]
Elements of $\SOp(G)$ are called \emph{mild distributions}\index{distribution!mild}.
\end{definition}
Before we analyse mild distributions in detail, let us note that a number of important function classes give rise to mild distributions.
\begin{definition}[Regular mild distribution]\label{FRSS-def:reg}
A mild distribution $\sigma_f\in \SOp(G)$ is called \emph{regular}, if there exists some locally integrable function $f \in \Lilocsp(G)$ such that
\[
\sigma^{}_{\!f}(g) \,=\, \int_G g(t) f(t) \dd t 
\]
holds for all $g\in \So(G)$.
\end{definition}

\begin{example}[Regular mild distributions]\label{FRSS-exa:reg} 
Consider the Wiener amalgam space 
\begin{displaymath}
\Wsp(\Lisp,\lsp^\infty)(G) \,=\, \{f\in \Lilocsp(G): \max_{i\in I} \|f \phi_i\|^{}_1<\infty\} \ts ,
\end{displaymath}
where  $\varPhi=(\phi_i)_{i\in I}$ is a BUPU on $\COsp(G)$. We will argue below that every $f\in \Wsp(\Lisp,\lsp^\infty)(G)$ gives rise to a regular mild distribution $\sigma^{}_{\!f}$.  As the function spaces $\Lsp^p(G)$ for $1 \leq p \leq \infty$, $\COsp(G)$, $\Cusp(G)$, $\Asp(G)$ and $\So(G)$ are subspaces of $\Wsp(\Lisp,\lsp^\infty)(G)$, they all give rise to regular mild distributions $\sigma_f$. In fact, the map $f\mapsto \sigma^{}_{\!f}$ provides a continuous embedding into $\SOp(G)$ in all cases; compare  \cite[Lemma~6.1]{FRSS-Jak}. 

As in the case of the Wiener algebra $\Wsp(G)$, it is straightforward to show that
\begin{displaymath}
 \|f\|^{}_{\Wsp_{1,\infty, \varPhi}} \, := \, \max_{i\in I} \|f \phi_i\|^{}_1
\end{displaymath}
defines a norm on $\Wsp(\Lisp,\lsp^\infty)(G)$, and that the definition of $\Wsp(\Lisp,\lsp^\infty)(G)$ is independent of the choice of $\varPhi$. In order to show that, for $f\in \Wsp(\Lisp,\lsp^\infty)(G)$, the linear functional $\sigma_f$ satisfies $\sigma_f\in \SOp(G)$, we fix a BUPU $\varPsi$ in $\Asp(G)$ and consider any $g\in \So(G)$. Using the notation from Remark~\ref{FRSS-rem:refbupu}, we can now estimate
\begin{displaymath}
\begin{split}
|\sigma_f(g)|&\,\le\, \sum_{(i,j)\in K} \|f\phi_i\cdot g\psi_j\|^{}_1\,\le\, 
\sum_{(i,j)\in K} \|f\phi_i\|^{}_1 \cdot \|g\psi_j\|^{}_\infty \\
&\,\le\, \|f\|^{}_{\Wsp_{1,\infty, \varPhi}} \sum_{(i,j)\in K} \|g\psi_j\|^{}_\infty
\,\le\, \|f\|^{}_{\Wsp_{1,\infty, \varPhi}} \cdot \|\delta^{}_Y\|^{}_{U-V}\cdot \|g\|^{}_{\Wsp, \varPsi} \ts , 
\end{split}
\end{displaymath}
from which the claim follows as $\|g\|^{}_{\Wsp, \varPsi} \le \|g\|^{}_{\So, \varPsi}$. Conceptually, this result rests on the Hölder inequality for the Wiener amalgam spaces $\Wsp(\Lisp,\lsp^\infty)(G)$ and 
$\Wsp(\Lsp^\infty,\lisp)(G)$.

Similarly to Section~\ref{FRSS-ssec:adcn}, one may introduce on $\Wsp(\Lisp,\lsp^\infty)(G)$ the equivalent continuous norm $\|f\|^{}_{1,\infty}=\|\kappa(f,g)\|^{}_\infty$, where $\kappa(f,g)(z)=\|f T_zg\|^{}_1$, for fixed $g\in \Wsp(\Lisp,\lsp^\infty)(G)\setminus\{0\}$.  This indicates that the space $\Wsp(\Lisp,\lsp^\infty)(G)$ can be identified with the Stepanov space $S^1(G)$. The elements of $\Wsp(\Lisp,\lsp^\infty)(G)$ can also be identified with the density functions of absolutely continuous translation bounded measures, compare Section~\ref{FRSS-sec:tbmintro}. A classic reference on the Wiener amalgam spaces $\Wsp(\Lsp^p,\lsp^q)(G)$ is \cite{FRSS-FS85}. \exend
\end{example}

\begin{remark}[Mild distributions are translation bounded]\label{FRSS-rem:tbS0}
Note that any mild distribution $\sigma\in \SOp(G)$ is translation bounded on $\So(G)$\index{translation~bounded!mild~distribution}, meaning that the function  $t \mapsto \sigma(T_tf)$ is bounded for any $f\in \So(G)$. This follows from the standard estimate
\[
|\sigma(T_xf)|\le \op{\sigma}_{\SOp}\cdot \|f\|^{}_{\So} \ts ,
\]
which uses translation invariance of the $\So$-norm. In fact, this function is uniformly continuous. \exend
\end{remark}

\begin{remark}[Mild distributions are quasimeasures]\label{FRSS-rem:BSdense}
Recall that the elements of the dual space of $\KLsp(G)$, equipped with the $\Asp(G)$-inductive limit topology, are called \emph{quasimeasures}, see \cite{FRSS-C78} and \cite[Rem.~2]{FRSS-Fei81}. Thus, the space of mild distributions coincides with the space of quasimeasures that are translation bounded with respect to $\KLsp(G)$. \exend
\end{remark}

\begin{remark}[Mild distributions are tempered]\label{FRSS-rem:sbmdis}
Recall from Remark~\ref{FRSS-rem:fce} that the Bruhat--Schwartz space $\Scsp(G)$ is a dense subspace of $\So(G)$, and that the identity provides a continuous embedding. Thus, any mild distribution on $G$ is a tempered distribution \cite[Thm.~B1]{FRSS-Fei80}, and moreover the space of mild distributions is weak*-dense in $\Scsp'(G)$. Both statements rely on an elementary duality argument, with the identity map as embedding map (as in the proof of~\cite[Thm.~4.12]{FRSS-RUD3}). \exend
\end{remark}

\begin{remark}[Some weak*-dense subspaces]\label{FRSS-rem:wsds}

The dual  $(\Wspp(G), \op{\cdot}_{\Wspp})$ of the Wiener algebra is weak*-dense and continuously embedded into the mild distributions $(\SOp(G),\op{\cdot}_{\SOp})$. 
This holds as $(\So(G), \|\cdot\|^{}_{\So})$ is dense and continuously embedded in $(\Wsp(G), \|\cdot\|^{}_{\Wsp})$ by Proposition~\ref{FRSS-prop:S0Wdense}, using the same elementary duality argument as in Remark~\ref{FRSS-rem:sbmdis}.  In fact, the span of all Dirac measures $\delta^{}_{x}\in \Wspp(G)$, given by 
\begin{displaymath}
\delta^{}_{x}(f)=f(x) \ts ,
\end{displaymath}
is weak*-dense in $\SOp(G)$. Also, the Feichtinger algebra $\So(G)$ is weak*-dense in $\SOp(G)$. The proof of the latter two statements given in \cite[Lemma~6.4]{FRSS-Jak} relies on the Hahn--Banach theorem. The latter statement implies that any dense subset of $\So(G)$, such as $\KLsp(G)$ or $\Scsp(G)$, is weak*-dense in $\SOp(G)$. It also implies the first statement of this remark, noting the canonical continuous embeddings $\So(G)\subseteq \Wspp(G)\subseteq \SOp(G)$.
\exend
\end{remark}

Let $\sigma \in \SOp(G)$  be given. As usual, the mild distributions $\overline{\sigma},  \sigma^\dag  \in \SOp(G)$ are then defined by
\begin{displaymath}
\overline{\sigma}(f) \, =\, \overline{ \sigma({\bar{f}})} \ts , \qquad 
\sigma^\dag(f)\,= \,  \sigma(f^\dag) \ts .
\end{displaymath}
Similarly, for $g\in \Asp(G)\supseteq \So(G)$, the mild distribution $g \cdot \sigma \in \SOp(G)$ is defined by
\begin{displaymath}
(g \cdot \sigma)(f) \, = \, \sigma (g \cdot f) \ts ,
\end{displaymath}
compare Theorem~\ref{FRSS-thm:S0}(e).
Also, for $g\in \Lisp(G)\supseteq \So(G)$, the mild distribution $g* \sigma \in \SOp(G)$ is defined by
\begin{displaymath}
(g* \sigma)(f) \, =\,  \sigma(g^\dag*f) \ts ,
\end{displaymath}
compare Theorem~\ref{FRSS-thm:S0}(f).

\begin{remark}[Compatibility of convolution definitions]
We will show in Proposition~\ref{FRSS-prop:convcomp} below that the above definition of $g* \sigma$ as a mild distribution is compatible with the definition as function $x \mapsto \sigma(T_x g^\dag)$, which we already used in Lemma~\ref{FRSS-lem:convcu}. In particular, for $g \in \So(G) \subseteq \Wsp(G)$ and for $\sigma \in \Wspp(G) \subseteq \SOp(G)$, both definitions are compatible, and no confusion might arise when interpreting the expression $g*\sigma$. \exend
\end{remark}

\begin{definition}[Fourier transform of a mild distribution]
For a mild distribution $\sigma\in \SOp(G)$, its \emph{Fourier transform} $\widehat{\sigma}\in \SOp(\widehat G)$ is defined via
\[
\widehat{\sigma}(f) \,=\, \sigma(\widehat{f}) 
\]
\index{Fourier~transform!of~a~mild~distribution}for all  $f\in \So(\widehat{G})$. 
\end{definition}

It is not hard to see that the distributional Fourier transform $\SOp(G) \to \SOp(\widehat G)$ coincides with the Fourier transform of tempered Bruhat--Schwartz distributions \cite{FRSS-B61, FRSS-O75}, when restricted to the invariant subspace of mild distributions. Moreover, it coincides with the Fourier transform of transformable measures \cite{FRSS-ARMA1}, as will be discussed in Theorem~\ref{FRSS-t1a} below. We list the basic properties of the Fourier transform for mild distributions.

\begin{theorem}[{\cite[Ex.~6.7 and 6.8]{FRSS-Jak}}]\label{FRSS-thm:Fdist} The distributional Fourier transform  $\sigma\mapsto \widehat\sigma$ acts isometrically between $(\SOp(G),\op{\cdot}_{\SOp(G),g})$ and $(\SOp(\widehat G),\op{\cdot}_{\SOp(\widehat G), \widehat g})$, and it has the following properties.
\begin{itemize}
\setlength\itemsep{2pt}
  \item[(a)] $\widehat{\widehat{\sigma}}= \sigma^\dag$ for all $\sigma \in \SOp(G)$.
  \item[(b)] $\widehat{g\cdot \sigma} = \widehat{g} * \widehat{\sigma}$ for all $g \in \Asp(G)\supseteq \So(G)$ and all $\sigma \in \SOp(G)$.
    \item[(c)] $\widehat{g* \sigma} = \widehat{g} \cdot \widehat{\sigma}$ for all $g \in \Lisp(G)\supseteq \So(G)$ and all $\sigma \in \SOp(G)$.
   \item[(d)]  $\widehat{\sigma_{\chi}} = \delta^{}_{\chi}$ and $\widehat{\delta^{}_{\chi}} = \sigma_{\overline{\chi}}$ for all $\chi \in \widehat{G}$. 
  \end{itemize}
Moreover, it is a weak*-weak*-continuous operator between $\SOp(G)$ and $\SOp(\widehat G)$, and as such it is the unique extension of the Fourier transform on $\Ltsp(G)$.
  \qed
\end{theorem}

We now show that the two definitions of convolution as a mild distribution and as a function are compatible.
\begin{proposition}[Convolution as a mild distribution]\label{FRSS-prop:convcomp}
Let $f \in \So(G)$ and $\sigma \in \SOp(G)$. Define $h : G \to \CC$ via $h(x) =  \sigma(T_x f^\dag)$.
Then $h \in \Cusp(G)$, and the regular distribution $\sigma^{}_{\! h}\in \SOp(G)$ satisfies $\sigma^{}_{\! h} = f* \sigma$.
\end{proposition}

For the proof, we use that $\So(G)$ is stable by convolution with finite Radon measures, which extends convolution with integrable functions in Theorem~\ref{FRSS-thm:S0}(f). There is an analogous result  for convolution stability of $\Wsp(G)$ that extends Theorem~\ref{FRSS-thm:W}(g), which will be discussed in Lemma~\ref{FRSS-lem:fw} below. In fact, convolution stability is satisfied for test functions from abstract homogeneous Banach spaces \cite[Thm.~2]{FRSS-Fei22}. Here we prefer a measure-theoretic argument that uses the continuous norm $\|\cdot\|_{\So,g}$ from Section~\ref{FRSS-sec:contSo}. We denote by $\cMb(G)$ the space of finite Radon measures on $G$.

\begin{lemma}\label{FRSS-lem:mufest}
Fix any norm $\|\cdot\|^{}_{\So,g}$ on $\So(G)$. We then have  $\mu*f\in \So(G)$ for every $\mu\in \cMb(G)$ and every $f\in \So(G)$. Moreover the norm estimate
\[
\|\mu*f\|^{}_{\So,g} \, \le \,  |\mu|(G) \cdot \|f\|^{}_{\So,g} 
\]
holds, where $|\mu|$ denotes the total variation measure of $\mu$.
\end{lemma}

\begin{beweis}
We first note
\begin{displaymath}
\|f\|^{}_{\So,g}\,=\,\int^{}_{\widehat G}\int^{}_G |(\mathcal V_gf)(t,\chi)| \dd t \dd\chi\,=\, \int^{}_{\widehat G} \|f*\overline{\chi}\overline{g}\|^{}_1 \dd\chi \ts ,
\end{displaymath}
which follows from the identity $(\mathcal V_gf)(t,\chi)=\chi(t)\cdot (f*\overline{\chi}\overline{g})(t)$. Together with the elementary estimate $\|\mu*f\|^{}_1\le |\mu|(G)\cdot \|f\|^{}_1$, we get
\begin{displaymath}
\begin{split}
\|\mu*f\|^{}_{\So,g}\,&=\,\int_{\widehat G} \|\mu*f*\overline{\chi g}\|^{}_1 \dd\chi \\
\,&\le\, |\mu|(G)\cdot \int_{\widehat G} \|f*\overline{\chi}\overline{g}\|^{}_1 \dd\chi 
 \,=\, |\mu|(G)\cdot \|f\|^{}_{\So,g} \ts .
\end{split}
\end{displaymath}
This shows the claimed norm estimate and, in particular, that $\mu*f\in \So(G)$.
\qed \end{beweis}

\begin{beweis}[Proof of Proposition~\ref{FRSS-prop:convcomp}]
Note that we have for all $x \in G$ the estimate
\begin{displaymath}
|h(x)|\,=\,\left| \sigma(T_x f^\dag) \right| \, \le \, \op{\sigma}_{\SOp,g} \cdot \| T_xf^\dag\|^{}_{\So,g}\,=\, \op{\sigma}_{\SOp,g} \cdot \| f^\dag\|^{}_{\So,g} \ts ,
\end{displaymath}
where we used translation invariance of the norm $\|\cdot\|_{\So,g}$.
This shows that $h$ is bounded. Uniform continuity of $h$ follows with a similar estimate, as the norm $\|\cdot\|_{\So,g}$ is homogeneous by Proposition~\ref{FRSS-prop:normSoG}. 

\noindent We next show that $\sigma^{}_{\! h}(k) = (f*\sigma) (k)$ holds for every $f \in \So(G)$ and for every $k\in \KLsp(G)$. This proves the claim as $\KLsp(G)$ is dense in $(\So(G), \|\cdot\|_{\So,g})$.
We follow a strategy familiar from distributions and approximate the convolution integral by a finite sum, see e.g.~\cite[Lem.~4.13]{FRSS-H89}. 

\noindent Consider any BUPU $\varPsi=(\psi_i)_{i\in I}$ in $\Asp(G)$ with non-negative functions $\psi_i\ge0$. Following \cite{FRSS-Fei22}, we replace $h\in \Cu(G)$ by the spline interpolation
\begin{displaymath}
\operatorname{Sp}_\varPsi \! h \,=\, \sum_{i\in I} h(x_i)\psi_i \ts ,
\end{displaymath}
which satisfies $\operatorname{Sp}_\varPsi  \!  h\in \Cu(G)$. Note that $\|\operatorname{Sp}_\varPsi  \!  h-h\|^{}_\infty$ vanishes if the size of $\varPsi$ shrinks to zero. Indeed, consider an arbitrary $\varepsilon>0$ and choose the size $U$ of $\varPsi$ such that $x-y\in U$ implies $|h(x)-h(y)|\le \varepsilon$. We then have
\begin{displaymath}
\begin{split}
|h(x)-(\operatorname{Sp}_\varPsi  \!  h)(x)| &\,=\, \big|\sum_{i\in I} h(x)\psi_i(x) -h(x_i)\psi_i(x)\big|\\ &\,\le\, \sum_{i\in I} |h(x)-h(x_i)| \, \psi_i(x) 
 \,\le\, 
\sum_{i\in I} \varepsilon \psi_i(x)\,=\,\varepsilon \ts ,
\end{split}
\end{displaymath}
where we used that the $\psi_i$ are non-negative. As $\varepsilon>0$ was arbitrary, this shows that $\|\operatorname{Sp}_\varPsi  \!  h-h\|^{}_\infty$ vanishes if the size of $\varPsi$ shrinks to zero.

\noindent Note that, since $\supp(k)$ is compact, there exists a finite index set $J \subseteq I$ such that $k\psi_i =0$ for all  $i \notin J$. Writing $k_i= \int_{G} k(x) \psi_i(x) \dd x$, we can now compute
\begin{displaymath}
\begin{split}
\sigma_{\operatorname{Sp}_\varPsi \! h}(k)&\,=\, \int_{G} k(x) (\operatorname{Sp}_{\varPsi} \! h)(x)\dd x \,=\, \int_{G} k(x) \sum_{i\in I} \sigma(T_{x_i}f^\dag) \psi_i(x) \dd x \\
&\,=\, \int_{G} k(x) \sum_{i\in J} \sigma(T_{x_i}f^\dag) \psi_i(x) \dd x \, =\, \sum_{i\in J} \sigma(T_{x_i}f^\dag)  \int_{G} k(x) \psi_i(x) \dd x \\
&\,=\,  \sigma\Big( \sum_{i\in J} k_i T_{x_i}f^\dag\Big) \, =\,  \sigma\Big( \sum_{i\in J} k_i T_{x_i}f^\dag\Big)  \,=\, \sigma \big((D_\varPsi k) *f^\dag\big)  \ts ,
\end{split}
\end{displaymath}
where the finite measure $D_\varPsi k$ is a discretisation of $k$ given by the \textit{finite} sum
\begin{displaymath}
D_\varPsi k = \sum_{i\in J} k_i \delta_{x_i}= \sum_{i\in I} k_i \delta_{x_i} \ts ,
\end{displaymath}
compare \cite{FRSS-Fei22}. 
Next, we quantify the error that results from replacing $h$ by $\operatorname{Sp}_\varPsi \! h$ and from replacing $k$ by $D_\varPsi k$. Note first that we have the standard estimate
\begin{displaymath}
|\sigma_{\operatorname{Sp}_\varPsi \! h}(k)-\sigma_h(k)|\,\le\,  \int_G |(\operatorname{Sp}_\varPsi \! h)(x)-h(x)|\cdot |k(x)| \dd x
 \,\le\, \|\operatorname{Sp}_\varPsi \! h-h\|^{}_\infty \cdot \|k\|^{}_1 \ts .
\end{displaymath}
Thus the approximation error for $\sigma_h(k)$ vanishes as the size of $\varPsi$ shrinks to zero.
Moreover we have
\begin{displaymath}
\begin{split}
|\sigma ((D_\varPsi k) *f^\dag) - \sigma (k *f^\dag)|&\,\le\, \op{\sigma}^{}_{\SOp, g}\cdot\|(D_\varPsi k)*f^\dag-k*f^\dag\|^{}_{\So,g}\\
&\,\le\, \op{\sigma}_{\SOp,g}^{}\cdot|D_\varPsi k-k|(G)\cdot\|f^\dag\|^{}_{\So,g} \ts ,
\end{split}
\end{displaymath}
where we used Lemma~\ref{FRSS-lem:mufest}.
Take now any $\varphi\in \KLsp(G)$ such that $\varphi\ge0$ and $\varphi=1$ on $\supp(D_\varPsi k-k)$ for $\varPsi$ and all BUPUs $\varPsi'$ of smaller size. For arbitrary fixed $\varepsilon>0$ we then have 
\begin{displaymath}
\begin{split}
|D_\varPsi k-k|(G)\,&=\,|D_\varPsi k-k|(\varphi) \,\le\,|(D_\varPsi k-k)(\varphi')|+\varepsilon \\\,&=\, |k(\operatorname{Sp}_\varPsi \! \varphi'-\varphi')| +\varepsilon \,\le\, \|k\|^{}_1 \cdot \|\operatorname{Sp}_\varPsi \! \varphi'-\varphi'\|^{}_\infty + \varepsilon \ts ,
\end{split}
\end{displaymath}
for some $\varphi'\in \Cc(G)$ such that $|\varphi'|\le \varphi$. 

\noindent Here we used that the total variation $|\mu|$ of a Radon measure $\mu$ satisfies
\begin{displaymath}
|\mu|(g)\,=\,\sup\{|\mu(f)|: f\in \Cc(G), |f|\le g\}
\end{displaymath}
for $g\in \Cc(G)$ such that $g\ge 0$.
 For $\varPsi$ of sufficiently small size, the rhs can be made smaller than $2\varepsilon$ as $\varphi'\in \Cu(G)$. As $\varepsilon>0$ was arbitrary, this shows that the approximation error for $\sigma(k*f^\dag)$ vanishes as the size of $\varPsi$ shrinks to zero.
\qed \end{beweis}

\subsection{Positive-definite mild distributions}

Let us now discuss positive-definite mild distributions and their Fourier transform.
We call a continuous function $f\in \Csp(G)$ \emph{positive definite}\index{positive~definite!function} if there exists a finite Borel measure $\mu$ on $\widehat G$ such that $f(x)=\int_{\widehat G} \chi(x) {\rm d}\mu(\chi)$. This is justified by Bochner's theorem on positive-definite functions \cite[Ch.~1.4]{FRSS-RUD2}. Then, $f\in \So(G)$ is positive definite if and only if $\widehat f\in \So(\widehat G)$ is non-negative. This holds as the measure $\mu=\widehat f \cdot m_{\widehat G}$ is uniquely determined by $f$.

We call a mild distribution $\sigma\in \SOp(G)$ \emph{positive definite}\index{positive~definite!mild~distribution} if $\sigma( f)\ge0$ for all positive-definite $f\in \So(G)$.

\smallskip

We first show that positive definiteness is implied by $\sigma(f)\ge0$ for all positive-definite $f \in \KLsp(G)$.
This will be used in the following section to argue that, for a mild distribution that is a measure,  the notion of positive-definite measure coincides with the notion of positive-definite mild distribution.

\smallskip

We start following denseness result for $(\Lisp(G), \|\cdot\|^{}_{1})$, which slightly improves \cite[Prop.~5.4.1]{FRSS-Rei2}.

\begin{lemma}\label{FRSS-lemma-appr} Let $f \in \Lisp(G)$ and let $\eps >0$ be arbitrary. Then,  there exists some $g \in \LKsp(G)$ such that $\| f-g\|^{}_{1} < \eps$. Furthermore, if $f$ is also non-negative, $g$ can be chosen non-negative.
\end{lemma}

\begin{beweis}
We follow \cite[Prop.~5.4.1]{FRSS-Rei2} and define $h(x)=\sqrt{|f(x)|}$ and
\[
h_1(x)\,=\,\left\{
\begin{array}{cc}
  \frac{f(x)}{h(x)} & \text{if } h(x) \neq 0  \\
  0 & \text{otherwise}
\end{array}
\right. \ts .
\]
Then, $|h|^2=|h_1|^2=|f|$ by construction. Therefore, $h,h_1 \in \Ltsp(G)$ and hence $\widehat{h}, \widehat{h_1} \in \Ltsp(\widehat{G})$. Since $\Ccsp(\widehat{G})$ is dense in $\Ltsp(\widehat{G})$, we can find two functions $\varphi, \psi \in \Ccsp(\widehat{G})$ such that
\begin{displaymath}
 \| \widehat{h_1} - \varphi \|^{}_{2}  \,\leq\,   \frac{\eps}{2  \|h\|^{}_{2} +2 } \ts , \qquad
  \| \widehat{h} - \psi \|^{}_{2}  \,\leq\,  \frac{\eps}{2  \|\varphi\|^{}_{2} +2 }   \ts .
\end{displaymath}
Then, $g= \widecheck{\varphi} \cdot \widecheck{\psi} = \widecheck{\varphi*\psi} \in \LKsp(G)$ and
\begin{displaymath}
\begin{split}
\| f&-g\|^{}_{1} \,=\,
 \| h\cdot h_1- \widecheck{\varphi} \cdot\widecheck{\psi} \|^{}_{1}
\,=\, \| h\cdot h_1- h \cdot \widecheck{\varphi} \|^{}_{1} +\| h\cdot \widecheck{\varphi}- \widecheck{\varphi} \cdot\widecheck{\psi} \|^{}_{1} \\
 &\,\leq\, \|h\|^{}_{2} \cdot  \| h_1-\widecheck{\varphi}\|^{}_{2} + \| h- \widecheck{\psi} \|^{}_{2} \cdot \| \widecheck{\varphi} \|^{}_{2} \\
& \,=\, \|h\|^{}_{2}  \cdot \| \widehat{h_1}-\varphi\|^{}_{2} + \| \widehat{h}- \psi \|^{}_{2} \cdot  \| \varphi \|^{}_{2} \\
 &\,\leq\,  \|h\|^{}_{2} \cdot \frac{\eps}{2  \|h\|^{}_{2} +2 }  + \| \varphi \|^{}_{2} \cdot \frac{\eps}{2  \|\varphi\|^{}_{2} +2 }  \,\leq\, \eps \ts .
\end{split}
\end{displaymath}
Consider finally the case that $f=h\cdot h_1$ is non-negative. Then, we can choose $h=h_1 \geq 0 $. Pick some  $\varphi\in \Ccsp(\widehat{G})$ such that
\[
\| \widehat{h_1} - \varphi \|^{}_{2}  \,\leq\, \frac{\eps}{2  \|h\|^{}_{2} +2 }  \ts ,
\]
and let $\psi = \widetilde{\varphi}$. We then have $g=\widecheck{\varphi} \cdot \widecheck{\psi}=\reallywidecheck{\varphi*\widetilde{\varphi}} \geq 0$. Moreover, $\|f-g\|^{}_{1}\le \varepsilon$, which follows by a similar estimate as above using
\[
\| \widehat{h} - \psi \|^{}_{2} 
\,=\,\| h -\reallywidecheck{\widetilde{\varphi}} \|^{}_{2} \,=\,  \| h -\overline{\widecheck{\varphi}} \|^{}_{2}
\,=\,  \| h -\widecheck{\varphi} \|^{}_{2}  \,=\,  \| \widehat{h} -\varphi \|^{}_{2} \ts .
\]
Here, the third equality uses $\overline{h}=h$, which holds by positivity.
\qed \end{beweis}

The above result has the following consequence for $\KLsp(G)$-approximation of positive-definite functions in $\So(G)$.

\begin{lemma}\label{FRSS-pd-appr} Let $f \in \So(G)$ be positive definite and let $\eps>0$ be arbitrary. Then, there exists a positive-definite $f_1 \in \KLsp(G)$ such that
\[
\| f-f_1\|^{}_{\So,g} \,<\, \eps \,.
\]
\end{lemma}

\begin{beweis}
By \cite[Prop.~4.16(i)]{FRSS-Jak}, there exists some $h \in \KLsp(\widehat{G})$ such that
\[
\| \widehat{f}* h - \widehat{f} \|^{}_{\So, \widehat{g}} \,<\, \frac{\eps}{2} \ts .
\]
Moreover, $h$ can be chosen non-negative; compare the proof of \cite[Prop.~4.16]{FRSS-Jak}. By Lemma~\ref{FRSS-lemma-appr}, there exists some non-negative $u \in \LKsp(G)$ such that
\[
\| h-u \|^{}_{1} \,<\, \frac{\eps}{2 \| f \|^{}_{\So,g}} \,.
\]
Define $f_1= \widecheck{u} \cdot f\in \KLsp(G)$. Then $\widehat{f_1}= u *\widehat{f} \geq 0$, hence $f_1$ is positive definite and,
\begin{displaymath}
\begin{split}
\| f-f_1\|^{}_{\So,g}&\,=\, \| f-\widecheck{u}\cdot f\|^{}_{\So,g}\,\leq\,  \| f-\widecheck{h}\cdot f\|^{}_{\So,g}+ \| \widecheck{h}\cdot f-\widecheck{u}\cdot f\|^{}_{\So,g} \\
&\,=\, \| \widehat{f}-h*\widehat{f} \|^{}_{\So,\widehat{g}}+ \| \widehat{f}*h-\widehat{f}*u \|^{}_{\So,\widehat{g}} \\
&\,\leq\, \frac{\eps}{2}+ \| h-u \|^{}_{1} \cdot \| \widehat{f} \|^{}_{\So, \widehat{g}}\, =\, \frac{\eps}{2}+\|h-u\|^{}_{1}\cdot  \| f \|^{}_{\So,g} \,<\,\eps \ts .
\end{split}
\end{displaymath}
This proves the claim.
\qed \end{beweis}
We call a mild distribution $\sigma\in \SOp(G)$ \emph{positive}\index{positive!mild~distribution} if $\sigma( f)\ge0$ for all $f\in \So(G)$ with $f \geq 0$. We will show in Proposition~\ref{FRSS-prop:Feipos} that positive mild distributions are positive Radon measures. For positive-definite mild distributions, our preparations lead to the following result.
\begin{corollary}\label{FRSS-cor-PD} For a mild distribution $\sigma \in \SOp(G)$, the following properties are equivalent.
\begin{itemize}
\setlength\itemsep{2pt}
  \item[(i)] The mild distribution $\sigma$ is positive definite.
  \item[(ii)] For all positive-definite functions $f \in \KLsp(G)$ we have $\sigma(f) \geq 0$.
  \item[(iii)] The Fourier transform $\widehat{\sigma}$ is a positive mild distribution.
\end{itemize}
\end{corollary}

\begin{beweis}
\noindent (i) $\Longleftrightarrow$ (ii) follows immediately from Lemma~\ref{FRSS-pd-appr}.

\noindent (i) $\Longleftrightarrow$ (iii) follows from the definition  $\widehat{\sigma}(\varphi) = \sigma(\widehat{\varphi})$ of the Fourier transform and the observation that $f \in \So(G)$ is non-negative if and only if $\widehat{f}\in \So(\widehat G)$ is positive definite.
\qed \end{beweis}

We have the following obvious consequence.

\begin{corollary}\label{FRSS-lem:ppdS0}
Let $\sigma\in \SOp(G)$, $\tau\in \SOp(\widehat G)$ be such that $\widehat \sigma=\tau$. Then,
\begin{itemize}
\setlength\itemsep{2pt}
  \item[(a)]$\sigma$ is positive if and only if $\tau$ is positive definite.
  \item[(b)] $\sigma$ is positive definite if and only if $\tau$ is positive. \qed
\end{itemize}
\end{corollary}

\begin{remark} If $\sigma$ is a positive-definite mild distribution, Proposition~\ref{FRSS-prop:Feipos} below implies that $\widehat{\sigma}$ is in fact a translation-bounded measure. \exend
\end{remark}

\section{Measures}\label{FRSS-sec:m}

In this section, we review the basics of Radon measures and translation-bounded measures, discuss the connection between translation boundedness and the dual Wiener algebra, and discuss mild measures.

\subsection{Radon measures}\label{FRSS-sec:rm}

A positive Borel measure $\mu$ on a locally compact (Hausdorff) Abelian group $G$ is called a \emph{Radon measure}\index{measure!Radon} if it is inner regular and locally finite, meaning that
\[
\mu(A)\,=\, \sup\{\mu(K): K\subseteq A \text{ compact}\}
\]
holds for any Borel set $A\subseteq G$ and that $\mu(K)<\infty$ holds for every compact set $K\subseteq G$.
More generally, a complex Borel measure $\mu$ on $G$ is called a \emph{complex Radon measure} if its four components are Radon measures (that is, the positive and negative parts of its real and imaginary part). When we speak of a measure in this article, we will always mean a complex Radon measure.

Let $\cM(G)$\index{$\cM(G)$} denote the space of measures on $G$. According to the Riesz representation theorem, a linear map $\mu: \Ccsp(G)\to \CC$ satisfies $\mu\in \cM(G)$ if and only if for every compact $K\subseteq G$ there exists a finite constant $C^{}_{\!K}$ such that
\[
|\mu(f)|\,\le\, C^{}_{\! K} \cdot \|f\|^{}_{\infty}
\]
holds for every $f\in \Ccsp(G)$ such that $\supp(f)\subseteq K$. To any given measure $\mu\in \cM(G)$ we can assign the positive measure $|\mu|\in \cM(G)$ by
\[
|\mu|(f)\, = \, \sup\{|\mu(f)|: g \in \Ccsp(G), |g|\le f\} 
\]
for all $f\in \Ccsp(G)$. The measure $|\mu|$ is the \emph{total variation measure}\index{measure!total~variation} associated to $\mu$. 
The \emph{support} of $\mu\in \cM(G)$ is defined as
\[
\supp(\mu) \,=\, \{x\in G: |\mu|(U)>0 \text{ for all open } U\ni x\} \ts .
\]
\index{support!of~a~measure}

In fact, $\cM(G)$ is the topological dual of $(\Ccsp(G), \tau_{\rm ind})$, where the \emph{inductive limit topology}\index{inductive~limit~topology} $\tau_{\rm ind}$ is defined as follows. Consider the function spaces
\[
\Csp^{}_{\mkern-4mu K}(G) \,=\, \{f\in \Ccsp(G): \supp(f)\subseteq K\}
\]
for compact $K\subseteq G$. Equip $\Csp^{}_{\mkern-4mu K}(G)$ with the supremum norm $\|\cdot\|^{}_{\infty}$.
Then, the inductive limit topology $\tau_{\rm ind}$ is generated by the seminorms on $\Ccsp(G)$ whose restrictions to $(\Csp^{}_{\mkern-4mu K}(G), \|\cdot\|^{}_{\infty})$ are continuous; compare~\cite[Ch.~II.6]{FRSS-S71}.
Similarly, the topological dual of $(\COsp(G), \|\cdot\|^{}_{\infty})$ coincides with the space $\cMb(G)$\index{$\cMb(G)$}\index{measure!bounded} of bounded Radon measures.

\subsection{Translation-bounded measures}\label{FRSS-sec:tbmintro}

Recall that a measure $\mu\in \cM(G)$ is \emph{translation bounded}\index{translation~bounded!measure}\index{measure!translation-bounded} if the expression
\[
\|\mu\|^{}_{K} \,=\, \sup\left\{|\mu|(x+K): x\in G\right\}
\]
\index{$\| \cdot \|^{}_{K}$}is finite for all compact $K\subseteq G$; compare \cite[p.~5]{FRSS-ARMA1}.
Translation boundedness can be checked on any relatively compact set $U\subseteq G$ having non-empty interior. Indeed, any compact $K\subseteq G$ can be covered by a finite number $C$ of translates of $U$, which implies
$\| \mu \|^{}_{K} \leq C \| \mu \|^{}_{U}$ for every $\mu\in \cM(G)$.

We denote by $\cM^\infty(G)$\index{$\cM^\infty(G)$} the space of translation-bounded measures\footnote{This notation should not be confused with the modulation space $M^\infty(G)=\SOp(G)$; compare~\cite[Def.~6.13]{FRSS-Jak}.}  on $G$. It is not difficult to see that, for every relatively compact set $U$ with non-empty interior, $\| \cdot \|^{}_{U}$ defines a norm on $\cM^\infty(G)$. The topology induced by $\| \cdot \|^{}_{U}$ is called the \emph{norm topology} on $\cM^\infty(G)$. Let summarise the above discussion as follows; compare \cite{FRSS-BaakeMoody2004}.

\begin{lemma}\label{FRSS-lem:precomp-equiv-tb-norm} Let $U \subseteq G$ be any relatively compact set with non-empty interior. Then the following hold.
\begin{itemize}
  \item[(a)] $\mu\in \cM(G)$ is translation bounded if and only if $\| \mu\|^{}_{U} < \infty$.
  \item[(b)] $\| \cdot \|^{}_{U}$ defines a norm on $\cM^\infty(G)$.
  \item[(c)] Let $V \subseteq G$ be any relatively compact set with non-empty interior. Then $\| \cdot \|^{}_{U}$ and $\| \cdot \|^{}_{V}$ are equivalent norms on  $\cM^\infty(G)$. \qed
\end{itemize}
\end{lemma}

\begin{remark}
A measure $\mu \in \cM(G)$ is translation bounded if and only if $\mu*f\in \Cusp(G)$ for any $f\in \Ccsp(G)$, see \cite[Thm.~1.1]{FRSS-ARMA1} and \cite[Ch.~13]{FRSS-ARMA}, or \cite[Prop.~4.9.21]{FRSS-MoSt}.
For more on the question which measures are translation bounded, see \cite{FRSS-RT84,FRSS-S19} and references therein. \exend
\end{remark}


Let now $U\subseteq G$ be any relatively compact set having non-empty interior. It turns out that $(\cM^\infty(G), \| \cdot \|^{}_{U})$ can be identified with the dual space $(\Wspp(G), \op{\cdot}_{\Wspp, \varPhi})$ of the Wiener algebra $(\Wsp(G), \|\cdot\|^{}_{\varPhi})$. The identification of $\cM^\infty(G)$ and $\Wspp(G)$ is the content of the following well-known result, see for example \cite[Thm.~6.1]{FRSS-LRW74} or the remark in \cite[Thm.~14]{FRSS-fe92}, which we re-prove for a self-contained presentation.

\begin{theorem}\label{FRSS-thm:tbW}
For any  $\mu \in \cM^\infty(G)$ and any  $f \in \Wsp(G)$, we have that $f \in \Lisp(|\mu|)$ and $\mu*f \in \Cusp(G)$. Moreover, $f\mapsto \mu(f)$ defines a bounded linear functional on $\Wsp(G)$. Under this identification, we have
$\Wspp(G)  =  \cM^\infty(G)$.
\end{theorem}

\begin{beweis}
Fix some BUPU $\varPhi=(\phi_i)_{i\in I}$ of size~$U$ throughout the proof.
We will first show that $\Wspp(G)\subseteq\cM^\infty(G)$.  Note that $\Wspp(G)\subseteq \cM(G)$ follows from Theorem~\ref{FRSS-thm:W}(e).  In order to analyse the corresponding estimate further, let $L: \Wsp(G) \to \CC$ be any bounded linear functional. Fix any compact $K\subseteq G$ and take an arbitrary $f\in \Ccsp(G)$ such that $\supp(f)\subseteq K$. We then get, by Lemma ~\ref{FRSS-lem:normpropw}(b), for any $x\in G$ the estimate
\begin{displaymath}
\begin{split}
|L(T_xf)| &\,\le\, \op{L}_{\Wspp, \varPhi}\cdot \|T_xf\|^{}_{\varPhi}
\,\le\,  \op{L}_{\Wspp, \varPhi} \cdot C\cdot \|f\|^{}_{\infty} \ts .
\end{split}
\end{displaymath}
This finite upper bound on $|L(T_xf)|$ is uniform in $x\in G$. This shows that $L$ is translation bounded, by \cite[Thm~1.1]{FRSS-ARMA1}  together with \cite[Ch.~13]{FRSS-ARMA}, see also \cite[Prop.~4.9.21]{FRSS-MoSt} and \cite[Thm.~1.1]{FRSS-Fei77}.  We thus have $\Wspp(G)\subseteq\cM^\infty(G)$.

\smallskip

\noindent Conversely, consider any  $\mu \in \cM^\infty(G)$ and any  $f \in \Wsp(G)$. Then, by Remark~\ref{FRSS-rem:Wiscountable}, the function $f=\sum_{i\in I} f\phi_i$ can be assumed to be a countable sum without loss of generality. Using monotone convergence, we estimate
  \begin{displaymath}
  \begin{split}
  \abs\mu(\abs f)  &\,\le\,\sum_{i\in I}\abs\mu(|f\cdot\phi_i|)
    \,\le\,\sum_{i\in I}\|f\phi_i\|^{}_{\infty}\cdot\abs\mu(x_i+U)\\
   &\,\le\,\|\mu\|^{}_{U}\cdot\sum_{i\in I}\|f\phi_i\|^{}_{\infty}
    \,=\,\|\mu\|^{}_{U}\cdot \|f\|^{}_{\varPhi} \,<\, \infty \ts .
    \end{split}
  \end{displaymath}
which shows $f \in \Lisp(|\mu|)$. The inequality $\left|  \mu(f) \right| \leq    |\mu|(|f|)\le \|\mu\|^{}_{U} \cdot \| f\|^{}_{\varPhi}$ implies that the measure $\mu$ extends to a bounded linear functional on $\Wsp(G)$. Consequently, $\cM^\infty(G)\subseteq \Wspp(G)$, and the claim $\mu*f \in \Cusp(G)$ now follows from Lemma~\ref{FRSS-lem:convcu}.
\qed \end{beweis}

\begin{remark}
Using the amalgam space description, one can compute that $\Wspp(G)=\Wspp(\COsp,\lisp)(G)=\Wsp(\COsp',\lisp')(G)=\Wsp(\cMb,\lsp^\infty)(G)=\cM^\infty(G)$. The underlying duality argument has been generalised considerably in \cite[Thm.~2.8]{FRSS-fegr85} and \cite[Thm.~5.1]{FRSS-HGF87}. \exend
\end{remark}

The following result for finite measures will be used later on. The proof is identical to that of Theorem \ref{FRSS-thm:W}(g)  and therefore omitted.

\begin{lemma}\label{FRSS-lem:fw}
Let $\mu\in \cMb(G)$ and let $\varPhi$ be any BUPU. Then, there exists a finite constant $C$ such that, for all $f\in \Wsp(G)$ we have
\[
\|\mu*f\|^{}_{\Wsp, \varPhi} \, \le \,  C \cdot\|f\|^{}_{\Wsp, \varPhi} \cdot |\mu|(G) \ts .
\]
In particular, for all $f \in \Wsp(G)$ we have $\mu*f\in \Wsp(G)$. \qed
\end{lemma}


\omit{
\begin{beweis}
Take any BUPU $\varPhi$ with translation bounded norm $\|\cdot\|^{}_{\varPhi}$.
We argue as in the proof of Theorem \ref{FRSS-thm:W} (g) that
\begin{displaymath}
\begin{split}
\|f*\mu\|^{}_{\varPhi}&\,=\,\sum_{i\in I} \|(f*\mu)\cdot \phi_i\|^{}_{\infty}
\,=\, \sum_{i\in I} \sup_{x\in G} \left| \int_G f(x-y) \phi_i(x) \, \dd  \mu(y) \right|\\
&\,\le\, \sum_{i\in I} \int_G \|(T_yf)\cdot \phi_i\|^{}_{\infty} \dd |\mu|(y)
\,=\,\int_G \sum_{i\in I}  \|(T_yf)\cdot \phi_i\|^{}_{\infty} \dd |\mu|(y)\\
& \,=\, \int_G    \|T_yf\|^{}_{\varPhi}  \dd |\mu|(y)
\,\le\,  \int_G    C\cdot\|f\|^{}_{\varPhi}  \dd |\mu|(y) \,=\,  C \cdot\|f\|^{}_{\varPhi} \cdot |\mu|(G) \ts ,
\end{split}
\end{displaymath}
where we used non-negativity in the second line in order to exchange integration and summation.
\qed \end{beweis}
}

We now argue that the norm $\| \cdot \|^{}_{U}$ on $\mathcal M^\infty(G)=\Wspp(G)$ is equivalent to the canonical operator norm
\[
\op{\mu}_{\Wspp} \, = \, \sup \{ \left| \mu(f) \right| : f \in \Wsp(G) , \|f \|^{}_{\Wsp}\le 1 \}\ts .
\]
\index{$\op{\cdot}_{\Wsp}$}

\begin{proposition}\label{FRSS-prop:equiv}
  Let $\varPhi=(\phi_i)_{i\in I}$ be any BUPU of norm $M$, size $U$ and point family $X$ as in Definition~\ref{FRSS-def:BUPUCc}.
  We then have the estimates
  \begin{displaymath}
   \op{\cdot}_{\Wspp, \varPhi} \leq \| \cdot \|^{}_{U}  \leq  M \cdot \|\delta^{}_{X}\|^{}_{U-U} \cdot \op{\cdot}_{\Wspp, \varPhi} \ts .
   \end{displaymath}
\end{proposition}

\begin{beweis}
  Fix arbitrary $\mu \in \cM^\infty(G)$. 
  The lower inequality follows from
  \[
  \left|  \mu(f) \right| \, \le \, \| \mu \|^{}_{U} \cdot \| f \|^{}_{\varPhi} \ts 
  \]
  for every $f\in\Ccsp(G)$, see the proof of Theorem~\ref{FRSS-thm:tbW}.

  \noindent For the upper inequality, let $\varepsilon >0$ be arbitrary.
  Then, there exists some $t \in G$ such that $\| \mu \|^{}_{U}  \,\le\, |\mu| (t+U) +\varepsilon$.
  By inner regularity of $\mu$ and Urysohn's lemma,
  there exists some $g \in \Ccsp(G)$ with $0\leq g \leq 1_{t+U}$
  such that $|\mu|(t+U)  \le |\mu| (g) + \varepsilon$.
  By definition of the total variation measure, there exists some $h \in \Ccsp(G)$ with
  $|h| \leq g \leq 1_{t+U}$ and $|\mu|(g) \leq  |\mu(h)| +\varepsilon$.
  This gives
  \begin{displaymath}
    \| \mu \|^{}_{U} \,\leq\,  \left| \mu(h) \right| +3 \varepsilon
   \, \leq \, \op{\mu}_{\Wspp,\varPhi} \cdot \| h \|^{}_{\varPhi}  +3 \varepsilon \ts .
  \end{displaymath}
  Let us now estimate the norm of $h$.
  As $\supp(h) \subseteq t+U$ and $\| h \|_\infty <1$, we can use Lemma~\ref{FRSS-lem:BUPU-tb} to obtain
  \[
    \| h \|^{}_{\varPhi}  \,=\, \sum_{i \in I} \| h \phi_i \|^{}_{\infty} \,\leq\, M \cdot \|\delta^{}_{X}\|^{}_{U-U}    \,.
  \]
  We thus get
  \[
  \| \mu \|^{}_{U}\, \leq\, M \cdot \|\delta^{}_{X}\|^{}_{U-U} \cdot  \op{\mu}_{\Wspp,\varPhi}  +3 \varepsilon \,.
  \]
  Since $\varepsilon >0$ was arbitrary, the claim follows.
\qed \end{beweis}

Since the dual of a normed space is a Banach space, the following result is obvious; compare  \cite[p.~342]{FRSS-Fei77} and \cite{FRSS-BaakeMoody2004}.
\begin{theorem} If $U \subseteq G$ is any relatively compact set with non-empty interior, the vector space $(\cM^\infty(G), \| \cdot \|^{}_{U})$ is a Banach space. Moreover, different choices of $U$ define equivalent norms. \qed
\end{theorem}

\subsection{Mild measures}\label{FRSS-sec:mm}

To study the relation between measures and mild distributions as introduced in Section~\ref{FRSS-sect:appB}, we use the following terminology.

\begin{definition}\index{measure!mild} A (Radon) measure $\mu\in \cM(G)$ is called a \emph{mild} measure if there exists a mild distribution $\sigma \in \SOp(G)$ such that
\[
\mu(f) \, = \, \sigma(f) 
\]
for all $f \in  \KLsp(G)$.
We denote the space of mild measures on $G$  by $\cM^{}_{\mkern-2mu m}(G)$\index{$\cM^{}_{\mkern-2mu m}(G)$}, and use $\iota: \cM^{}_{\mkern-2mu m}(G)\to \SOp(G)$\index{$\iota$} for the canonical inclusion map.
\end{definition}

\begin{remark}\label{FRSS-rem:dense}
By the denseness of $\KLsp(G)=\Ccsp(G) \cap \So(G)$ in both spaces $(\Ccsp(G), \|\cdot \|^{}_{\infty})$ and $(\So(G),  \|\cdot \|^{}_{\So(G)})$, a mild measure $\mu\in \cM(G)$ and the corresponding mild distribution $\sigma\in \SOp(G)$ uniquely determine one another, and therefore $\iota$ is well defined and injective. \exend
\end{remark}

The following result states that any translation-bounded measure is mild, and that $\So(G)$ is integrable with respect to any translation-bounded measure. Note that there exist mild measures that are not translation-bounded measures, see Remark~\ref{FRSS-rem-agntb} below.
\begin{corollary}[{\cite[Thm.~B1]{FRSS-Fei80}}] \label{FRSS-cor1} We have $\cM^\infty(G)\subseteq \cM^{}_{\mkern-2mu m}(G)$. Moreover, for $\mu \in \cM^\infty(G)$, the following properties hold.
\begin{itemize}
\setlength\itemsep{2pt}
  \item[(a)] For all $f \in \So(G)$, we have $f \in \Lisp(|\mu|)$.
  \item[(b)] For all $f \in \So(G)$, we have $(\iota(\mu))(f) \, = \, \int_{G} f(t) \, {\rm d} \mu(t)$.
\end{itemize}
\end{corollary}
\begin{beweis}
Consider an arbitrary $\mu\in \cM^\infty(G)$. Then, (a) follows from Theorem~\ref{FRSS-thm:tbW} as $\So(G)\subseteq \Wsp(G)$.

\noindent As to (b), note that the mapping $f\mapsto \mu(f)$ is a bounded linear functional on $\Wsp(G)$ by Theorem~\ref{FRSS-thm:tbW}. Hence, its restriction to $\So(G) \subseteq \Wsp(G)$ is a mild distribution because the embedding of $\So$ into $\Wsp$ is continuous by Proposition~\ref{FRSS-prop:S0Wdense}. This shows $\mu \in \cM^{}_{\mkern-2mu m}(G)$ and (b).
\qed \end{beweis}

Next, we show that positive mild distributions are translation-bounded measures, and hence also mild measures.
We will give a proof that is analogous to the tempered distribution case \cite[Thm.~1.2.7]{FRSS-H03}. A different proof can be found in  \cite[Thm.~2.3]{FRSS-H89}.

Recall that a translation-bounded measure $\mu\in \cM^\infty(G)$ is \emph{positive}\index{measure!positive}\index{positive!measure} if  $\mu(f)\ge0$ for all $f\in \Wsp(G)$ with $f\ge0$. This is equivalent to $\mu(f)\ge0$ for all $f\in \Ccsp(G)$ with $f\ge0$, which is the usual definition of positivity for measures.

\begin{proposition}[{\cite[Prop.~B4]{FRSS-Fei80}}]\label{FRSS-prop:Feipos}
The set of positive mild distributions coincides with the set of positive translation-bounded measures, that is,
\begin{displaymath}
\left\{ \sigma \in \SOp(G): \sigma \ge 0 \right\} \,=\, \iota\left(\left\{ \mu \in \cM^\infty(G): \mu\ge 0 \right\}\right) \ts .
\end{displaymath}
\end{proposition}

\begin{remark}[Which positive measures are mild?]
As a corollary, we observe that a positive measure is mild if and only if it is translation bounded. In order to see this, note first that any translation-bounded measure is mild by Corollary~\ref{FRSS-cor1}. Assume now that $\mu$ is a mild measure that is positive as a mild distribution. We show that $\mu$ is a positive measure, which yields the claim by Proposition~\ref{FRSS-prop:Feipos}. Consider any non-negative $f\in \Cc(G)$. Take a sequence $(f_n)_{n\in \NN}$ of non-negative functions $f_n\in \KLsp(G)\subseteq \So(G)$ such that $\|f-f_n\|^{}_\infty\to 0$, compare the proof of Lemma~\ref{FRSS-lem:AGC0}. As $\supp(f)$ is compact, we can assume that $f$ and all $f_n$ are supported within a common compact set $K$. Now the claim follows from the estimate
\begin{displaymath}
\mu(f)\,\ge\, \mu(f_n)-|\mu(f)-\mu(f_n)| \,\ge\,  0-C^{}_K \|f-f_n\|^{}_\infty 
\end{displaymath}
for some positive finite constant $C_K$ and for all $n$, where we used that $\mu$ is a measure. \exend
\end{remark}

\begin{beweis}[Proof of Proposition~\ref{FRSS-prop:Feipos}]
``$\subseteq$'': Consider any positive $\sigma\in \SOp(G)$. Fix any compact $K\subseteq G$ with non-empty interior and take some non-negative $h\in \KLsp(G)\subseteq \So(G)$ such that $h|_K\ge 1$. For the existence of such $h$, see the proof of \cite[Lemma~4.3 (iv)]{FRSS-Jak}. By positivity we have $\sigma(h)\ge0$. In particular, $\sigma(h)$ is real.

\noindent Next, consider any real-valued $f\in \KLsp(G)\subseteq \So(G)$ such that $\mathrm{supp}(f)\subseteq K$. Then $\|f\|^{}_{\infty} \cdot h - f\ge 0$, which implies $\|f\|^{}_{\infty} \cdot \sigma(h)  -  \sigma(f) \ge   0$
by positivity. Similarly, $\|f\|^{}_{\infty} \cdot h + f\ge 0$, which implies 
$\|f\|^{}_{\infty} \cdot \sigma(h)  +  \sigma(f) \ge 0$.
In particular, $\sigma(f)$ is real, and $\left| \sigma(f) \right| \le  \|f\|^{}_{\infty} \cdot \sigma(h)$.

\noindent Finally, consider an arbitrary $f\in \KLsp(G)\subseteq \So(G)$ such that $\mathrm{supp}(f)\subseteq K$. Then ${\rm Re}(f), {\rm Im}(f)\in \So(G)$ by \cite[Cor.~4.2 (iv)]{FRSS-Jak}, and ${\rm Re}(f), {\rm Im}(f)$ are both supported inside $K$. We thus have ${\rm Re}(f), {\rm Im}(f) \in \Ccsp(G) \cap \So(G)=\KLsp(G)$. We thus have
\[
\sigma(f)\,=\, \sigma({\rm Re}(f)) \,+\, \ii \cdot \sigma({\rm Im}(f)) \ts,
\]
where both $\sigma({\rm Re}(f)), \sigma({\rm Im}(f))$ are real by the previous argument. Moreover, we can estimate 
\begin{displaymath}
\begin{split}
|\sigma(f)|\, &\le  \, \left| \sigma({\rm Re}(f)) \right| \,+\, \left| \sigma({\rm Im}(f)) \right| \\[1mm]
     & \le \, \|{\rm Re}(f)\|^{}_{\infty} \cdot \sigma(h)+ \|{\rm Im}(f)\|^{}_{\infty} \cdot \sigma(h) \\[1mm]
     & \le \, \|f\|^{}_{\infty} \cdot \sigma(h)+ \|f\|^{}_{\infty} \cdot \sigma(h)=2 \|f\|^{}_{\infty} \cdot \sigma(h) \ts .
\end{split}
\end{displaymath}
Hence, defining $\mu(f)=\sigma(f)$ for $f\in \KLsp(G)$ yields a mild measure $\mu\in \cM^{}_{\mkern-2mu m}(G)$ such that $\iota(\mu)=\sigma$.
Translation boundedness of $\mu$ follows from translation boundedness of $\sigma$, as explained in Remark~\ref{FRSS-rem:tbS0}. Indeed we have by positivity that
\[
\| \mu \|^{}_{K} \,=\,\sup_{t \in G} \mu(t+K) \,\leq\, \sup_{t \in G} \mu(T_th) \,\leq\,  \op{\sigma}_{\SOp} \cdot \| h \|^{}_{\So} \,<\, \infty \ts .
\]
Now, $\mu$ inherits positivity from $\sigma$ by the denseness of $\So(G)$ in $\Wsp(G)$; see Proposition~\ref{FRSS-prop:S0Wdense}.

\smallskip
\noindent ``$\supseteq$'': Fix any positive $\mu\in \cM^\infty(G)$. Then, $\mu\in \cM^{}_{\mkern-2mu m}(G)$ by Corollary~\ref{FRSS-cor1}. Let $\sigma=\iota(\mu)\in \SOp(G)$.
Then $\sigma(f) = \mu(f) \ge  0$
for all $f\in \So(G)$ with $f\ge0$ by positivity of $\mu$. Hence $\sigma\ge0$.
\qed \end{beweis}

Recall that a measure $\mu\in \cM(G)$ is \emph{positive definite}\index{positive~definite!measure}\index{measure!positive~definite} if $\mu(\widetilde f*f)\ge 0$ for all $f\in \Ccsp(G)$; compare~\cite[Ch.~4]{FRSS-ARMA1}. Corollary~\ref{FRSS-cor-pdme} below will imply that any positive definite measure is a mild measure.
Here we show that, for mild measures, the two definitions of positive definiteness as a measure and as a mild distribution are equivalent.

\begin{proposition}\label{FRSS-pd-char} For a mild measure $\mu$ on $G$, the following properties are equivalent.
\begin{itemize}
\setlength\itemsep{2pt}
  \item[(i)] $\iota(\mu)$ is a positive-definite mild distribution on $G$.
  \item[(ii)] $\mu(f) \geq 0$ for all positive-definite $f \in \KLsp(G)$.
  \item[(iii)] $\mu$ is a positive-definite measure on $G$.
\end{itemize}
\end{proposition}

\begin{beweis}
(i) $\Longleftrightarrow$ (ii) is a particular case of Corollary~\ref{FRSS-cor-PD}.

\noindent (ii) $\Longrightarrow$ (iii) follows from $\widetilde f*f \in \KLsp(G)$ for all $f \in \Ccsp(G)$.

\noindent (iii) $\Longrightarrow$ (ii) follows for example from \cite[Cor.~4.2]{FRSS-ARMA1}.
\qed \end{beweis}

\begin{remark}[positive-definite mild distributions and measures]
Assume that $\sigma\in \SOp(G)$ is a mild distribution. Then, $\sigma$ is a  translation-bounded, positive-definite measure if and only if the mild distribution $\widehat \sigma$ is a positive measure that is Fourier transformable in AG sense, see Theorem~\ref{FRSS-t1a} below. As it is difficult to decide whether a positive translation-bounded measure is AG transformable, it is also difficult to decide whether a positive-definite mild distribution is a measure. \exend
\end{remark}

\smallskip

Recall that any mild distribution has a mild distribution as its Fourier transform by Theorem~\ref{FRSS-thm:Fdist}. Thus, any mild measure admits a Fourier transform being a mild distribution. In particular, any translation-bounded measure admits a Fourier transform as a mild distribution.

\begin{corollary} Let $\mu\in\cM^{}_{\mkern-2mu m}(G)$. Then there exists a unique mild distribution $\tau \in \SOp(\widehat G)$ such that for every  $f \in \KLsp(G)$ we have
\begin{displaymath}
\mu(f) \, = \, \tau(\widecheck{f}) \ts .
\end{displaymath}
We call $\tau=\reallywidehat{\iota(\mu)}\in \SOp(\widehat G)$
the \emph{Fourier transform of $\mu$}\index{Fourier~transform!of~a~mild~measure}. \qed
\end{corollary}

If $\mu\in \cM^{}_{\mkern-2mu m}(G)$ with mild distribution $\sigma=\iota(\mu)\in \SOp(G)$, we will sometimes write $\widehat \mu$ instead of $\widehat\sigma=\reallywidehat{\iota(\mu)}$, in order to simplify notation.

\section{Fourier-transformable measures}\label{FRSS-sec:FTm}

We review the Fourier transform of a measure as introduced by Argabright and Gil de Lamadrid \cite{FRSS-ARMA1, FRSS-ARMA}. The Fourier transform of translation-bounded measures is relevant for diffraction analysis, as will be explained in Section~\ref{FRSS-sec:dtr}.

\subsection{AG-transformable measures}

By Bochner's theorem \cite[Thm.~1.4.3]{FRSS-RUD2}, any positive-definite continuous function is the inverse Fourier transform of a finite positive measure. This has been extended to positive-definite measures. As we will discuss below, positive-definite measures are inverse Fourier transforms of positive measures~\cite[Thm.~4.5]{FRSS-BF75}.

In order to extend the latter theory, Argabright and Gil de Lamadrid asked when a measure admits a Fourier transform that is a measure \cite[p.~8]{FRSS-ARMA1}. Here, we give a definition that is equivalent to their notion of  Fourier transform, which better reflects the functional analysis viewpoint; compare \cite[Prop.~3.9]{FRSS-CRS} and \cite[Prop.~4.9.9]{FRSS-MoSt}.

\begin{definition}[{\cite[Prop.~3.9]{FRSS-CRS}}]\label{FRSS-def:AGtrans}\index{Fourier~transformable!AG--transformable}\index{measure!AG--transformable} A measure $\mu\in\cM(G)$ is called \emph{AG transformable}~\index{Fourier~transformable!AG~transformable} if there exists a measure $\nu\in \cM(\widehat{G})$ such that for all $f \in \KLsp(G)$ we have $\widecheck{f} \in \Lisp(|\nu|)$ and
\[
\mu(f) \, =\, \nu(\widecheck{f}) \ts .
\]
The measure $\nu$ is denoted by $\widehat\mu$ and is called the \emph{Fourier transform of $\mu$}\index{Fourier~transform!of~a~measure}. The space of AG-transformable measures on $G$ is denoted by $\cM^{}_{T}(G)$\index{$\cM^{}_{T}(G)$}.  \qed
\end{definition}

\begin{example} The Poisson summation formula implies that the measure $\delta_{\ZZ}\in \cM^\infty(\RR)$ is AG transformable and satisfies $\reallywidehat{\delta_{\ZZ}} \, = \delta_{\ZZ}$, where we identify $\widehat \RR$ with $\RR$.
More generally, if $G$ is an LCA group and $H$ is a closed subgroup of $G$, consider the pushforward $\mu_H$ of the Haar measure $m_H$ on $H$ to $G$, which is for $f\in \Cc(G)$ defined by
\[
\mu_H(f) \,= \, \int_{H} f(x) \, \dd m_H(x) \ts .
\]
Then $\mu_H$ is a translation bounded measure, and $\mu_H$ is AG transformable. Its Fourier transform $\widehat{\mu_H}$ is the pushforward to $\widehat{G}$ of a Haar measure on the \emph{annihilator}
\[
H^0_{} \, = \, \{ \chi \in \widehat{G} : \chi(x)= 1 \text{ for all } x \in H \}
\]
of $H$ in $\widehat G$, see \cite[Prop.~6.2.]{FRSS-ARMA}. \exend
\end{example}

\medskip

Some properties of AG-transformable measures are discussed in \cite{FRSS-ARMA1} or in \cite[Ch.~4.9]{FRSS-MoSt}. As the AG transform $\nu$ is a translation-bounded measure \cite[Thm.~2.5]{FRSS-ARMA1}, we could assume $\nu\in\cM^\infty(\widehat G)$ in the definition of AG transform without loss of generality. In this case, Section~\ref{FRSS-sec:tbmintro} gives that the integrability condition in the definition of AG transformability is unnecessary. Indeed, by Corollary~\ref{FRSS-cor1}, if $\nu \in \cM^\infty(\widehat{G})$ for all $f \in \KLsp(G) \subseteq \So(G)$, we have $|\widecheck{f}| \in \Lisp(\nu)$.
We thus have the following simpler characterisation of AG transformability.

\begin{lemma}\label{FRSS-lem:AGtrans} A measure $\mu\in\cM(G)$ is AG transformable if and only if there exists a translation-bounded measure $\nu\in \cM^\infty(\widehat{G})$ such that
\[
\mu(f) \, = \, \nu(\widecheck{f}) 
\]
holds for all $f \in \KLsp(G)$. \qed
\end{lemma}

We now state the characterisation of AG transformability in terms of mild distributions due to Feichtinger.

\begin{theorem}[{\cite[Thm.~C1(i)]{FRSS-Fei80}}] \label{FRSS-t1a}  Consider $\mu\in \cM(G)$. Then, $\mu\in \cM^{}_{T}(G)$  if and only if $\mu\in \cM^{}_{\mkern-2mu m}(G)$ and $\reallywidehat{\iota(\mu)}\in \iota\big(\cM^\infty(\widehat{G})\big)$.
Moreover, we have
\[
\reallywidehat{\iota(\mu)} \, =\, \iota(\widehat{\mu}) 
\]
in this case.
\end{theorem}

\begin{beweis}
``\emph{only if part}'': Assume that $\mu\in \cM(G)$ is AG transformable. Then $\widehat{\mu}\in \cM^\infty(\widehat G)$ by \cite[Thm.~2.5]{FRSS-ARMA1}. Thus, by Corollary~\ref{FRSS-cor1} we have $\widehat{\mu}\in\cM^{}_{\mkern-2mu m}(\widehat G)$.
Consider now the mild distribution
\[
\sigma\, =\,  \reallywidehat{\iota(\widehat{\mu})}^\dag\in \SOp(G) \ts.
\]
Then for all $f \in \KLsp(G)$ we have
\[
\mu(f) \, = \, \widehat{\mu}(\widecheck{f}) \, =\,  \iota(\widehat{\mu}(\widecheck{f})) \, = \, \reallywidehat{\iota(\widehat{\mu})}(\widecheck{\widecheck{f}}) \, =\, \sigma(f) \ts ,
\]
where we used $\widecheck{\widecheck{f}}=f^\dag$.
Hence, $\mu\in \cM^{}_{\mkern-2mu m}(G)$ and $\iota(\mu)=\sigma$. Noting $\widehat{\sigma}= \iota(\widehat{\mu})$ where $\widehat\mu\in \cM^\infty(\widehat G)$, we infer
$\reallywidehat{\iota(\mu)} = \widehat{\sigma}  = \iota(\widehat{\mu})$.
Therefore $\reallywidehat{\iota(\mu)}\in \iota\big(\cM^\infty(\widehat{G})\big)$.

\smallskip

\noindent
``\emph{if part}'': Consider $\sigma=\iota(\mu)\in \SOp(G)$. Since $\widehat \sigma=\iota(\nu)$ for some $\nu\in \cM^\infty(\widehat G)$ by assumption, we have for all $f \in \KLsp(G)$ that
\[
\mu(f)\, =\, \sigma(f)\, =\, \widehat{\sigma}(\widecheck{f})\, =\, \nu(\widecheck f)  \ts,
\]
which completes the proof.
\qed \end{beweis}

\begin{remark}\label{FRSS-rem-agntb}
There exist AG-transformable measures which are not transla\-tion-bounded measures;  see~\cite[Prop.~7.1]{FRSS-ARMA1} and \cite[Prop.~3.6]{FRSS-BS22}. In particular, any such measure is mild but not translation bounded. \exend
\end{remark}

We finally discuss the connection between positive definiteness and AG transformability. The following result is an immediate consequence of Corollary~\ref{FRSS-lem:ppdS0} and Theorem~\ref{FRSS-t1a}.

\begin{lemma} A mild distribution $\sigma$ is positive definite if and only if $\widehat{\sigma}$ is a positive, translation-bounded measure. \qed
\end{lemma}

We have the following characterisation of positive definiteness. It shows in particular that any positive-definite measure is mild.
\begin{corollary}\label{FRSS-cor-pdme} Let $\mu\in \cM(G)$ be a measure. Then, the following properties are equivalent.
\begin{itemize}
\setlength\itemsep{2pt}
  \item[(i)] $\mu$ is a positive-definite measure.
  \item[(ii)] $\mu \in \cM^{}_{T}(G)$ and $\widehat{\mu}$ is a positive measure.
  \item[(iii)] $\mu\in \cM^{}_{\mkern-2mu m}(G)$ and $\reallywidehat{\iota(\mu)}\in \SOp(\widehat G)$ is positive.
  \item[(iv)] $\mu\in \cM^{}_{\mkern-2mu m}(G)$ and $\iota(\mu)\in \SOp(G)$ is positive definite.
\end{itemize}
\end{corollary}

\begin{beweis}
(i) $\Longrightarrow$ (ii) is stated in~\cite[Thm.~4.1]{FRSS-ARMA1}; see also \cite[Thm.~4.5]{FRSS-BF75}.

\noindent (ii) $\Longrightarrow$ (iii) follows from Theorem~\ref{FRSS-t1a} and Proposition~\ref{FRSS-prop:Feipos}.

\noindent (iii) $\Longrightarrow$ (iv) follows from Corollary~\ref{FRSS-lem:ppdS0}.

\noindent (iv) $\Longrightarrow$ (iv) follows from Proposition~\ref{FRSS-pd-char}.
\qed \end{beweis}

\begin{remark}[Positive mild distributions and AG transformability]\label{FRSS-rem:pdAG}
Note that positivity of $\mu\in \cM^\infty(G)$ does not imply $\mu\in \cM^{}_{T}(G)$. Indeed, if $\mu\in \cM^{}_{T}(G)\cap \cM^\infty(G)$, then $\mu$ is weakly almost periodic, which implies that $\mu$ is pure-point diffractive \cite[Thm.~7.5]{FRSS-LS2019}. Diffraction will be discussed in Section~\ref{FRSS-sec:dtr} below. There are many examples of Delone sets $\varLambda\subset \mathbb R$ such that the corresponding Dirac comb $\delta^{}_{\!\varLambda}$ is not pure-point diffractive. Explicit examples may be taken from the class of model set combs $\omega_h\in \cM^\infty(G)$, which is discussed in Section~\ref{FRSS-sec:ppdpsr}. Indeed, one may take any positive $h\in \Ccsp(H)\setminus \KLsp(H)$. Then, $\omega_h\in \cM^{}_{T}(G)$ would imply $\widehat h\in \Lisp(\widehat H)$ by \cite[Lemma~5.2]{FRSS-RS17}, which is contradictory.  \exend
\end{remark}

\subsection{Doubly AG-transformable measures}

According to \cite[Thm.~2.1]{FRSS-ARMA1}, the Fourier transform $\mu\mapsto \widehat \mu$ is an injective mapping from $\cM^{}_{T}(G)$ to $\cM^\infty(\widehat G)$. On the other hand, $\mu\in \cM^{}_{T}(G)$ need not be translation bounded. In order to obtain a theory where the Fourier transform is a bijective map between appropriate spaces, either the space of transformable measures has to be restricted, or the notion of Fourier transform has to be extended beyond measures. The first approach has been taken in \cite{FRSS-ARMA1}. Let us denote\footnote{The space $\cM_{dT}(G)$ is denoted by $\vartheta(G)$ in \cite[p.~21]{FRSS-ARMA1}.} by $\cM_{dT}(G)$ the space of \emph{doubly AG-transformable measures}\index{Fourier~transformable!doubly~AG--transformable}\index{measure!doubly~AG--transformable} on $G$, that is,
\begin{displaymath}\index{$\cM_{dT}(G)$}
\cM_{dT}(G)\,=\,\{\mu\in \cM(G): \mu\in  \cM^{}_{T}(G), \widehat \mu\in\cM^{}_{T}(\widehat G)\} \ts .
\end{displaymath}
A measure in $\cM_{dT}$ is also called \emph{twice Fourier transformable}.
Properties of $\cM_{dT}(G)$ have been discussed in \cite[Chs.~4,7]{FRSS-ARMA1}. 
Here we review the following result.

\begin{proposition}\label{FRSS-Prop-TFT} Consider any $\mu \in \cM_{dT}(G)$, and let $f \in L^1(G)$ be such that $\widecheck{f} \in L^1(|\widehat{\mu}|)$. Assume that there exists an open $0$-neighbourhood $U\subseteq G$ with the following properties.
\begin{itemize}
\item[(i)] For each $x \in U$ we have $(T_xf)^\dag \in L^1(|\mu|)$.
\item[(ii)] The function $x \to \mu((T_xf)^\dag)$ is continuous for all $x\in U$.
\end{itemize}
We then have
\[
\mu(f)= \widehat{\mu}(\widecheck{f}) \ts .
\]
\end{proposition}

\begin{beweis}
This follows from \cite[Thm.~3.3]{FRSS-ARMA1} and \cite[Thm.~1.2]{FRSS-ARMA1}, where the latter theorem applies as $\cM_{dT}(G) \subseteq \cM^\infty(G)$.
\qed \end{beweis}

Motivated by the above result, we restrict to the space $\WOsp(G)\subseteq \Wsp(G)$\index{$\WOsp(G)$} defined by
\[
\WOsp(G)\,=\,\{h \in \Wsp(G): \widehat{h} \in \Wsp(\widehat G)\} 
\]
and equipped with norm $\|f\|^{}_{\WOsp(G)}=\|f\|^{}_{\Wsp(G)}+ \|\widehat f\|^{}_{\Wsp(\widehat G)}$. Note that  we have $\So(G)\subseteq \WOsp(G)$ and that $\WOsp(G)$ is a Segal algebra; compare~\cite{FRSS-Bu81, FRSS-Fei81, FRSS-Jak}. In fact, $\So(G)=\WOsp(G)$ if and only if $G$ is compact or discrete \cite[Thm.~2]{FRSS-Los80}. Observe that the assumptions in Proposition~\ref{FRSS-Prop-TFT} are satisfied for every $\mu \in \cM_{dT}(G)$ and for every $f\in \WOsp(G)$ by Lemma~\ref{FRSS-lem:convcu}. The resulting identity is called a generalized Poisson summation formula by Argabright and Gil de Lamadrid \cite[Thm.~3.3]{FRSS-ARMA}.
\begin{theorem}\label{FRSS-dT-Wo} Let $\mu \in \cM_{dT}(G)$. We then have
$\mu(f)= \widehat{\mu}(\widecheck{f})$ for every $f \in \WOsp(G)$. \qed
\end{theorem}

The space $\cM_{dT}(G)$ has a simple description in terms of translation boundedness. Argabright and Gil de Lamadrid remarked that, for a transformable measure $\mu\in \cM^{}_{T}(G)$, the necessary condition $\mu \in \cM^\infty(G)$ might be sufficient for double transformability \cite[p.~50]{FRSS-ARMA1}. Indeed, sufficiency has been proved by  Feichtinger \cite{FRSS-Fei79a, FRSS-Fei80}. This result subsumes some more recent results on double transformability from \cite{FRSS-CRS,FRSS-S19}.

\begin{theorem}[{\cite[Thm.~C1(ii)]{FRSS-Fei80}}]\label{FRSS-t1b}  Let $\mu\in \cM(G)$. Then, the following properties are equivalent.
\begin{itemize}
\setlength\itemsep{2pt}
\item[(i)] $\mu\in \cM_{dT}(G)\,$.
\item[(ii)] $\mu\in \cM^{}_{T}(G)\cap \cM^\infty(G)\,$.
\item[(iii)] $\mu\in \cM^\infty(G)$ and there exists $\nu\in \cM^\infty(\widehat G)$ such that $\mu(f)=\nu(\widecheck f)$ for all $f\in \So(G)$. In this case, we have $\nu=\widehat\mu$.
\item[(iv)] $\mu\in \cM^\infty(G)$ and there exists $\nu\in \cM^\infty(\widehat G)$ such that $\mu(f)=\nu(\widecheck f)$ for all $f\in \WOsp(G)$. In this case, we have $\nu=\widehat\mu$.
\end{itemize}
\end{theorem}

\begin{remark}
Note that (iii) does not contain an integrability condition such as the definition~\ref{FRSS-def:AGtrans} for AG transformability. In fact, it suffices to check the condition in (iii) on a dense subset of $(\So(G), \|\cdot \|^{}_{\So(g)})$ such as $\KLsp(G)$, $\LKsp(G)$, or as the Schwartz--Bruhat functions on $G$; compare Lemma~\ref{FRSS-lem:dtcrit}. For calculations on doubly transformable measures, this means that one may freely switch between the AG transform, the mild distribution transform, and the tempered distribution transform; compare Proposition~\ref{FRSS-prop:tddt}. \exend
\end{remark}

\begin{remark}
Restricting to Euclidean space, it can be shown that $\cM_{dT}(\mathbb R^d)$ is the dual of the Banach space of test functions $\Wisp(\mathbb R^d)=\Wsp(\mathbb R^d)+\mathcal F\Wsp(\mathbb R^d)$, equipped with the quotient norm
\[
\|f\|^{}_\Wisp \, := \, \inf_{f=f_1+f_2}(\|f_1\|^{}_{\Wsp(G)}+\|\widecheck f_2\|^{}_{\Wsp(\widehat{G})}) \ts;
\]
see  \cite[Prop.~9]{FRSS-Fei25}. More generally, the latter reference describes Fourier invariant Banach spaces of test functions beyond $\So(G)$. These lead to Fourier invariant spaces of distributions that are subspaces of mild distributions. \exend
\end{remark}

\begin{beweis}[Proof of Theorem~\ref{FRSS-t1b}]
\noindent (i) $\Longrightarrow$ (ii): holds because the AG transform satisfies $\widehat\mu\in \cM^\infty(\widehat G)$ by definition.

\smallskip

\noindent (ii) $\Longrightarrow$ (iii):   Obviously, $\nu=\widehat\mu\in \cM^\infty(\widehat G)$ satisfies the equation in (iii) for all $f\in \KLsp(G)$. We have $\mu\in\cM^{}_{\mkern-2mu m}(G)$ and $\widehat \mu \in\cM^{}_{\mkern-2mu m}(\widehat G)$ by Corollary~\ref{FRSS-cor1}, as $\mu \in \cM^\infty(G)$ and $\widehat \mu\in \cM^\infty(\widehat G)$ by assumption. By denseness, the equation holds for all $f\in \So(G)$; compare Remark~\ref{FRSS-rem:dense}.

\smallskip

\noindent (iii) $\Longrightarrow$ (i): We have $\mu\in \cM^{}_{T}(G)$ by Lemma~\ref{FRSS-lem:AGtrans} as $\KLsp(G) \subseteq \So(G)$. Next, consider an arbitrary $\varphi \in \KLsp(\widehat{G}) \subseteq \So(\widehat G)$. Defining $f=\widehat{\varphi} \in \So(G)$, we have by assumption that
$\nu(\varphi)= \nu(\widecheck{f})=\mu(f)= \mu^\dag(\widecheck{\varphi})$.
As the measure $\mu^\dag$ inherits translation boundedness from $\mu$, we conclude $\nu\in\cM^{}_{T}(\widehat G)$ from Lemma~\ref{FRSS-lem:AGtrans}. This proves (i).
\smallskip

\noindent (i) $\Longrightarrow$ (iv): follows from Theorem~\ref{FRSS-dT-Wo}.

\smallskip

\noindent (iv) $\Longrightarrow$ (iii): is obvious, since $\So(G) \subseteq \WOsp(G)$.
\qed \end{beweis}

As a consequence of Theorem~\ref{FRSS-t1b}, we obtain the following useful criterion for double transformability of a measure.

\begin{lemma}\label{FRSS-lem:dtcrit} Let $\mu \in \cM^\infty(G)$ and $\nu \in \cM^\infty(\widehat{G})$. Then, $\mu\in \cM_{dT}(G)$ and $\nu = \widehat\mu$ if and only if
\[
A_{\mu,\nu} \, = \,  \{ f \in \So(G) : \mu(f)= \nu(\widecheck{f}) \}
\]
is a dense subspace of $(\So(G), \|\cdot\|^{}_{\So(G)})$.
\end{lemma}

\begin{beweis}
``\emph{only if part}'': We have $\KLsp(G)\subseteq A_{\mu,\nu}$ by Lemma~\ref{FRSS-lem:AGtrans}. Denseness of $A_{\mu,\nu}$ now follows from denseness of $\KLsp(G)$, which is stated in Theorem~\ref{FRSS-thm:S0} (d).

\smallskip

\noindent ``\emph{if part}'': We have $\mu\in \cM^{}_{\mkern-2mu m}(G)$ and $\nu\in \cM^{}_{\mkern-2mu m}(\widehat G)$ by  Corollary~\ref{FRSS-cor1}. As the mild distributions $\iota(\mu)\in \SOp(G)$ and $\reallywidehat{\iota(\nu)}\in \SOp(G)$ agree on $A_{\mu,\nu}$, which is dense in $\So(G)$ by assumption, we have
$\iota(\mu) =\reallywidehat{\iota (\nu)}\in \SOp(G)$.
Noting Corollary~\ref{FRSS-cor1}(b), we infer $\mu\in \cM_{dT}(G)$ from Theorem~\ref{FRSS-t1b}~(iii).
\qed \end{beweis}

The above criterion allows to characterise double AG transformability in terms of tempered distributions.

\begin{proposition}\label{FRSS-prop:tddt} Let $\mu\in \cM^\infty(G)$. Then, $\mu\in  \cM_{dT}(G)$  if and only if its Fourier transform $\widehat \mu$ as a tempered distribution is a translation-bounded measure.
\end{proposition}

\begin{beweis}
``\emph{only if part}'': Since $\widehat\mu \in \cM^{}_{T}(\widehat G)$, it defines a tempered distribution by~\cite[Thm.~7.2]{FRSS-ARMA1}.
Alternatively, we might observe $\widehat\mu\in \cM^{}_{\mkern-2mu m}(\widehat G)$ from Corollary~\ref{FRSS-t1a}. In particular, $\iota(\widehat\mu)\in \SOp(\widehat G)$ is a tempered distribution \cite[Thm.~B1]{FRSS-Fei80}; compare Remark~\ref{FRSS-rem:sbmdis}.  Translation boundedness of $\widehat \mu$ is a consequence of AG transformability.

\smallskip

\noindent ``\emph{if part}'': We have $\mu\in\cM^\infty(G)$ and $\widehat\mu\in \cM^\infty(\widehat G)$ by assumption. Moreover $A_{\mu,\widehat\mu}$ in Lemma~\ref{FRSS-lem:dtcrit} contains the Schwartz--Bruhat functions $\Scsp(G)$; compare Remark~\ref{FRSS-rem:fce}. Now the claimed double AG transformability follows from Lemma~\ref{FRSS-lem:dtcrit} as $\Scsp(G)$ is dense in $\So(G)$; compare Remark~\ref{FRSS-rem:fce}.
\qed \end{beweis}

\subsection{Fourier--Bohr coefficients}

Assume that $\mu\in\cM(G)$ is twice AG transformable. Then, the discrete part of its transform $\widehat\mu\in \cM(\widehat G)$ at $\chi\in\widehat G$ can be computed as a certain average of the measure $\overline{\chi}\mu$.

For ease of presentation, we restrict here to $\sigma$-compact $G$. In this situation, averaging often uses so-called van Hove sequences $(A_n)_{n\in\NN}$ as introduced by Schlottmann~\cite{FRSS-Martin2}. These are sequences of compact subsets of $G$ of positive Haar measure such that $m_G(\partial^K A_n)=o(m_G(A_n))$ as $n\to\infty$, where $K\subseteq G$ is an arbitrary compact set. Here, the van Hove boundary $\partial^K A$ is defined by
\begin{displaymath}
\partial^K A \,=\, ((K+A) \cap \overline{A^c}) \,\cup\, ((-K+\overline{A^c})\cap A)\ts .
\end{displaymath}
Van Hove sequences are tailored for the estimate $|\mu(\partial^K A_n)|=o(m_G(A_n))$ for any translation-bounded measure $\mu$.

In Euclidean space, any sequence of closed $n$-balls constitutes a van Hove sequence. 
On general LCA groups, averaging can be done by using van Hove nets instead of van Hove sequences. Any F{\o}lner net yields a van Hove net by a simple thickening procedure \cite{FRSS-PRS22}.

Assume now that $G$ admits a van Hove sequence $\cA=(A_n)_{n\in\NN}$. If for a measure $\mu\in \cM(G)$ and for $\chi\in\widehat G$ the limit
\begin{equation}\label{FRSS-eq:FBC}
a_{\chi}^\cA(\mu) \,=\, \lim_{n\to\infty} \frac{1}{m_G(A_n)} \int_{A_n} \overline{\chi(t)}\,  \dd \mu(t)
\end{equation}
exists, it is called the \emph{Fourier--Bohr (FB) coefficient}\index{Fourier--Bohr~coefficient} of $\mu$ at $\chi$ along $\cA$.

\begin{proposition}\label{FRSS-prop:FBC}
Assume that $G$ admits a van Hove sequence $\cA$, and let $\mu\in \cM_{dT}(G)$.  Then, the Fourier--Bohr coefficient $a_{\chi}^\cA(\mu)$ exists for every $\chi\in\widehat G$ and satisfies
\[
\widehat\mu(\{ \chi \}) = a_{\chi}^\cA(\mu) \ts .
\]
In particular, the limits are uniform in translates of $\cA$ and independent of the choice of $\cA$.\qed
\end{proposition}

For $G=\RR^d$ and for a sequence of increasing cubes $\cA$, the above result is proved in \cite[Thm.~3.2]{FRSS-Hof1} using tempered distributions. The general case is proved in \cite[Prop.~3.14]{FRSS-CRS} by arguments similar to \cite{FRSS-Hof1}. It is also proved in \cite[Thm.~4.10.14]{FRSS-MoSt} by an argument based on almost periodicity.  In fact, the above result already follows from \cite[Thm.~11.3]{FRSS-ARMA}.

Recall that any twice transformable measure is weakly almost periodic \cite[Thm.~4.10.4]{FRSS-MoSt}.
Existence of the Fourier--Bohr coefficients continues to hold for translation-bounded measures $\mu$ that are weakly almost periodic; see Def.~2.18 and Lemma 2.16 in \cite{FRSS-LS2019}.  In this setting, the FB coefficient is identified as the (abstract) mean of the measure $\overline{\chi}\mu$, which also exists if $G$ does not admit a van Hove sequence; compare~\cite[Lemma~4.10.6]{FRSS-MoSt}. Reference \cite{FRSS-LS2019} also establishes a link to the dynamical systems approach to FB coefficients from \cite{FRSS-Lenz2009}.

We conclude this section with an extension of Proposition~\ref{FRSS-prop:FBC} to certain mild distributions.
\begin{proposition}\label{FRSS-prop:FBC}
Assume that $G$ admits a van Hove sequence $\cA$. Consider $\sigma \in \SOp(G)$ such that $\widehat{\sigma} \in \iota(\cM^\infty(G))$.  Then, for every $f \in \So(G)$ and for every $\chi \in \widehat{G}$, the Fourier--Bohr coefficient
\[
a_{\chi}^{\cA}(f*\sigma) \,=\, \lim_{n\to\infty} \frac{1}{m_G(A_n)} \int_{t+A_n} \overline{\chi(x)} \sigma(T_x f^\dag) \dd x
\]
exists uniformly in $t\in G$ and satisfies
\[
 a_{\chi}^\cA(f*\sigma)\,=\, \widehat{f}(\chi) \cdot \nu(\{ \chi \})  \ts ,
\]
where $\nu \in \cM^\infty(\widehat{G})$ is such that $\iota(\mu)=\widehat{\sigma}$. In particular, this result holds for all positive-definite mild distributions.
\end{proposition}
\begin{beweis}
Since $\widehat{f} \in \So(\widehat{G}) \subseteq \Wsp(\widehat{G})$ and $\nu \in \cM^\infty(\widehat{G})$, we have $\widehat{f} \in \Lisp(|\nu|)$. Therefore, $\widehat{f} \nu$ is a finite measure on $\widehat G$. Moreover, its inverse Fourier transform as a finite measure on $\widehat{G}$ is
\begin{displaymath}
\begin{split}
\reallywidecheck{\widehat{f} \nu}(x) &\,=\, \int_{\widehat{G}} \chi(x) \widehat{f}(\chi) \dd \nu(\chi ) \,=\, \int_{\widehat{G}} \widecheck{ T_x f^\dag}(\chi) \dd \nu(\chi) \\[1mm]
&\,=\, \widehat{\sigma}(\reallywidecheck{T_x f^\dag})=\sigma(T_x f^\dag) \,.
\end{split}
\end{displaymath}
The claim now follows from \cite[Lemma~4.8.7]{FRSS-MoSt} or \cite{FRSS-Ebe} applied to $\reallywidecheck{\widehat{f} \nu}$.
\qed \end{beweis}

\section{Diffraction theory revisited}\label{FRSS-sec:dtr}

Diffraction theory in Euclidean space has been exposed in \cite[Ch.~9]{FRSS-TAO1}. Whereas our presentation resembles the latter review, it will allow for more general LCA groups. This is motivated by prominent point sets in Euclidean space that are projected from lattices in certain non-Euclidean embedding groups; compare \cite{FRSS-BaakeMoody2004}. As reviewed in Section~\ref{FRSS-sec:ppdpsr}, the diffraction of such point sets can be deduced from harmonic analysis in the embedding non-Euclidean group. It is thus natural to also formulate diffraction theory for LCA groups beyond Euclidean space.

We will first review how to associate to a translation-bounded measure its diffraction measure.
For ease of presentation, we will restrict to second countable LCA groups. As second countability implies $\sigma$-compactness, we can average translation-bounded measures along van Hove sequences in this case. An extension to arbitrary LCA groups using van Hove nets \cite{FRSS-PRS22} is straightforward. Whereas this setting is well studied \cite{FRSS-Hof1, FRSS-BaakeLenz2004, FRSS-RS17, FRSS-LSS2020a}, our arguments treat translation-bounded measures systematically as elements of the dual of the Wiener algebra. This streamlines the previous approaches and also indicates extensions to the mild distribution and the tempered distribution setting \cite{FRSS-ST16}.

\subsection{Wiener diagram for a finite sample}\label{FRSS-sec:Wienerfs}

The physicist's recipe to describe kinematic diffraction of a piece of matter starts with a finite measure $\omega\in\cM(G)$, which might be a Dirac comb of the finitely many positions of the atoms or an absolutely continuous measure describing the mass distribution. Its Fourier--Stieltjes transform $\reallywidehat{\omega}$ is a continuous function, whose squared modulus is the intensity of the diffraction at the given position in reciprocal space $\widehat G$. For sufficiently large finite patches of a \emph{lattice}\index{lattice} $L$, that is, of a Delone subgroup of $G$, its diffraction intensity function peaks at the \emph{dual lattice}\index{lattice!dual} $L^\circ$, the annihilator of $L$ in $\widehat G$. The above approach is summarised in the lower left part of a so-called \emph{Wiener diagram}.
\begin{displaymath}
\begin{CD}
\omega @> * >> \omega *\widetilde{\omega}\\
@V{\widehat {} \, }VV @VV{\, \widehat {} }V\\
\reallywidehat{\omega} @>{|\cdot|^2}>>  \reallywidehat{\omega *\widetilde{\omega}}
\end{CD}
\end{displaymath}
The Wiener diagram displays that there is an alternative way to compute the diffraction via the upper right part of the Wiener diagram, by the Fourier transform of the autocorrelation of $\omega$, that is, of the convolution of $\omega$ with its twisted version $\widetilde{\omega}$, where $\widetilde{\omega}(f)=\overline{\omega(\widetilde{f})}$ and $\widetilde{f}(x)=\overline{f(-x)}$. When working with a mass distribution, the associated autocorrelation function is also called the Patterson function in the physics literature.

Mathematical diffraction analysis often uses an infinite idealisation of the given finite sample in order to define the spectral type of the diffraction. The diffraction of the infinite idealisation then approximates the diffraction of the finite sample, by continuity of the Fourier transform.  However, when dealing with unbounded measures in the above Wiener diagram, the operations of convolving, transforming and squaring may no longer be well defined. Thus, the above approach needs an appropriate modification. In fact, for bounded functions on the line, such a modification has already been described by Wiener \cite{FRSS-W30}, and our approach resembles his constructions.

\subsection{Autocorrelation measure}\label{FRSS-sec:dmsd}

The infinite idealisation is usually described by a translation-bounded measure $\omega\in \Wspp(G)$; see~\cite{FRSS-Hof1, FRSS-BaakeLenz2004, FRSS-CRS}. In particular, any uniformly discrete point set gives rise to a translation-bounded measure via its associated Dirac comb. In Euclidean space, one is typically interested in infinite point sets where the number of points in a ball grows like the volume of the ball. The diffraction measure associated to $\omega$ is then defined through properly normalised finite sample approximations.

Take a van Hove sequence  $\cA=(A_n)_{n\in\mathbb N}$ and consider the finite sample measures $\omega_n=\omega|_{A_n}$.
These give rise to positive-definite finite autocorrelation measures
\begin{displaymath}
\gamma_n \,=\, \frac{\omega_n*\widetilde{\omega_n}}{m_G(A_n)} \ts .
\end{displaymath}
Normalising by $m_G(A_n)$ ensures that the sequence $(\gamma_n)_{n\in \NN}$ is bounded in $(\Wspp(G), \op{\cdot}_{\Wspp})$. As second countability of $G$ implies metrisability of $G$, the Banach--Alaoglu theorem ensures that a subsequence of $(\gamma_n)_{n\in \NN}$ converges in the weak*-topology. We will arrive at this result in Proposition~\ref{FRSS-propAC} and first give an estimate on the growth of a van Hove sequence in the $\Wsp$-norm.

\begin{lemma}\label{FRSS-lem1} Assume that $\cA=(A_n)_{n\in \NN}$ is a van Hove sequence in $G$, and consider the space $(\Wsp(G), \|\cdot\|^{}_{\Wsp})$. Let $\varphi\in \Ccsp(G)$ be such that $0\le \varphi\le 1$ and let $K=\supp(\varphi)$. Then, there exists a finite constant $C$ such that
\[
\frac{\|1_{A_n-K}* \varphi\|^{}_{\Wsp}}{m_G(A_n)} \,\leq\,  C
\]
for every $n\in \NN$.
\end{lemma}

\begin{beweis}
Abbreviate $K=\supp(\varphi)$. Take a BUPU $\varPhi=(\phi_i)_{i\in I}$ of norm $M$, centered at a Delone set $\varLambda\subseteq G$, and abbreviate
\[
I_n \, = \, \{i\in I: \supp(\phi_i) \cap (A_n-K+K) \neq \varnothing\} 
\]
for $n\in\NN$. We then have
\begin{displaymath}
\begin{split}
\|1_{A_n}*\varphi\|^{}_{\Wsp,\varPhi}&
\,=\,\sum_{i\in I_n} \|(1_{A_n}*\varphi)\cdot\phi_i\|^{}_{\infty}
\,\leq\, \sum_{i\in I_n} \|\phi_i\|^{}_{\infty}\\[1mm]
&\,\le\,  M\cdot  \card(\varLambda \cap (A_n-K+K)) \ts .
 \end{split}
\end{displaymath}
Since $\cA$ is van Hove and $\varLambda$ is uniformly discrete, we have
\[
\limsup_{n\to\infty} \frac{\card(\varLambda \cap (A_n-K+K))}{m_G(A_n)} \,<\, \infty \ts .
\]
The claim is now obvious.
\qed \end{beweis}

Boundedness of the sequence $(\gamma_n)_{n\in \NN}$ is expressed in the following result.
\begin{lemma} Let $G$ be a second countable LCA group. Let $\cA=(A_n)_{n\in \NN}$ be any van Hove sequence in $G$ and consider any $\omega \in \Wspp(G)$. Then, there exists  some  weak*-compact set  $\boldsymbol K \subseteq \Wspp(G)$ such that   $\gamma_n\in \boldsymbol K$ for all $n\in \NN$, and such that the weak*-topology  is metrisable on $\boldsymbol K$.
\end{lemma}

\begin{beweis}
Our arguments use the auxiliary measures $\gamma_n'$ given by
\begin{displaymath}
\gamma_n' \,=\, \frac{\omega * \widetilde{\omega_n} }{m_G(A_n)} \ts .
\end{displaymath}
Note that $\gamma_n'\in \Wspp(G)$ is the convolution of a compactly supported measure and a translation-bounded measure, and is thus translation bounded.
Fix a BUPU $\varPhi$ in $\Wsp(G)$ on some Delone set $\varLambda\subseteq G$ with norm $\|\cdot \|^{}_{\Wsp,\varPhi}$ and associated translation invariant norm $\|\cdot\|^{}_{\Wsp}$.
Let $\varphi\in\Ccsp(G)$ be such that $0\le \varphi\le 1$ and $m_G(\varphi)=1$. Let $K=\supp(\varphi)$ and define
\[
f_n\,=\, 1_{A_n-K}*\varphi \, \in \, \Wsp(G) \ts.
\] 
Then, $f_n\ge 1_{A_n}$ by construction.
For arbitrary $f \in \Wsp(G)$, we can now estimate
\begin{displaymath}
\begin{split}
| &\omega * \widetilde{\omega_n}(f)|
\,=\,\left|\omega(  \omega_n*f)\right|
\,\leq\,
 \op{\omega}_{\Wspp, \varPhi}  \cdot \|\omega_n*f\|^{}_{\Wsp,\varPhi} \\[1mm]
 &\,\leq\,
 \op{\omega}_{\Wspp, \varPhi}  \cdot  \|f\|^{}_{\Wsp, \varPhi} \cdot |\omega_n|(G)
 \,\leq\,  C \cdot \op{|\omega|}_{\Wspp, \varPhi}  \cdot  \|f\|^{}_{\Wsp,\varPhi} \cdot |\omega|(f_n) \\[1mm]
 & \,\leq\, C \cdot \op{|\omega|}_{\Wspp, \varPhi}^2 \cdot  \|f\|^{}_{\Wsp,\varPhi} \cdot \|f_n\|^{}_{\Wsp,\varPhi}
 \,\leq\, C' \cdot \op{|\omega|}_{\Wspp,\varPhi}^2 \cdot  \|f\|^{}_{\Wsp,\varPhi} \cdot m_G(A_n) \ts ,
\end{split}
\end{displaymath}
where we used Lemma~\ref{FRSS-lem:fw} for the second inequality, Lemma~\ref{FRSS-lem:hnW} for the third, and Lemma~\ref{FRSS-lem1} for the last. As $f\in \Wsp(G)$ was arbitrary, we then have
\[
\op{\frac{1}{m_G(A_n)} \omega * \widetilde{\omega_n}}_{\Wspp,\varPhi} \,\leq\,  C'  \cdot \op{|\omega|}_{\Wspp,\varPhi}^2 \ts .
\]
Thus, $\gamma_n'\in \boldsymbol K$ for the set $ \boldsymbol K\subseteq \Wspp(G)$ defined by
\[
\boldsymbol K\,=\, \left\{ \sigma \in \Wspp(G): \op{\sigma}_{\Wspp, \varPhi} \,\leq\, C'  \cdot \op{|\omega|}_{\Wspp,\varPhi}^2 \right\} \ts .
\]
Note that $\boldsymbol K$ is weak*-compact by the Banach--Alaoglu theorem. Moreoverc, the weak*-topology is metrisable on $\boldsymbol K$, as $\Wsp(G)$ is separable by second countability of $G$.
In fact, also $\gamma_n\in \boldsymbol K$, which follows from the estimate on $\gamma_n'$ through
\begin{displaymath}
\begin{split}
|\omega_n * \widetilde{\omega_n}(f)| &\,=\,
\left|\omega_n(  \omega_n*f)\right|
\,\le\, |\omega_n| (|\omega_n*f|) \,\leq\, |\omega| (|\omega_n*f|) \ts .
\end{split}
\end{displaymath}
This finishes the proof of the lemma.
\qed \end{beweis}

We can summarise our arguments as follows.
\begin{proposition}\label{FRSS-propAC} Let $G$ be a second countable group.
 Let $\cA=(A_n)_{n\in \NN}$ be any van Hove sequence in $G$. Consider arbitrary $\omega \in \Wspp(G)$ and define $\omega_n=\omega|_{A_n}$ for $n\in \NN$. Then, there exists some $\gamma \in \Wspp(G)$ and a subsequence $(n_k)$  such that
\[
\gamma\, =\, \lim_{k\to\infty} \frac{\omega_{n_k} * \widetilde{\omega_{n_k}} }{m_G(A_{n_k})}
\]
in the weak*-topology of $\Wspp(G)$.
\qed
\end{proposition}
In the following, we thus assume that $\omega\in \Wspp(G)$ admits an autocorrelation $\gamma_\omega\in \Wspp(G)$ along some van Hove sequence $\cA$, that is, we have the weak*-limit
\begin{displaymath}
\gamma_\omega\,=\, \lim_{n\to\infty} \frac{\omega_n*\widetilde{\omega_n}}{m_G(A_n)}  \ts .
\end{displaymath}
We write $\gamma_\omega= \omega \circledast_{\cA} \widetilde{\omega}$ and call $\gamma_\omega$ the \emph{Eberlein convolution of $\omega$ along $\cA$}; compare \cite{FRSS-LSS2020a}.

\subsection{Diffraction measure}

By construction, any autocorrelation $\gamma_\omega$ is a positive-definite, translation-bounded measure.
Thus, its Fourier transform exists as a mild distribution and is positive. By Proposition~\ref{FRSS-prop:Feipos}, we have $\reallywidehat{\gamma_\omega}\in \Wspp(\widehat G)$, so the Fourier transform is a translation-bounded measure on $\widehat G$. The measure $\reallywidehat{\gamma_\omega}$  is called the \emph{diffraction measure of $\omega$ with respect to $\cA$}.

Note that the Fourier transform is vaguely continuous on the cone of positive-definite, translation-bounded measures by \cite[Thm.~4.16]{FRSS-BF75}. Moreover, recalling $\Wspp(G)\subseteq \SOp(G)$, note that the Fourier transform on $(\SOp(G), \op{\cdot}_{\So})$ is weak*-continuous by Theorem~\ref{FRSS-thm:Fdist}. We thus have $\widehat{\gamma_n}(h) \to \widehat{\gamma_\omega}(h)$ for every $h\in \Ccsp(\widehat G)$ and every $h\in \So(\widehat G)$.
In this sense, the diffraction measure approximates the diffraction of a large finite sample $\omega_n$.

The diffraction measure can be decomposed into a discrete part and a continuous part; we write
\begin{displaymath}
\reallywidehat{\gamma_\omega}\, =\, (\reallywidehat{\gamma_\omega})_{\mathsf d} \,+\, (\reallywidehat{\gamma_\omega})_{\mathsf c} \ts .
\end{displaymath}
Then, $\omega\in \Wspp(G)$ is called \emph{pure-point diffractive along $\cA$} if
$\reallywidehat{\gamma_\omega}=(\reallywidehat{\gamma_\omega})_{\mathsf d}$.

\subsection{Weak topologies on translation-bounded measures}\label{FRSS-sec:wttbm}

Translation-bounded measures, that is, elements of $\cM^\infty(G)=\Wspp(G)$, may be evaluated on various test function spaces $\Tsp(G)\subseteq \Wsp(G)$ such as $\Ccsp(G)$, $\KLsp(G)$, $\So(G)$ or $\Wsp(G)$. One might also consider the Schwartz--Bruhat space $\Scsp(G)$ of test functions for tempered distributions or, in Euclidean space, the space $\Dcsp(G)=\Ccsp^{\!\infty}(G)$ of test functions for distributions.
Let us consider the associated topologies $\sigma(\Wspp,\Tsp)$ of pointwise convergence on $\Wspp(G)$ with respect to $\Tsp(G)$. For example, the topology $\sigma(\Wspp, \Wsp)$ is the weak*-topology on $\Wspp(G)$, whereas the topology $\sigma(\Wspp, \Ccsp)$ is the trace topology of the weak*-topology on $\cM(G)=\Ccsp^{\!\prime}\mkern-2mu(G)$, also called the \emph{vague topology}. We have the following result.

\begin{proposition}\label{FRSS-prop:top}  Assume that $\Tsp(G)\in \{\KLsp(G), \Ccsp(G), \So(G), \Wsp(G)\}$ or $\Tsp(G)\in \{\Dcsp(G), \Scsp(G)\}$, and consider the topology $\sigma(\Wspp,\Tsp)$ of pointwise convergence on $\Wspp(G)$ with respect to $\Tsp(G)$. Consider any closed $r$-ball
\[
B_r \,=\,  \{ \mu \in \Wspp(G): \op{\mu}_{\Wspp(G)} \leq r \}
\]
in $\Wspp(G)$. Then, all the above topologies $\sigma(\Wspp,\Tsp)$ have the same trace on $B_r\subseteq \Wspp(G)$.
\end{proposition}

This result shows that the traditional definition of an autocorrelation agrees with the above approach. Indeed, the topology on translation-bounded measures from \cite{FRSS-Hof1, FRSS-BaakeLenz2004} coincides with the above topology by Proposition~\ref{FRSS-prop:equiv}, and weak*-convergence on balls is equivalent to vague convergence on balls by Proposition~\ref{FRSS-prop:top}. Adopting the traditional viewpoint, we have an alternative argument for \cite[Thm.~2]{FRSS-BaakeLenz2004}. We also observe that any vague autocorrelation limit is even weak*-convergent. Moreover, it suffices to check weak*-convergence or vague convergence on test functions from $\KLsp(G)$ or from $\Scsp(G)$ or, in Euclidean space, from $\Dcsp(G)$.

\smallskip

The proof of Proposition~\ref{FRSS-prop:top} relies on the Banach--Alaoglu theorem, in conjunction with a well-known fact from topology; compare \cite[Lemma~2.23]{FRSS-LS2019}.

\begin{lemma}[{\cite[Ch.~I, 3A]{FRSS-Loo53}}]\label{FRSS-lem:top} Assume that $(X, \tau_1)$ is Hausdorff and that $(X, \tau_2)$ is compact. Then, $\tau_1 \subseteq \tau_2$ implies $\tau_1=\tau_2$. \qed
\end{lemma}

\begin{beweis}[Proof of Proposition~\ref{FRSS-prop:top}]
Let $B_r$ be any closed $r$-ball in $(\Wspp, \op{\cdot}_{\Wsp})$.  Note that $\Tsp(G)\subseteq \Wsp(G)$ for $\Tsp\in \{\KLsp(G), \Ccsp(G), \So(G), \Dcsp(G), \Scsp(G)\}$, which implies
\[
\sigma(\Wspp,\Tsp)|_{B_r} \, \subseteq \,  \sigma(\Wspp,\Wsp)|_{B_r}
\]
for all such $\Tsp$.
Also, $\big(B_r, \sigma(\Wspp,\Tsp)|_{B_r}\big)$ is Hausdorff for all such $\Tsp$.  Using standard arguments for measures, this follows from the property that for every compact $K\subseteq G$ there exists some $f\in \Tsp(G)$ such that $1_K\le f$, for all such $\Tsp$.

\noindent Since the topological space $\big(B_r,\sigma(\Wspp,\Wsp)|_{B_r}\big)$ is  compact by the Banach--Alaoglu theorem, we obtain
\[
\sigma(\Wspp,\Tsp)|_{B_r} \, =\, \sigma(\Wspp,\Wsp)|_{B_r}
\]
for all such $\Tsp$ by Lemma~\ref{FRSS-lem:top}.
\qed \end{beweis}

\subsection{Wiener diagram for Besicovitch almost-periodic measures}\label{FRSS-sec:mWd}

Here, we ask whether there exists a modified Wiener diagram for the infinite idealisation $\omega\in \cM^\infty(G)$.
As in Section~\ref{FRSS-sec:dmsd} above, consider normalised Fourier transforms of $\omega_n$ along a van Hove sequence $\cA=(A_n)_{n\in\NN}$. For any $\chi\in\widehat G$, we might then ask for existence of the limit
\begin{displaymath}
\lim_{n\to\infty} \frac{\widehat{\omega_n}(\chi) }{m_G(A_n)} \,=\, a^\cA_\chi(\omega) \ts ,
\end{displaymath}
which is to say that we might ask for the existence of the FB coefficients of $\omega$ along $\cA$. For the following discussion, we restrict to measures that are pure-point diffractive along $\cA$. We thus restrict to measures that are mean almost periodic along $\cA$; see  \cite{FRSS-ARMA, FRSS-MoSt, FRSS-LSS2020a, FRSS-LRS2024} for background of various notions of almost periodicity with respect to pure-point diffractivity.

For concreteness, consider the class $\cM_{dT}(G)$ of twice AG-transformable measures.  Such measures are indeed pure-point diffractive, as AG transformability implies weak almost periodicity \cite[Thm.~4.10.4]{FRSS-MoSt}, which in turn implies pure-point diffractivity \cite[Thm.~4.4]{FRSS-LS2019}. For $\omega\in \cM_{dT}(G)$, both its autocorrelation measure and its FB coefficients are independent of the choice of the van Hove sequence $\cA$. The relation $a^\cA_\chi(\omega)=\widehat{\omega}(\{\chi\})$ then suggests a Wiener diagram where the left-down arrow corresponds to taking the discrete part of the AG transform $\widehat\omega$, and where $|\cdot|^2$ means taking the absolute square of the FB coefficients. For $\omega\in\cM_{dT}(G)$, we can thus write
\begin{displaymath}
\begin{CD}
\omega @>\circledast>>\gamma_\omega\\
@V{(\widehat{\hspace{1ex}})_{\mathsf d} \,}VV @VV{\, \widehat{}}V\\
(\widehat{\omega})_{\mathsf d} @>{|\cdot|^2}>> \reallywidehat{\gamma_{\omega}}
\end{CD}
\end{displaymath}
\smallskip

\noindent If $\omega\in\cM_{dT}(G)$ is even strongly almost periodic, then $(\widehat{\omega})_{\mathsf d}=\widehat{\omega}$, and both down arrows in  the above Wiener diagram correspond to the AG transform.

\begin{example}
A simple explicit example is $\omega=\delta^{}_{\mkern-2mu L}\in \cM_{dT}(G)$ the  Dirac comb of a lattice $L$. Its autocorrelation $\gamma_\omega$ is given by
\[
\gamma_\omega \, = \, \dens(L)\cdot \delta^{}_{\mkern-2mu L}\in \cM^\infty(G) \ts.
\]
Taking the AG Fourier transform, we get
\[
\reallywidehat{\gamma_\omega} \, =\, \dens(L)^2\cdot \delta^{}_{\mkern-2mu L^\circ}\in \cM^\infty(\widehat G) \ts ,
\]
by the Poisson summation formula for the lattice $L$\index{Poisson~summation~formula!for~a~lattice}. In fact, we also have  $\widehat{\omega}=\dens(L)\cdot \delta^{}_{\mkern-2mu L^\circ}\in \cM^\infty(\widehat G)$, which shows that $\omega$ is a strongly almost periodic measure. \exend
\end{example}

A simple example of a weakly almost-periodic measure $\omega$ that is not strongly almost periodic is
\[
\omega \, =  \, \delta^{}_{\mkern-2mu L} \,+\, \delta^{}_{0}\ts .
\]
More generally, one might take $\omega=\mu^{}_{\mathsf s}+\mu^{}_{\mathsf 0}$ with $\mu^{}_{\mathsf s}$ a strongly almost-periodic measure and $\mu^{}_{\mathsf 0}$ a finite measure. We will give further examples of strongly almost-periodic twice transformable measures below; see the discussion of weighted model sets in Section~\ref{FRSS-sec:dwms}.

More generally, consider the class $\Bap_\cA(G)\subseteq \cM^\infty(G)$ of translation-bounded measures that are Besicovitch almost periodic along $\cA$ in the sense of \cite[Def.~3.30]{FRSS-LSS2020a}. Examples are Dirac combs of model sets of maximal density along $\cA$; compare \cite[Sec.~3.5]{FRSS-LSS2020a}. In this case, the autocorrelation and the FB coefficients of $\omega$ along $\cA$ indeed exist. In fact, for $\omega\in \Bap_\cA(G)$, the Wiener diagram
\begin{displaymath}
\begin{CD}
\omega @>\circledast_{\cA}>>\gamma_\omega\\
@V{{\rm FB}^{}_{\!\cA} \,}VV @VV{\, \widehat{}}V\\
(a_\chi^\cA(\omega))_{\chi\in \widehat G} @>{|\cdot|^2}>> \reallywidehat{\gamma_{\omega}}
\end{CD}
\vspace{2mm}
\end{displaymath}
commutes, where ${\rm FB}^{}_{\!\cA}$ means taking Fourier--Bohr coefficients along $\cA$, and where $|\cdot|^2$ means taking the absolute square of the FB coefficients \cite[Thm.~3.36]{FRSS-LSS2020a}. For the subclass of so-called Weyl almost-periodic measures, the above limits are independent of $\cA$ by \cite[Thm.~4.16]{FRSS-LSS2020a}. Examples are given by weakly almost-periodic measures \cite[Cor.~4.26]{FRSS-LSS2020a} and hence in particular by the doubly transformable measures considered above. Note that, for Besicovitch almost-periodic measures, $\widehat\omega\in \SOp(\widehat G)$ exists as mild distribution, but might not be a measure.
Note also that the FB series $(a_\chi^\cA(\omega))_{\chi\in \widehat G}$ might not define a measure or a mild distribution. Even if it does, the corresponding mild distribution need not coincide with the mild distribution $\widehat \omega$.

Note also that, for a translation-bounded measure $\omega\in\cM^\infty(G)$ that is pure-point diffractive along $\cA$, commutativity of the above Wiener diagram characterises Besicovitch almost periodicity along $\cA$ by \cite[Thm.~3.36]{FRSS-LSS2020a}. The above Wiener diagram fails to be commutative for the pure-point diffractive Dirac comb $\omega=\delta_\NN\in \cM^\infty(\RR)$ on the van Hove sequence where $A_n=[-n,n]$ for $n\in \NN$. For another simple class of pure-point diffractive counterexamples, see \cite[Prop.~A.2]{FRSS-LSS2020a}.

\section{Pure-point diffraction and Poisson summation}\label{FRSS-sec:ppdpsr}

In the following, we specialise to weighted versions of Meyer's model sets \cite{FRSS-Mey72}. We will obtain extended versions of results from \cite{FRSS-CRS}, using straightforward proofs based on mild distributions. Throughout this section, we consider LCA groups $G,H$ with a lattice $\cL\subseteq G\times H$\index{lattice}, meaning that $\cL$ is a discrete co-compact subgroup of $G\times H$.

\subsection{Weighted model sets that are AG transformable}

Weighted model sets in $G$ are supported on projected lattice points.  For a weight function $h:H\to \CC$, consider the formal sum
\begin{displaymath}\index{model~set!weighted}
\omega_h\,=\sum_{(x,y)\in \cL} h(y)\cdot \delta_x \ts .
\end{displaymath}
We are interested in weight functions that result in translation-bounded measures $\omega_h\in\cM^\infty(G)$. This is the case for the continuous weight functions from the Wiener algebra.

\begin{lemma}\label{FRSS-lem:ohW} Let $G,H$ be two LCA groups and let $\cL$ be a lattice in $G \times H$. Then, there exists a finite constant $C=C(\cL)$ such that
\[
\op{\omega_{h}}_{\Wspp(G)} \,\le\,  C\cdot \| h \|^{}_{\Wsp(H)}
\]
for all $h\in \Wsp(H)$. In particular, for all $h \in \Wsp(H)$, we have $\omega_h\in \cM^\infty(G)$.
\end{lemma}

\begin{beweis}
Note that $\delta_\cL=\sum_{\ell\in \cL}\delta_\ell$ satisfies $\delta_\cL\in \cM^\infty(G\times H)$ as $\cL$ is a Delone set in $G\times H$.

\noindent Fix BUPUs $\varPhi$ on $G$ and $\varPsi$ on $H$, and
%
consider arbitrary $g\in \Wsp(G)$ and $h \in \Wsp(H)$. We then have $g\otimes h \in \Wsp(G\times H)$ by Lemma~\ref{FRSS-lem:WW}.
Using the product BUPU $\varPhi\otimes \varPsi$ on $G\times H$, compare Lemma~\ref{FRSS-lem:prodbupu}, we get
\begin{displaymath}
\begin{split}
|\omega_h(g)|&\,=\,|\delta_\cL(g\otimes h)|\,\le\, \op{\delta_\cL}_{\Wspp(G\times H), \varPhi\otimes\varPsi} \cdot \|g\otimes h\|^{}_{\Wsp(G\times H), \varPhi\otimes\varPsi}\\[1mm]
&\,\le\, \op{\delta_\cL}_{\Wspp(G\times H), \varPhi\otimes\varPsi} \cdot \|g\|^{}_{\Wsp(G), \varPhi} \cdot \|h\|^{}_{\Wsp(H), \varPsi} \ts ;
\end{split}
\end{displaymath}
compare Lemma~\ref{FRSS-lem:WW}.
As $g\in \Wsp(G)$ was arbitrary, this shows that we have $\omega_h\in \Wspp(G)=\cM^\infty(G)$ and
\[
\op{\omega_{h}}_{\Wspp(G), \varPhi} \,\le\, C \cdot \| h \|^{}_{\Wsp(H), \varPsi} \ts ,
\]
with $C=\op{\delta_\cL}_{\Wspp(G\times H), \varPhi\otimes\varPsi}\cdot \|h\|^{}_{\Wsp(H), \varPsi}$.
\qed \end{beweis}

We are particularly interested in translation-bounded measures $\omega_h$ that allow for Fourier analysis based on the Poisson summation formula (PSF) for the underlying lattice $\cL$. This idea has been pursued for AG-transformable measures in \cite{FRSS-CRS}, which also contains a review of the diffraction theory literature, and for tempered distributions in Euclidean space in \cite{FRSS-TAO1}.

Here we note that, by Theorem~\ref{FRSS-dT-Wo}, the PSF for a discrete co-compact subgroup of $G$ does hold not only on $\So(G)\subseteq \Wsp(G)$, but also on the larger Segal algebra $\WOsp(G)\subseteq \Wsp(G)$.
This yields the following statement, which refines previous results from \cite{FRSS-Mey12, FRSS-TAO1, FRSS-CRS, FRSS-M19, FRSS-M22}.

\begin{theorem}[PSF for weighted model sets]\label{FRSS-equiv-PSF-diff-dens}
Consider LCA groups $G,H$, and let $\cL$ be a lattice in $G\times H$. We then have
\begin{equation}\label{FRSS-eq1}\index{Poisson~summation~formula!for~a~weighted~model~set}
\omega_h(g)\,=\,  \mathrm{dens}(\cL)\cdot \omega_{\widecheck{h}}(\widecheck{g}) 
\end{equation}
for all $g \in \WOsp(G)$ and all $h \in \WOsp(H)$, where $\omega_{\widecheck{h}}\in \cM^\infty(\widehat G)$ is given by $\omega_{\widecheck{h}}=\sum_{(\chi,\eta)\in \cL^0} \widecheck{h}(\eta)\cdot \delta_\chi$. We thus have $\omega_h\in \mathcal M_{dT}(G)$ and
\[
\widehat{\omega_h} \, = \, \mathrm{dens}(\cL)\cdot \omega_{\widecheck{h}}
\]
for every $h\in \WOsp(H)$.
In particular, both $\omega_h\in \cM^\infty(G)$ and $\widehat{\omega_h}\in \cM^\infty(\widehat G)$ are strongly almost-periodic measures for every $h\in \WOsp(H)$.
\end{theorem}

\begin{beweis} We infer from Lemma~\ref{FRSS-lem:ohW} that  the formal sums $\omega_h$ and $\omega_{\widecheck{h}}$ are translation-bounded measures for all $h \in \WOsp(H)$.
\noindent Let now $g \in \WOsp(G)$ be arbitrary. We then have
\begin{displaymath}
g \otimes h \,\in\, \WOsp(G)\otimes \WOsp(H)\,\subseteq\, \WOsp(G\times H)
\end{displaymath}
by Lemma~\ref{FRSS-lem:WW}.
%
%
Since $\widecheck{g \otimes h}=\widecheck{g} \otimes \widecheck{h}$, Theorem~\ref{FRSS-dT-Wo} gives
\[
\omega_h(g)\,=\, \delta^{}_{\cL}(g \otimes h) \,=\, \dens(\cL) \cdot \delta^{}_{\cL^0}(\widecheck{g} \otimes \widecheck{h}) \,=\,  \mathrm{dens}(\cL)\cdot \omega_{\widecheck{h}}(\widecheck{g}) \ts ,
\]
with all sums being absolutely convergent by Theorem~\ref{FRSS-thm:tbW}. This proves Eq.~\eqref{FRSS-eq1}.

\noindent Fix now $h\in \WOsp(H)$. Noting $\omega_h\in \cM^\infty(G)$, $\omega_{\widecheck{h}}\in \cM^\infty(\widehat G)$ and validity of Eq.~\eqref{FRSS-eq1} for all $g \in \So(G) \subseteq \WOsp(G)$, the  transformability claim follows from Theorem~\ref{FRSS-t1b}~(iii).

\noindent Finally, strong almost periodicity follows from \cite[Cor.~4.10.13]{FRSS-MoSt}.
\qed \end{beweis}

Using the PSF, we obtain the following extended version of the density formula, by a proof exactly as in \cite[Prop.~4.13]{FRSS-CRS}. The density formula for Riemann integrable weight functions then follows from a standard approximation argument; compare~\cite[Thm.~4.14]{FRSS-CRS}.

\begin{proposition}[Density formula for weighted model sets]\label{FRSS-prop:df}  Assume that $G$ is $\sigma$-compact and that $\cL$ projects densely to $H$. Let $(A_n)_{n\in\mathbb N}$ be any van Hove sequence in $G$.  We then have for  $h \in \WOsp(H)$ that
\begin{displaymath}
\lim_{n\to\infty} \frac{ \omega_h(t+A_n)}{m_G(A_n)}\,=\,\mathrm{dens}(\mathcal \cL)\cdot \int_H h(y) \, {\rm d}m_H(y) \ts ,
\end{displaymath}
where the convergence is uniform in $t\in G$. \qed
\end{proposition}

\begin{beweis}
With $c=\mathrm{dens}(\mathcal \cL)$, we can argue
\begin{displaymath}
\lim_{n\to\infty} \frac{ \omega_h(t+A_n)}{m_G(A_n)}
\,=\,\widehat{\omega_h}(\{e\})\,=\,c\cdot \omega_{\widecheck{h}}(\{e\})\,=\,c\cdot \widecheck h(e)
\,=\,c\cdot \int_H h(y) \, {\rm d}m_H(y)
\end{displaymath}
uniformly in $t\in G$. Here, the first equation holds by Proposition~\ref{FRSS-prop:FBC}, and the second one is the PSF from Theorem~ \ref{FRSS-equiv-PSF-diff-dens}.
\qed \end{beweis}

\subsection{Diffraction of AG-transformable weighted model sets}\label{FRSS-sec:dwms}

Versions of the following result have been proved for Riemann integrable weight functions \cite[Prop.~5.1]{FRSS-CRS} and for continuous weight functions that are called admissible in \cite[Thm.~3.2]{FRSS-LR}. Here, we treat continuous weight functions from $\WOsp(H)$. Instead of adapting the direct proof of \cite[Prop.~5.1]{FRSS-CRS} that uses the density formula, we prefer a more general argument  from \cite{FRSS-LS2019} that is based on almost periodicity.

\begin{proposition}\label{FRSS-thm:ac}
Assume that $G$ is $\sigma$-compact and that $\pi^H(\cL)$ is dense in $H$. Assume further that $h\in \WOsp(H)$. Then, the weighted model set comb $\omega_h\in\mathcal M^\infty(G)$ has a unique autocorrelation measure $\gamma=\omega_h\circledast \widetilde{\omega_h}\in\mathcal{M}_{dT}(G)$ and a unique diffraction measure $\widehat{\gamma}\in\mathcal{M}_{dT}(G)$. They are given by
\[
\gamma \,= \, \mathrm{dens}(\mathcal L)\cdot\omega_{h*\widetilde h} \ts , \qquad 
\widehat{\gamma} \, =\, \dens(\cL)^2\cdot\omega_{|\widecheck{h}|^2} \ts .
\]
\end{proposition}

\begin{beweis}
Let $h\in \WOsp(H)$. As $\omega_h\in \cM^\infty(G)$ is strongly almost-periodic by Theorem \ref{FRSS-equiv-PSF-diff-dens}, it has a unique autocorrelation $\gamma=\omega_h\circledast \widetilde{\omega_h}\in \cM_{dT}(G)$ by  \cite[Thm.~7.6]{FRSS-LS2019}. 
Its Fourier transform $\widehat{\gamma}\in \cM_{dT}(\widehat G)$ is a discrete measure, which can be evaluated at 
$\chi\in \widehat G$ via
\[
\widehat{\gamma}(\{\chi \}) \,=\, |a_\chi(\omega_h)|^2 \ts ,
\]
where $a_\chi(\omega_h)$ is the FB coefficient of $\omega_h\in\cM_{dT}(G)$ at $\chi$; compare  \cite[Thm.~7.6]{FRSS-LS2019}, \cite[Thm.~4.10.14]{FRSS-MoSt} and Eq.~\eqref{FRSS-eq:FBC}. We thus have
\begin{displaymath}
\widehat{\gamma}(\{\chi \})\,=\,\left| \widehat{\omega_h}(\{\chi\})\right|^2
\,=\, \mathrm{dens}(\cL)^2\cdot \big| \omega_{\widecheck{h}}(\{\chi\})\big|^2
\end{displaymath}
for all $\chi\in \widehat G$.
In particular, we get $\widehat{\gamma}(\{\chi \})=0$ for all $\chi\notin \pi^{\widehat G}(\cL^\circ)$.

\noindent On the other hand, for $(\chi,\eta)\in \cL^\circ$ we get
\[
\widehat{\gamma}(\{\chi \}) \, =\, \dens(\cL)^2\cdot\big|\widecheck{h}(\eta)\big|^2 \ts .
\]
Here, we used that $\cL^\circ$ projects injectively to $\widehat G$, which follows from the assumption that $\cL$ projects densely to $H$ by Pontryagin duality. This shows
\[
\widehat{\gamma} \, =\, \dens(\cL)^2\cdot\omega_{|\widecheck{h}|^2} \ts .
\]
Noting $h*\widetilde{h} \in \WOsp(H)$, Theorem~\ref{FRSS-equiv-PSF-diff-dens} implies that $\dens(\cL) \cdot \omega_{h*\widetilde{h}}$ is AG transformable and has the same Fourier transform as $\gamma$. Consequently,
\[
\gamma \, =\, \mathrm{dens}(\mathcal L)\cdot\omega_{h*\widetilde h} \ts .
\]
This completes the proof of the proposition. 
\qed \end{beweis}

\subsection{Discussion of results}\label{FRSS-sec:disc}

For any weight function $h\in \WOsp(H)$, the weighted model set comb given by $\omega_h\in\cM_{dT}(G)$ is a strongly almost-periodic measure. In this case, the Wiener diagram from Section~\ref{FRSS-sec:mWd} has the following explicit form.
\begin{displaymath}
\begin{CD}
\omega_h @>\circledast>>\mathrm{dens}(\mathcal L) \cdot \omega_{h*\widetilde h}\\
@V{\widehat{} \,}VV @VV{\, \widehat{}}V\\
\mathrm{dens}(\mathcal L)\cdot \omega_{\widecheck h} @>{|\cdot|^2}>> \mathrm{dens}(\mathcal L)^2\cdot \omega_{|\widecheck h|^2}
\end{CD}
\end{displaymath}
All objects in the above diagram can be interpreted as translation-bounded measures; compare Theorem~\ref{FRSS-equiv-PSF-diff-dens}. Also note that, on the level of the weight functions, the above Wiener diagram reflects the Wiener diagram for finite measures from Section~\ref{FRSS-sec:Wienerfs}.

\begin{remark}
The above result extends \cite[Rem.~5.6]{FRSS-CRS}. In addition, it applies to weight functions from the Schwartz--Bruhat class and thus includes the Gelfand--Shilov functions, which were used in \cite{FRSS-B86} for the Fourier analysis of model sets in Euclidean cut-and-project schemes. \exend
\end{remark}

For appropriate weight functions, weighted model set combs $\omega_h$ provide additional examples of Besicovitch almost-periodic measures. Let us first consider examples that are strongly almost-periodic, translation-bounded measures (and hence Besicovitch almost periodic; compare Section~\ref{FRSS-sec:mWd}). Such examples arise from continuous weight functions and already appear implicitly in de Bruijn's work; see \cite{FRSS-B86, FRSS-B87} and the recent discussion in \cite{FRSS-LLRSS}. If $h\in \Ccsp(H)$, then $\omega_h\in\cM^\infty(G)$ is strongly almost periodic \cite{FRSS-LR}. In this case, we have  $\omega_h\in \cM_{dT}(G)$ if and only if
$h \in \KLsp(G)$; see \cite[Thm.~5.3]{FRSS-CRS}. In particular, we get examples that are not AG transformable.
More generally, admissble $h\in \COsp(H)\cap \Lisp(H)$ result in $\omega_h\in\cM^\infty(G)$  being strongly almost periodic; see \cite[Thm.~3.1]{FRSS-LR} and \cite[Prop.~6.7]{FRSS-LR}.

If the weight function fails to be continuous, the weighted model set combs will usually fail to be strongly almost periodic. Consider for example weak model sets of so-called maximal density \cite[Sec.~3.5]{FRSS-LSS2020a}. In this case, $h=1_W$ is the characteristic function of the model set window $W$, which is assumed to be compact. To be specific, assume that $W$ has non-empty interior and is Riemann measurable. Then, $\omega_h\in \cM^\infty(G)$ is almost automorphic \cite[Thm.~1]{FRSS-KR2015}. Moreover, typically $\widecheck{1_W}\notin \Lisp(\widehat H)$, such that AG transformability might be violated.

\section{Origins and perspectives}\label{FRSS-sec:app}

We have put diffraction theory for translation-bounded measures on an LCA group into a functional-analytic framework. This links mathematical diffraction to abstract harmonic analysis and suggests future research directions.

\smallskip

In 1930, Wiener developed Fourier analysis for time signals \cite{FRSS-W30}, a central object being the autocorrelation of the time signal. Subsequently, his methods were abstracted to Fourier analysis of essentially bounded functions and pseudo-measures, as described in the monograph \cite{FRSS-B75}. This can be adapted to translation-bounded measures on LCA groups, using mild distributions. For example, one may analyse the Wiener spectrum of a translation-bounded measure. Moreover, methods for reconstructing a signal from its spectrum \cite{FRSS-WW38, FRSS-B91} can be transferred to translation-bounded measures. In particular, techniques for spectrum estimation, as described in \cite{FRSS-B91}, may be relevant in crystallography. One may also extend Wiener's generalised harmonic analysis beyond translation bounded measures, in order to describe the diffraction of crystalline measures \cite{FRSS-MRS26}. Here, modulation spaces might provide appropriate test functions \cite[Ch.~11]{FRSS-G01}.

\smallskip

Another direction concerns almost periodic measures. An important subclass of pure point diffractive measures is Besicovitch almost periodic measures \cite{FRSS-L00, FRSS-LSS2020}.  In 1941, Wiener and Wintner characterised Besicovitch almost periodicity for functions on the line \cite{FRSS-WW41}. This can be extended to translation-bounded measures. Another important class of translation-bounded measures is translation-bounded point measures supported within a Meyer set. In that case, almost periodicity can be characterised within an underlying cut-and-project scheme \cite{FRSS-LRS2024}. More general discrete supporting sets are deformed model sets \cite{FRSS-LLRSS}, and curved model sets \cite{FRSS-Mey23}. As their Fourier transform as a tempered distribution typically fails to be a translation bounded measure,  the AG Fourier transform cannot be used for analysis. Here, mild distributions provide an alternative framework, that readily extends to general LCA groups.

\smallskip

A third direction concerns function reconstrution from samples and, for prescribed values on a given set, interpolation by functions from a given class. Here model sets provide aperiodically ordered point sets that allow for a substantial analysis beyond the lattice case. For example, it was noted by Meyer that simple model sets yield points sets of stable sampling and stable interpolation in the Paley-Wiener class arbitrarily close to the critical density; see \cite{FRSS-Mey12} for an overview. A crucial proof ingredient is a certain duality between sampling and interpolation, which is particular for model sets and rests on the Poisson summation formula. Here Wiener amalgams provide a flexible collection of test function spaces, which might also be useful beyond the Euclidean setting, compare \cite{FRSS-RiSchu20}. With respect to function reconstruction using time-frequency sampling \cite{FRSS-G01}, Gabor analysis for model sets seems a promising approach. Initial investigations are \cite{FRSS-M19, FRSS-M22} and \cite{FRSS-EV25}.

\bigskip

\noindent\emph{Acknowledgements}.
NS was supported by the Natural Sciences and Engineering Council of Canada via grants 2020-00038 and 2024-04853, and he is grateful for the support. We thank Jan Maz\'a\v c for a careful reading of the manuscript.

%% file: main.bbl
\begin{thebibliography}{AO2}\itemsep=2pt

\bibitem[A\nts O1]{FRSS-TAO1}
\newblock Baake M. and Grimm U. (2013).
\newblock \emph{Aperiodic Order. Vol.~1: A Mathematical Invitation}
\newblock (Cambridge University Press, Cambridge).

\bibitem[A\nts O2]{FRSS-TAO2}
\newblock Baake M. and Grimm U. (eds.) (2017).
\newblock \emph{Aperiodic Order. Vol.~2: Crystallography and Almost Periodicity}
\newblock (Cambridge University Press, Cambridge).

\bibitem{FRSS-ARMA2}
\newblock Argabright L.N. and Gil~de~Lamadrid J. (1971).
\newblock Fourier transforms of unbounded measures,
\newblock \emph{Bull.~Amer.~Math.~Soc.}~\textbf{77}, 355--359.

\bibitem{FRSS-ARMA1}
\newblock Argabright L.N. and Gil~de~Lamadrid J. (1974).
\newblock \emph{Fourier Analysis of Unbounded Measures on Locally Compact Abelian Groups},
\newblock Memoirs Amer.~Math.~Soc.~\textbf{145}.

\bibitem{FRSS-AG79}
\newblock Argabright L.N. and Gil~de~Lamadrid J. (1979).
\newblock Über fastperiodische Maße.
\newblock In \emph{Winterschule 1979, Internationale Arbeitstagung \"uber Topologische Gruppen und Gruppenalgebren}, pp.~27--32
(Univ.~Wien).

\bibitem{FRSS-BaakeLenz2004}
\newblock Baake M. and Lenz D. (2004).
\newblock Dynamical systems on translation-bounded measures: pure point dynamical and diffraction spectra,
\newblock \emph{Ergod.~Th.~\& Dynam.~Syst.}~\textbf{24}, 1867--1893. \texttt{arXiv:math/0302061}.

\bibitem{FRSS-BaakeMoody2004}
\newblock Baake M. and Moody R.V. (2004).
\newblock Weighted Dirac combs with pure point diffraction,
\newblock \emph{J.~reine angew. Math.~(Crelle)}~\textbf{573}, 61--94. \texttt{arXiv:math/0203030}.

\bibitem{FRSS-BS22}
\newblock Baake M. and Strungaru N. (2023).
\newblock A note on tempered measures,
\newblock \emph{Colloq.~Math.}~\textbf{172}, 15--30. \texttt{arXiv:2202.09175}.

\bibitem{FRSS-B75}
\newblock Benedetto J. (1975).
\newblock \emph{Spectral Synthesis}
\newblock (Teubner, Stuttgart).

\bibitem{FRSS-B91}
\newblock Benedetto J. (1991).
\newblock A multidimensional Wiener-Wintner theorem and spectrum estimation,
\newblock \emph{Trans.~Amer.~Math.~Soc.} \textbf{327}, 833--852.

\bibitem{FRSS-BF75}
\newblock Berg C. and Forst G. (1975).
\newblock \emph{Potential Theory on Locally Compact Abelian Groups}
\newblock (Springer, Berlin).

\bibitem{FRSS-B61}
\newblock Bruhat F. (1961).
\newblock Distributions sur un groupe localement compact et applications \`a l'\'etude des repr\'esentations des groupes $\wp$-adiques, 
\newblock \emph{Bull.~Soc.~Math.~France}~\textbf{89}, 43--75.

\bibitem{FRSS-C78}
\newblock Cowling M.G. (1978).
\newblock Some applications of Grothendieck’s theory of topological tensor products in harmonic analysis,
\newblock \emph{Math.~Ann.}~\textbf{232}, 273--285.

\bibitem{FRSS-B86}
\newblock de Bruijn N.G. (1986).
\newblock Quasicrystals and their Fourier transform,
\newblock \emph{Indag.~Math.}~\textbf{48}, 123--152.

\bibitem{FRSS-B87}
\newblock de Bruijn N.G. (1987).
\newblock Modulated quasicrystals,
\newblock \emph{Indag.~Math.}~\textbf{49}, 121--132.

\bibitem{FRSS-Bu81}
\newblock B\"urger R. (1981).
\newblock Functions of translation type and functorial properties of Segal algebras. II.,
\newblock \emph{Monatsh.~Math.}~\textbf{92}, 253--268.

\bibitem{FRSS-DE}
\newblock Deitmar A. and Echterhoff S. (2009).
\newblock\emph{Principles of Harmonic Analysis}
\newblock (Springer, New York).

\bibitem{FRSS-Ebe}
Eberlein W.F.\ (1955).
A note on Fourier--Stieltjes transforms,
\textit{Proc.\ Amer.\ Math.\ Soc.} \textbf{6}, 310--312.

\bibitem{FRSS-EV25}
\newblock Enstad U. and Van Velthoven J.T. (2025).
\newblock Coherent systems over appoximate lattices in amenable groups,
\newblock \emph{Ann.~Inst.~Fourier (Grenoble)}~\textbf{75}, 2293--2315.
\texttt{arXiv:2208.05896}.

\bibitem{FRSS-Fei77}
\newblock Feichtinger H.G. (1977).
\newblock Multipliers from $\Lisp(G)$ to a Homogeneous Banach Space,
\newblock \emph{J.~Math.~Anal.~Appl.}~\textbf{61}, 241--356.

\bibitem{FRSS-fe77-3}
\newblock {F}eichtinger H.G. (1977).
\newblock {A} characterization of {W}iener's algebra on locally compact groups,
\newblock \emph{Arch.~Math.~(Basel)}~\textbf{29}, 136--140.

\bibitem{FRSS-Fei79a}
\newblock Feichtinger H.G. (1979).
\newblock Eine neue Segalalgebra mit Anwendungen in der Harmoni\-schen Analyse.
\newblock In \emph{Winterschule 1979, Internationale Arbeitstagung \"uber Topologische Gruppen und Gruppenalgebren}, pp.~23--25 (Univ.~Wien).

\bibitem{FRSS-Fei80}
\newblock Feichtinger H.G. (1980).
\newblock Un espace de Banach de distributions temp\'er\'ees sur les groupes localement compacts ab\'eliens,
\newblock \emph{C.~R.~Acad.~Sci.~Paris~S\'er.~A-B}~\textbf{290}, A791--A794.

\bibitem{FRSS-Fei81b}
\newblock Feichtinger H.G. (1981).
\newblock A characterization of minimal homogeneous Banach spaces,
\newblock \emph{Proc.~Amer.~Math.~Soc.}~\textbf{81}, 55--61.

\bibitem{FRSS-Fei81}
\newblock Feichtinger H.G. (1981).
\newblock On a new Segal algebra,
\newblock \emph{Monatsh.~Math.}~\textbf{92}, 269--289.

\bibitem{FRSS-Fei81c}
\newblock Feichtinger H.G. (1981).
\newblock Banach spaces of distributions of Wiener’s type and interpolation.
\newblock Functional analysis and approximation (Oberwolfach, 1980),
\newblock \emph{Internat.~Ser.~Numer.~Math.}~\textbf{60}, pp.~153--165 (Birkhäuser, Basel).

\bibitem{FRSS-F83}
\newblock Feichtinger H.G. (1983).
\newblock Banach convolution algebras of Wiener type.
Functions, series, operators, Vol. I, II (Budapest, 1980),
\emph{Colloq.~Math.~Soc.~J\'anos~Bolyai}~\textbf{35}, 509--524
\newblock (North-Holland, Amsterdam).

\bibitem{FRSS-HGF87}
\newblock Feichtinger H.G. (1987).
\newblock Banach spaces of distributions defined by decomposition methods. II.
\newblock \emph{Math.~Nachr.}~\textbf{132}, 207--237.

\bibitem{FRSS-fe92}
\newblock Feichtinger H.G. (1992).
\newblock New results on regular and irregular sampling based on Wiener amalgams.
\newblock In \emph{Function Spaces (Edwardsville, IL, 1990)}, Jarosz K. (ed.), pp.~107--121
(Marcel Dekker, New York).

\bibitem{FRSS-Fei22}
\newblock Feichtinger H.G. (2022).
\newblock Homogeneous Banach spaces as Banach convolution modules over $M(G)$,
\newblock \emph{Mathematics}~\textbf{10}, 364, 22 pp.

\bibitem{FRSS-Fei23}
\newblock Feichtinger H.G. (2023).
\newblock Sampling via the Banach Gelfand Triple.
\newblock In \emph{Sampling, Approximation, and Signal Analysis--Harmonic Analysis in the Spirit of J. Rowland Higgins}, Casey S.D., Dodson M.M. Ferreira P.J.S.G. and  Zayed A. (eds.), pp.~211--242,
\newblock \emph{Applied and Numerical Harmonic Analysis} (Birkhäuser/Springer, Cham).

\bibitem{FRSS-Fei24}
\newblock Feichtinger H.G. (2024).
\newblock The ubiquitous appearance of BUPUs.
\newblock In \emph{Harmonic Analysis and Partial Differential Equations}, Feuto J. and Kpata B.A. (eds.), pp.~107--138,
\newblock \emph{Applied and Numerical Harmonic Analysis} (Birkhäuser/Springer, Cham).

\bibitem{FRSS-Fei25}
\newblock Feichtinger H.G. (2025).
\newblock An axiomatic approach to mild distributions.
\newblock \emph{Axioms} \textbf{14}, 302, 25 pp.

\bibitem{FRSS-fegr85}
{F}eichtinger H.G. and {G}r{\"o}bner P. (1985).
\newblock {B}anach spaces of distributions defined by decomposition methods. {I},
\newblock \emph{Math.~Nachr.}~\textbf{123}, 97--120.

\bibitem{FRSS-fegr92-3}
\newblock Feichtinger H.G. and Gr\"ochenig K. (1992).
\newblock Iterative reconstruction of multivariate band-limited functions from irregular sampling values,
\newblock \emph{SIAM~J.~Math.~Anal.}~\textbf{23}, 244--261.

\bibitem{FRSS-Fo99}
\newblock Folland G. (1999).
\newblock \emph{Real Analysis. Modern Techniques and Their Applications},
\newblock  2nd ed. (John Wiley \& Sons, New York).

\bibitem{FRSS-FS85}
\newblock Fournier J.J.F. and Stewart J. (1985).
\newblock Amalgams of $L^p$ and $l^q$,
\newblock \emph{Bull. Amer. Math. Soc.} \textbf{13}, 1--21. 

\bibitem{FRSS-Gra}
\newblock Grafakos L. (2011).
\newblock \emph{Classical Fourier Analysis},
\newblock GTM \textbf{249},
\newblock 3rd ed. (Springer, New York).

\bibitem{FRSS-ARMA}
\newblock Gil~de~Lamadrid J. and Argabright L.N. (1990).
\newblock \emph{Almost Periodic Measures},
\newblock Memoirs Amer. Math. Soc.~\textbf{85}.

\bibitem{FRSS-G01}
\newblock Gr\"ochenig K. (2001).
\newblock \emph{Foundations of Time-Frequency Analysis},
\newblock \emph{Applied and Numerical Harmonic Analysis}
\newblock (Birkhäuser, Boston).

\bibitem{FRSS-HR70}
\newblock Hewitt E. and Ross K.A. (1970).
\newblock \emph{Abstract Harmonic Analysis II}
\newblock (Springer, Berlin).

\bibitem{FRSS-Hof1}
\newblock Hof A. (1995).
\newblock On diffraction by aperiodic structures,
\newblock \emph{Commun.~Math.~Phys.}~\textbf{169}, 25--43.

\bibitem{FRSS-H03}
\newblock H\"ormander L. (2003).
\newblock \emph{The Analysis of Linear Partial Differential Operators. I. Distribution Theory and Fourier Analysis},
\newblock reprint (Springer, Berlin).

\bibitem{FRSS-H89}
\newblock H\"ormann W. (1989).
\newblock \emph{Generalized Stochastic Processes and Wigner Distribution},
\newblock PhD thesis, Univ.~of~Vienna.

\bibitem{FRSS-Jak}
\newblock Jakobsen M.S. (2018).
\newblock On a (no longer) new Segal algebra - a review of the Feichtinger algebra,
\newblock \emph{J.~Fourier.~Anal.~Appl.}~\textbf{24}, 1579--1660. \texttt{arXiv:1608.04566}.

\bibitem{FRSS-Kah}
\newblock Kahane J.P. (1955).
\newblock Sur les fonctions sommes de séries trigonom\'etriques absolument convergentes,
\newblock \emph{C.~R.~Acad.~Sc.~Paris}~\textbf{240},  36--37.

\bibitem{FRSS-KR2015}
\newblock Keller G. and Richard C. (2018).
\newblock Dynamics on the graph of the torus parametrisation,
\newblock \emph{Ergod.~Th.~\& Dynam.~Syst.}~\textbf{38}, 1048--1085. \texttt{arXiv:1511.06137}.

\bibitem{FRSS-Lag99}
\newblock Lagarias J.C. (1999).
\newblock Geometric models for quasicrystals I. Delone sets of finite type,
\newblock \emph{Discrete~Comput.~Geom.}~\textbf{21}, 161--191.

\bibitem{FRSS-L00}
\newblock Lagarias J.C. (2000).
\newblock Mathematical quasicrystals and the problem of diffraction.
\newblock In \emph{Directions in Mathematical Quasicrystals}, Baake M. and Moody R.V. (eds.), pp.~61--93
\newblock (Amer.~Math.~Soc., Providence).

\bibitem{FRSS-LLRSS}
\newblock Lee J.-Y., Lenz D., Richard  C., Sing B. and Strungaru N. (2020).
\newblock Modulated crystals and almost periodic measures,
\newblock \emph{Lett.~Math.~Phys.}~\textbf{110}, 3435--3472. \texttt{arXiv:1907.07017}.

\bibitem{FRSS-Lenz2009}
\newblock Lenz D. (2009).
\newblock Continuity of eigenfunctions of uniquely ergodic dynamical systems and intensity of Bragg peaks,
\newblock \emph{Commun.~Math.~Phys.}~\textbf{287}, 225--258. \texttt{arXiv:math-ph/0608026}.

\bibitem{FRSS-LR}
\newblock Lenz D. and Richard C. (2007).
\newblock Pure point diffraction and cut and project schemes for measures: the smooth case,
\newblock \emph{Math.~Z.}~\textbf{256}, 347--378. \texttt{arXiv:math/0603453}.

\bibitem{FRSS-LRS2024}
\newblock Lenz D., Richard C. and Strungaru N. (2026).
\newblock Which Meyer sets are regular model sets? A characterization via almost periodicity,
\newblock \emph{J.~Func.~Anal.}~\textbf{291}, 111495:1--32. \texttt{arXiv:2410.22536}.

\bibitem{FRSS-LSS2020}
\newblock Lenz D., Spindeler T. and Strungaru N. (2020).
\newblock Pure point diffraction and mean, Besicovitch and Weyl almost periodicity,
\newblock \emph{preprint}, 87 pp. \texttt{arXiv:2006.10821v1}.

\bibitem{FRSS-LSS2020a}
\newblock Lenz D., Spindeler T. and Strungaru N. (2024).
\newblock The (reflected) Eberlein convolution of measures,
\newblock \emph{Indag.~Math.}~\textbf{35}, 959--988. \texttt{arXiv:2211.06969}.

\bibitem{FRSS-LS2019}
\newblock Lenz D. and Strungaru N. (2019).
\newblock On weakly almost periodic measures,
\newblock \emph{Trans.~Am.~Math.~Soc.}~\textbf{371}, 6843--6881. \texttt{arXiv:1609.08219}.

\bibitem{FRSS-LRW74}
\newblock Liu T.S., van Rooij A. and Wang J.K. (1974).
\newblock On some group algebra modules related to the Wiener algebra $M_1$,
\newblock \emph{Pacific~J.~Math.}~\textbf{55}, 507--520.

\bibitem{FRSS-Loo53}
\newblock Loomis L.H. (1953).
\newblock  \emph{An Introduction to Abstract Harmonic Analysis}
\newblock (Van Nostrand, Toronto).

\bibitem{FRSS-Los80}
\newblock  Losert V. (1980).
\newblock  A characterization of the minimal strongly character invariant Segal algebra,
\newblock  \emph{Ann. Inst. Fourier (Grenoble)}~\textbf{30}, 129--139. 

\bibitem{FRSS-M19}
\newblock Matusiak E. (2019).
\newblock Gabor frames for model sets,
\newblock \emph{J.~Fourier~Anal.~Appl.}~\textbf{25}, 2570--2607. \texttt{arXiv:1801.10026}.

\bibitem{FRSS-M22}
\newblock Matusiak E. (2022).
\newblock Frames of translates for model sets,
\newblock \emph{Appl.~Comput.~Harmon.~Anal.}~\textbf{57}, 27--57. \texttt{arXiv:1801.05213}.

\bibitem{FRSS-MRS26}
\newblock Maz\'a\v c J.,  Richard C. and Strungaru N. (2026).
\newblock On almost periodicity in crystalline measures,
\newblock  preprint; \texttt{arXiv:2605.23884}.

\bibitem{FRSS-Mey72}
\newblock Meyer Y. (1972).
\newblock \emph{Algebraic Numbers and Harmonic Analysis}
\newblock  (North-Holland, Amsterdam-London).

\bibitem{FRSS-Mey12}
\newblock Meyer Y. (2012).
\newblock Quasicrystals, almost periodic patterns, mean-periodic functions and irregular sampling,
\newblock \emph{African~Diaspora~J.~Math.}~\textbf{13}, 1--45.

\bibitem{FRSS-Mey23}
\newblock Meyer Y. (2023).
\newblock Curved model sets and crystalline measures.
\newblock In \emph{Theoretical Physics, Wavelets, Analysis, Genomics -- An Indisciplinary Tribute to Alex Grossmann}, Flandrin P. et al (eds.), pp.~389--407,
\newblock \emph{Appl.~Comput.~Harmon.~Anal.} (Birkhäuser/Springer, Cham).

\bibitem{FRSS-MoSt}
\newblock Moody R.V. and Strungaru N. (2017).
\newblock Almost Periodic Measures and their Fourier Transforms.
\newblock In \cite{FRSS-TAO2}, 173--270.

\bibitem{FRSS-O75}
\newblock Osborne M.S. (1975).
\newblock On the Schwartz--Bruhat space and the Paley--Wiener theorem for locally compact abelian groups,
\newblock \emph{J.~Func.~Anal.}~\textbf{19}, 40--49.

\bibitem{FRSS-PRS22}
\newblock Pogorzelski F., Richard C. and Strungaru N. (2022).
\newblock Leptin densities in amenable groups,
\newblock \emph{J.~Fourier~Anal.~Appl.}~\textbf{28}, 85, 36 pp. \texttt{arXiv:2107.04760}.

\bibitem{FRSS-P80}
\newblock Poguntke D. (1980).
\newblock Gewisse Segalsche Algebren auf lokalkompakten Gruppen,
\newblock \emph{Arch. Math.}~\textbf{33}, 454--460.

\bibitem{FRSS-Rei1}
\newblock Reiter H. (1968).
\newblock \emph{Classical Harmonic Analysis and Locally Compact Groups}
\newblock (Clarendon Press, Oxford).

\bibitem{FRSS-Rei3}
\newblock Reiter H. (1971).
\newblock \emph{$\Lisp$-algebras and Segal Algebras},
\newblock LNM \textbf{231} (Springer, Berlin).

\bibitem{FRSS-Rei4}
\newblock Reiter H. (1989).
\newblock \emph{Metaplectic Groups and Segal Algebras},
\newblock LNM \textbf{1382} (Springer, Berlin).

\bibitem{FRSS-Rei5}
\newblock Reiter H. (1993).
\newblock On the Siegel-Weil formula,
\newblock \emph{Monatsh.~Math.}~\textbf{116}, 299--330.

\bibitem{FRSS-Rei2}
\newblock Reiter H. and Stegeman J.D. (2000).
\newblock \emph{Classical Harmonic Analysis and Locally Compact Groups}, 2nd ed.
(Clarendon Press, New York).

\bibitem{FRSS-R23}
\newblock Richard C. (2023).
\newblock Mild distributions in diffraction theory. In
\newblock \emph{Report on the workshop Aspects of Aperiodic order}, MFO Report \textbf{37/23}, pp.~26--28.

\bibitem{FRSS-RiSchu20}
\newblock Richard C. and Schumacher C. (2020).
\newblock On sampling and interpolation by model sets,
\newblock \emph{J.~Fourier~Anal.~Appl.}~\textbf{26}, 39, 37 pp. \texttt{arXiv:1809.10750}.

\bibitem{FRSS-CRS}
\newblock Richard C. and Strungaru N. (2017).
\newblock Pure point diffraction and Poisson summation,
\newblock \emph{Ann.~H.~Poincar\'e}~\textbf{18}, 3903--3931. \texttt{arXiv:1512.00912}.

\bibitem{FRSS-RS17}
\newblock Richard C. and Strungaru N. (2017).
\newblock A short guide to pure point diffraction in cut-and-project sets,
\newblock \emph{J.~Phys.~A: Math.~Theor.}~\textbf{50}, 154003, 25 pp. \texttt{arXiv:1606.08831}.

\bibitem{FRSS-RT84}
\newblock Robertson A.P. and Thornett M.L. (1984).
\newblock On translation-bounded measures,
\newblock \emph{J.~Austral.~Math.~Soc.}~\textbf{A37}, 139--142.

\bibitem{FRSS-RudArt}
\newblock Rudin W. (1955).
\newblock Non-analytic functions of absolutely convergent Fourier series,
\newblock \emph{Proc.~Nat.~Acad.~Sc.~U.S.A.}~\textbf{41}, 238--240.

\bibitem{FRSS-RUD2}
\newblock Rudin W. (1987).
\newblock \emph{Real and Complex Analysis}, 3rd ed.
\newblock (McGraw--Hill, Singapore).

\bibitem{FRSS-RUD3}
\newblock Rudin W. (1991).
\newblock \emph{Functional Analysis}, 2nd ed.
\newblock (McGraw--Hill, Singapore).

\bibitem{FRSS-R02}
\newblock Ryan R.A. (2002).
\newblock \emph{Introduction to Tensor Products of Banach Spaces}
\newblock (Springer, London).

\bibitem{FRSS-S71}
\newblock Schaefer H.H. (1971).
\newblock \emph{Topological Vector Spaces}
\newblock (Springer, New York).

\bibitem{FRSS-Martin2}
\newblock Schlottmann M. (2000).
\newblock Generalised model sets and dynamical systems.
\newblock In \emph{Directions in Mathematical Quasicrystals}, Baake M. and Moody R.V. (eds.), pp.143--159
(Amer.~Math.~Soc., Providence).

\bibitem{FRSS-Sch74}
\newblock Schwartz L. (1974).
\newblock \emph{Radon Measures on Arbitrary Topological Spaces and Cylindrical Measures}
\newblock (Oxford University Press, Oxford).

\bibitem{FRSS-St17}
\newblock Strungaru N. (2017).
\newblock Almost periodic pure point measures,
\newblock In  \cite{FRSS-TAO2}, pp.~271--342.

\bibitem{FRSS-S19}
\newblock Strungaru N. (2020).
\newblock On the Fourier transformability of strongly almost periodic measures,
\newblock \emph{Canad.~J.~Math.}~\textbf{72}, 900--927. \texttt{arXiv:1704.04778}.

\bibitem{FRSS-ST16}
\newblock Strungaru N. and Terauds V. (2016).
\newblock Diffraction theory and almost periodic distributions,
\newblock \emph{J.~Stat.~Phys.} \textbf{164}, 1183--1216. \texttt{arXiv:1603.04796}.

\bibitem{FRSS-T04}
\newblock Toft J. (2004).
\newblock Continuity properties for modulation spaces, with applications to pseudo-differential calculus. I,
\newblock \emph{J. Funct. Anal.} \textbf{207}, 399--429. 

\bibitem{FRSS-T67}
\newblock Tr\'eves F. (1967).
\newblock \emph{Topological Vector Spaces, Distributions and Kernels}
\newblock (Academic Press, New York).

\bibitem{FRSS-W30}
\newblock Wiener N. (1930).
\newblock Generalized harmonic analysis,
\newblock \emph{Acta~Math.} \textbf{55}, 117--258.

\bibitem{FRSS-W32}
\newblock Wiener N. (1932).
\newblock Tauberian theorems,
\newblock \emph{Ann.~of~Math.} \textbf{33}, 1--100.

\bibitem{FRSS-W59}
\newblock Wiener N. (1959).
\newblock \emph{The Fourier Integral and Certain of its Applications}
\newblock (Dover Publications, New York).

\bibitem{FRSS-WW38}
\newblock Wiener N. and Wintner A. (1938).
\newblock On singular distributions,
\newblock \emph{J.~Math.~and~Phys.} \textbf{17}, 233--246.

\bibitem{FRSS-WW41}
\newblock Wiener N. and Wintner A. (1941).
\newblock On the ergodic dynamics of almost periodic systems,
\newblock \emph{Amer.~J.~Math.} \textbf{63}, 794--824.

\end{thebibliography}
